\documentclass{article}
\usepackage[utf8]{inputenc}

\usepackage{cancel}
\usepackage[textsize=tiny, backgroundcolor=white, linecolor=black]{todonotes}

\usepackage{pdfpages}

\usepackage{xcolor}
\usepackage[margin=1.5in]{geometry}
\usepackage{blindtext}
\usepackage{parskip}

\usepackage{tikz}

\usepackage{amsmath}
\usepackage{amsfonts}
\usepackage{amssymb}
\usepackage{amsthm}

\usepackage{dsfont}
\usepackage{stmaryrd}
\usepackage{relsize}
\usepackage{enumitem}

\usepackage{comment}

\usepackage{hyperref}
\hypersetup{
	unicode,
	pdfauthor={Sanjay Ramassamy, Benjamin Terlat},
}

\theoremstyle{plain}
\newtheorem{theorem}{Theorem}[section]
\newtheorem{corollary}[theorem]{Corollary}
\newtheorem{lemma}[theorem]{Lemma}

\newtheorem{proposition}[theorem]{Proposition}

\newtheorem{conjecture}[theorem]{Conjecture}
\newtheorem{question}[theorem]{Open question}

\theoremstyle{definition}
\newtheorem{definition}[theorem]{Definition}
\newtheorem{example}[theorem]{Example}
\newtheorem{remark}[theorem]{Remark}

\numberwithin{theorem}{section}

\usepackage{graphicx, color}
\graphicspath{{fig/}}

\newcommand{\ZZ}{\mathbb{Z}}
\newcommand{\RR}{\mathbb{R}}

\DeclareMathOperator{\Gr}{\mathrm{Gr}}
\DeclareMathOperator{\DC}{\mathrm{DC}}

\DeclareMathOperator{\proj}{\mathrm{proj}}

\newcommand\Set[2]{\{\,#1\mid #2\,\}}

\newcommand\lrSet[2]{\left\{\, #1 \mid #2 \, \right\}}

\newcommand\binDyn[2]{ \binDynLetter(#1 , #2) }
\newcommand\carDyn[2]{ \carDynLetter(#1 , #2) }
\newcommand\binDynLetter{ \Phi }
\newcommand\carDynLetter{ \Psi }

\newcommand\binConf{ \mathcal{A} }
\newcommand\carConf{ \mathcal{C} }
\newcommand\cursorLetter{ c }
\newcommand\cursorLetterBin{ \cursorLetter }
\newcommand\cursorLetterCar{ \widehat\cursorLetter }

\newcommand{\front}[1]{f(#1)}
\newcommand{\nbParam}{N}
\newcommand{\nextJumpTime}{\tau}
\newcommand{\nextJumpTimeBin}{\nextJumpTime}
\newcommand{\nextJumpTimeCar}{\widehat\nextJumpTime}
\newcommand{\indexOfSpeedCar}{\widehat{i}}
\newcommand\setFixedSpeed[1]{\{z_1 = #1\}}

\newcommand{\Addresses}{{
  \bigskip
  \footnotesize

  \textsc{Université Paris-Saclay, CNRS, CEA, Institut de physique théorique, 91191 Gif-sur-Yvette, France}\par\nopagebreak
  \textit{E-mail address}: \texttt{sanjay.ramassamy at ipht.fr}

 \medskip
 \textsc{Université Paris-Saclay, CNRS, Laboratoire de mathématiques d'Orsay, 91405 Orsay, France}\par\nopagebreak
 \textsc{Université Paris-Saclay, CNRS, CEA, Institut de physique théorique, 91191 Gif-sur-Yvette, France}\par\nopagebreak
  \textit{E-mail address}: \texttt{benjamin.terlat at gmail.com}

}}

\title{Wall-crossing phenomenon for the liquid bin model}
\author{Sanjay Ramassamy \and Benjamin Terlat}
\date{\today}

\begin{document}

\maketitle
\begin{abstract}
We introduce the liquid bin model as a continuous-time deterministic dynamics, arising as the hydrodynamic limit of a discrete-time stochastic interacting particle system called the infinite bin model. For the liquid bin model, we prove the existence and uniqueness of a stationary evolution, to which the dynamics converges exponentially fast. The speed of the front of the system is explicitly computed as a continuous piecewise rational function of the parameters of the model, revealing an underlying wall-crossing phenomenon. We show that the regions on which the speed is rational are of non-empty interior and are naturally indexed by Dyck paths. We provide a complete description of the adjacency structure of these regions, which generalizes the Stanley lattice for Dyck paths. Finally we point out an intriguing connection to the topic of extensions of partial cyclic orders to total cyclic orders.
\end{abstract}

\section{Introduction}
\label{sec:intro}

The infinite bin model is an interacting particle system corresponding to a rank-biased discrete-time branching random walk. Introduced in 2003 by Foss and Konstantopoulos \cite{FossKonstantopoulos}, it has since then been the subject of multiple studies, see the recent survey \cite{FKMR}. In this article, we study properties of a continuous-time deterministic dynamics which arises as a certain hydrodynamic limit of the infinite bin model. We call this deterministic dynamics the liquid bin model.

\subsection{The infinite bin model and the liquid bin model}
\label{subsec:IBMLBM}

The state space of the \emph{infinite bin model} consists of infinitely many bins indexed by $\ZZ$, with each of them containing a finite number of particles. We require that there exists a non-empty bin such that all the bins to its right are empty. Such a bin is called the \emph{front} of the system. The infinite bin model is parameterized by a probability measure $\mu$ on $\ZZ_{>0}$ and an initial configuration at time $n=0$. At each time step $n\geq1$, we add a particle in the bin immediately to the right of the bin containing the $\xi_n$-th rightmost particle, where the $(\xi_n)_{n\geq1}$ are i.i.d. distributed like $\mu$. It does not matter how ties are broken within a bin to determine which particle is the $\xi_n$-th rightmost one. 
One may show using sub-additivity that the front moves to the right at a linear speed $v_\mu\in(0,1]$, a constant depending only on $\mu$ and not on the starting configuration \cite{FossKonstantopoulos,MalleinRamassamy1}. Denoting by $f_n$ the position of the front at time $n$, the speed $v_\mu$ is defined as the almost sure limit of $\tfrac{f_n}{n}$ as $n$ goes to infinity.

Special cases of interest for $\mu$ include:
\begin{itemize}
\item Geometric distributions: here the infinite bin model can be coupled to Barak-Erd\H{o}s random graphs and $v_\mu$ gives the linear growth rate of the length of the longest path in such graphs \cite{FossKonstantopoulos} ;
\item The uniform distribution on $\llbracket1,k\rrbracket$ for some $k\in \ZZ_{>0}$: in that case the infinite bin model can be coupled with a branching random walk with selection of the rightmost $k$ individuals \cite{AldousPitman,MalleinRamassamy1}.
\end{itemize}

An important question for the infinite bin model is to compute $v_\mu$ for a given $\mu$. In general the answer to this question is complicated. When $\mu$ has a finite support bounded by $K\geq1$, the infinite bin model may be reduced to a Markov chain with a finite state space. In that case, $v_\mu$ is a rational function of the probabilities of the integers in the support of $\mu$. However the number of states of the reduced Markov chain is exponential in $K$, hence the formulas quickly become very intricate as $K$ grows, even if we restrict ourselves to measures supported by just two integers.

Such infinite bin models with a finitely supported $\mu$ were crucially used in \cite{MalleinRamassamy1}, in particular to derive explicit estimates on the linear growth rate of the length of the longest path in Barak-Erd\H{o}s graphs. To have an idea of the potential complexity of $v_\mu$ for an infinite bin model where $\mu$ is supported by $\llbracket1,3\rrbracket$, look at either the upper or the lower bound appearing in the formula in display that is located between formula (5.2) and Figure 5 in \cite{MalleinRamassamy1}. Infinite bin models with $\mu$ supported by two integers were also considered in \cite{ChernyshRamassamy}, to prove that two runs of the Markov chain started at different states would eventually couple almost surely.

We introduce in this article the \emph{liquid bin model} as a continuous-time deterministic dynamics of liquid in bins: liquid gets added according to some rules that arise as the hydrodynamic limits of the rules ``Add a particle to the right of the $k$th particle with probability $\mu(k)$". In this setting, one can explicitly compute the speed of the front and some nice combinatorial structures emerge.

The parameters of the model are:
\begin{itemize}
\item an integer $N\geq1$ corresponding to the number of rules;
\item positive real numbers $a_1<\cdots<a_N$ describing the locations where liquid is added;
\item positive real numbers $p_1,\ldots,p_N$ describing the rates at which liquid is added.
\end{itemize}

The state space of the liquid bin model consists of infinitely many bins indexed by $\mathbb Z$, with each of them containing a finite volume of liquid. We again require the existence of a front bin, that is, a non-empty bin such that all the bins to its right are empty.

For every $i\in\llbracket 1,N\rrbracket$, we place the $i$-th \emph{cursor} in some bin of index $c_i\in\mathbb Z$, in such a way that below it in bin $c_i$ and in all the bins to its right, the total volume of liquid is equal to $a_i$. This property uniquely defines $c_i$. The dynamics consists in adding, for every $i\in\llbracket 1,N\rrbracket$, liquid at a rate $p_i$ in the bin of index $c_i+1$, which is immediately to the right of the bin containing the $i$-th cursor. As liquid gets added, the cursors move down inside each bin. More precisely, the $i$-th cursor moves down at rate $\sum_{j:c_j(x)\ge c_i(x)}p_j$. Once a cursor reaches the bottom of a bin, it jumps to the top of the next bin to its right. As a consequence, the locations where liquid gets added evolve with time. See Figure \ref{fig:LMBdyn} for an example.

\begin{figure}
    \centering
    \includegraphics[scale=0.73]{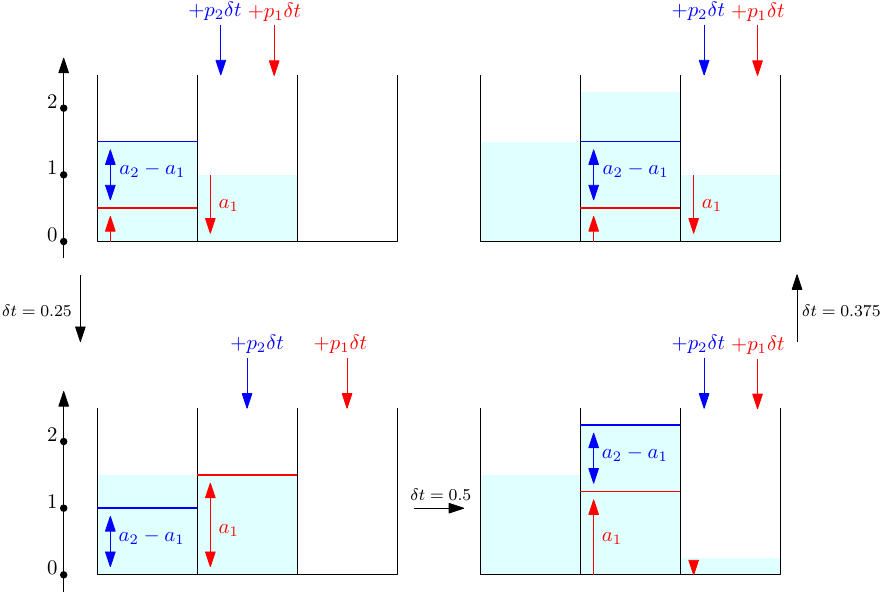}
    \caption{Illustration of the dynamics of the liquid bin model with parameters $N=2$, $a_1 = 1.5$, $p_1 =0.5$, $a_2 = 2.5$, $p_2 =1.5$. Cursor $1$ (respectively $2$) is represented in red (resp. blue): under it and to its right, there is always $a_1$ (resp. $a_2$) quantity of liquid. After a time $0.25$, starting from the configuration in the top-left, cursor $1$ goes from bin $1$ to bin $2$, yielding the configuration in the bottom-left. After an additional time $0.5$, cursor $2$ goes from bin $1$ to bin $2$ to obtain the configuration in the bottom-right. After waiting for an extra time $0.375$, the configuration in the top-right is reached. These configurations pertain to a stationary evolution with period $T=1.125$.}
    \label{fig:LMBdyn}
\end{figure}

The liquid bin model arises as a hydrodynamic limit of the infinite bin model in the following sense. Assume that $p_1 + \dots + p_{\nbParam} = 1$, and for every $s\in\mathbb R_{>0}$ such that $s\geq1/a_1$ define the probability measure $\mu^{(s)} := \sum_{i = 1 }^{\nbParam} p_i \delta_{ \lfloor s \cdot a_i \rfloor}$ on $\ZZ_{>0}$. For every $s\geq1/a_1$, let $X^{(s)} = (X^{(s)}(n))_{n \geq0}$ be the infinite bin model with some initial configuration $X^{(s)}(0)$ and with move distribution $\mu^{(s)}$. It was shown in \cite{Terlat} that, as $s$ goes to infinity, if the rescaled initial configurations $\tfrac{X^{(s)}(0)}{s}$ converge to some configuration of liquid $x(0)$, then the rescaled infinite bin models $\tfrac{X^{(s)}(\lfloor st\rfloor)}{s}$ converge in distribution to the liquid bin model with initial configuration $x(0)$ and parameters $a_1,\ldots,a_N,p_1,\ldots,p_N$. The convergence holds for the sup norm for $t$ pertaining to any compact interval. An important missing property in the study of this hydrodynamic limit is the following.

\begin{conjecture}
As $s$ goes to infinity, the speed of the rescaled infinite bin model $X^{(s)}(\lfloor st\rfloor)$ converges to the speed of the liquid bin model with parameters $a_1,\ldots,a_N,p_1,\ldots,p_N$. 
\end{conjecture}

One way to prove this conjecture would be to obtain a control of $\tfrac{X^{(s)}(\lfloor st\rfloor)}{s}-x(t)$ that is uniform in $t\in\RR_{>0}$, not just for $t$ in a compact interval. This would make it possible to interchange the limits $s\to\infty$ and $t\to\infty$. We have not been able to do so.

\subsection{Properties of the liquid bin model}
\label{subsec:LBMproperties}

Throughout the article, we will fix $N\in\ZZ_{>0}$ and denote $N$-tuples using the underline notation. For example, the $N$-tuple $(a_1,\ldots,a_N)$ will be denoted by $\underline a$. Using this notation, $(\underline a,\underline p)$ will denote the $2N$-tuple $(a_1,\ldots,a_N,p_1,\ldots,p_N)$. Let $P^{2N}$ denote the set of all the $2N$-tuples $(\underline a,\underline p)$ of positive real numbers such that $a_i<a_{i+1}$ for all $1\leq i \leq N-1$.

From the point of view of the dynamics of the liquid bin model, we care only about the liquid lying below or to the right of the $N$-th cursor. It corresponds to the first $a_N$ units of liquid, counted from right to left, and within a bin from bottom to top. We will thus consider that, for a fixed value of parameters $(\underline a,\underline p)\in P^{2N}$, two configurations are equal if their first $a_N$ units of liquid are in the same position. A configuration $x(0)$ of liquid in bins is called a \emph{stationary configuration} if there exists some time $T>0$ such that, running the liquid bin model starting from $x(0)$, one obtains at time $T$ a configuration $x(T)$ satisfying the following property: the first $a_N$ units of liquid of $x(T)$ are positioned like the first $a_N$ units of liquid of $x(0)$, up to a shift by one bin to the right. See Figure~\ref{fig:LMBdyn} for an example when $N=2$. A \emph{stationary evolution} for the liquid bin model is a map $x:t\in\RR_{\geq0}\mapsto x(t)$ of configurations evolving like the liquid bin model, such that $x(t)$ is a stationary configuration for every $t\in\RR_{\geq0}$. A stationary liquid bin model may be regarded as a traveling wave. The first main result that we will prove in this article is the following.

\begin{theorem}
\label{theorem:mainstatio}
For any $(\underline a,\underline p)\in P^{2N}$, there exists a unique stationary evolution $\tilde{x}_\infty$ for the liquid bin model. Moreover, for any choice of an initial configuration $x(0)$, the liquid bin model $x(t)$ converges exponentially fast to the stationary configuration $\tilde{x}_\infty(t)$.
\end{theorem}

For any $k\in\ZZ$, denote by $I_k$ the time interval during which the front bin of $x(t)$ is at position $k$. The aforementioned convergence is exponential in the index $k$ of these time intervals. The statements about uniqueness and convergence are up to a shift in time and refer to the distribution of the first $a_N$ units of liquid, as explained above. We stress that we do not prove the convergence of the whole configuration to the stationary configuration. A precise reformulation of Theorem~\ref{theorem:mainstatio} will be proved in Section~\ref{sec:stationary} in the guise of  Theorem~\ref{theorem:convergenceOfDynamics}, stated in the language of a car model equivalent to the bin model. The derivation of Theorem~\ref{theorem:mainstatio} from Theorem~\ref{theorem:convergenceOfDynamics} is given right after the proof of Theorem~\ref{theorem:mainstatio}. In the course of that derivation, we provide the precise formulation of the convergence statement in Theorem~\ref{theorem:convergenceOfDynamics} in the form of statements A and B.

After a finite amount of time which depends on the starting configuration, the first $a_N$ units of liquid are located in at most $\lceil \tfrac {a_N}{a_1}\rceil +1$ bins, thus the configurations live in a finite-dimensional space. The exponential rate of convergence has an upper bound of $1-\tfrac{q_1}{q_N}$, which is independent of the starting configuration.

Thanks to Theorem~\ref{theorem:mainstatio}, the computation of the speed of the front of the liquid bin model for any initial configuration and any choice of parameters $(\underline a,\underline p)\in P^{2N}$ simply boils down to computing this speed for the stationary liquid bin model for these parameters. We explain below how to explicitly compute this speed. The idea is to partition the parameter space $P^{2N}$ into regions $P_G$ labelled by some directed graphs $G$, and to associate to each such graph $G$ a rational function of the parameters $(\underline a,\underline p)$, that will be the speed of the front when $(\underline a,\underline p)$ lies in the region $P_G$.

Let $E_N$ be the collection of all the pairs of integers $(i,j)$ such that $1\leq i<j\leq N$. We will consider directed graphs with the vertex set $\llbracket 1, N\rrbracket$ and with an edge set that is a subset of $E_N$. We will represent such graphs in the plane by drawing each vertex $1\leq i\leq N$ of the graph at the point of coordinates $(2i-1,1)$ and drawing each edge $(i,j)$ of the graph as the broken line connecting $(2i-1,1)$ to $(2j-1,1)$ via the point $(i+j-1,j-i+1)$. See the left picture of Figure~\ref{fig:DCgraph-Dyckpath} for an example. The set $E_N$ possesses a partial order $\prec_E$, where $(i',j')\prec_E (i,j)$ if and only if $1\leq i\leq i'<j'\leq j$. This partial order corresponds to a nesting relationship for edges of a graph: with the above representation convention, $(i',j')\prec_E (i,j)$ whenever the edge $(i',j')$ is drawn below the edge $(i,j)$. We say that $(i',j')$ is \emph{nested} in $(i,j)$. For example, for the directed graph represented in the left picture of Figure~\ref{fig:DCgraph-Dyckpath}, the edge $(1,2)$ is nested in the edge $(1,3)$.

\begin{figure}
    \centering
    \includegraphics[scale=0.73]{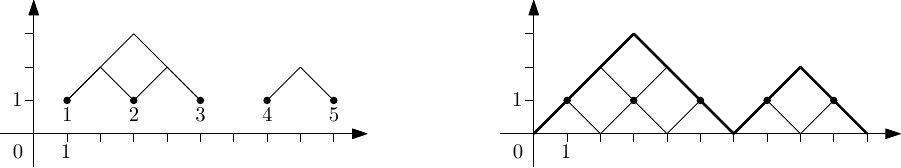}
    \caption{On the left, the downward closed graph on vertex set $\llbracket1,5\rrbracket$ with edges $(1,2)$, $(1,3)$, $(2,3)$ and $(4,5)$. Each edge $(i,j)$ is such that $i<j$ and is directed from $i$ to $j$. We omit the depiction of the direction to avoid overloading the picture. On the right, the same DC graph, where we added an extra broken line for each vertex. The supremum of all the lines present in the picture is indicated in bold. It corresponds to the Dyck path of length $10$ $+++---++--$, where $+$ (resp. $-$) denotes an ``up'' (resp. ``down'') vector.}
    \label{fig:DCgraph-Dyckpath}
\end{figure}

We associate to every choice of parameters $(\underline a,\underline p)\in P^{2N}$ the graph $\Gr(\underline a,\underline p)$ with vertex set $\llbracket 1,N\rrbracket$ and with edge set constructed as follows. For any $(i,j)\in E_N$, the pair $(i,j)$ is an edge of $\Gr(\underline a,\underline p)$ if and only if there exists a time $t\geq0$ at which the $i$-th cursor and the $j$-th cursor are both in the same bin of the stationary configuration $x_\infty(t)$. For the example shown on Figure~\ref{fig:LMBdyn}, the graph $\Gr(\underline a,\underline p)$ has vertex set $\{1,2\}$ and has a single edge $(1,2)$, since there is a time in the stationary evolution when both cursors are in the same bin. For general $(\underline a,\underline p)$, the graph $\Gr(\underline a,\underline p)$ has the following property: for any pair $(i',j')\prec_E (i,j)$ of nested elements of $E_N$, if $(i,j)$ is an edge of $\Gr(\underline a,\underline p)$, then so is $(i',j')$. Indeed, at a time when cursors $i$ and $j$ are in the same bin, then cursors $i'$ and $j'$ will be sandwiched between them, hence will also be in the same bin. In the language of partial orders, it means that the edge set of the graph is downward closed for the partial order $\prec_E$. We call graphs with such a property \emph{downward closed graphs} (or \emph{DC graphs} for short). We denote by $\DC_N$ the set of all DC graphs with vertex set $\llbracket 1,N\rrbracket$. For every $G\in\DC_N$, we denote by $P_G$ the collection of all $(\underline a,\underline p)\in P^{2N}$ such that $\Gr(\underline a,\underline p)=G$.

We may further enrich the broken line representation of a DC graph by replacing each vertex $(2i-1,1)$ for $i\in\llbracket 1,N\rrbracket$ by the broken line connecting $(2i-2,0)$ to $(2i,0)$ via the point $(2i-1,1)$. Then the supremum of all the broken lines corresponding to either edges or vertices defines a \emph{Dyck path} of length $2N$, namely a concatenation of $N$ ``up'' vectors $(1,1)$ and $N$ ``down'' vectors $(1,-1)$ that realizes an excursion above height $0$ from $(0,0)$ to $(2N,0)$. See the right picture of Figure~\ref{fig:DCgraph-Dyckpath} for an example. It is not hard to see that this provides a bijection between $\DC_N$ and the set of Dyck paths of length $2N$, which is known to be enumerated by the Catalan numbers $C_N:=\tfrac{1}{N+1}\binom{2N}{N}$  \cite{Stanley}.

Let us now associate a rational function of the parameters $(\underline a,\underline p)$ to every DC graph $G\in\DC_N$. For every $i\in\llbracket 1,N\rrbracket$, define the auxiliary variables $d_i:=a_i-a_{i-1}$ and $q_i:=p_1+\ldots+p_i$, with the convention that $a_0=0$. For every $G\in\DC_N$, denote by $b_G(i)$ the greatest vertex $j>i$ such that $(i,j)$ is an edge of $G$. If there is no vertex $j>i$ linked to $i$ in $G$, set $b_G(i):=i$. We adopt the convention that $b_G(0)=1$. For every edge $(i,j)$ of $G$, define its weight to be 
\begin{align}\label{eq:weightEdges}
    \gamma^{(G)}_{i,j} := \frac{q_{b_G(i)}-q_{\max(j-1,b_G(i-1))}}{q_{b_G(i-1)}}\geq0 .
\end{align} 
For every $1 \leq i < j \leq \nbParam$, denote by $\mathcal P^{(G)}_{i,j}$ the set of directed paths from $i$ to $j$ in $G$ and define
\begin{equation}
\label{eq:weightPaths}
\Gamma^{(G)}_{i,j} :=  \sum_{\pi\in\mathcal P^{(G)}_{i,j}} \prod_{(i',j') \in \pi }\gamma^{(G)}_{i',j'}, 
\end{equation}
Set also $\Gamma^{(G)}_{i,i} := 1$ for all $i \in \llbracket 1 , \nbParam \rrbracket$. Then we can give the following formula for the speed of the front of the stationary liquid bin model, namely the inverse of the time it takes for the first $a_N$ units of liquid of a stationary configuration to move exactly one bin to the right.

\begin{theorem}
\label{theorem:mainspeed}
Let $G$ be a DC graph and let $(\underline a,\underline p)\in P_G$. The speed of the front of the stationary liquid bin model with parameters $(\underline a,\underline p)$ is equal to
\begin{equation}
\label{eq:mainspeed}
\frac{ 1 + \sum_{j=1}^N \Gamma^{(G)}_{1,j} \frac{  q_{b_G(j)}-q_{b_G(j-1)}}{q_{b_G(j-1)}}}{\sum_{j=1}^N \Gamma^{(G)}_{1,j} \frac{ d_j }{q_{b_G(j-1)}}}.
\end{equation}
\end{theorem}

Since the only paths in $G$ involved in \eqref{eq:mainspeed} are paths starting at vertex $1$, the speed formula in the region $P_G$ only depends on the connected component of vertex $1$ in $G$. Hence the number of possible speed formulas is equal to the number of connected DC graphs with at most $N$ vertices, namely $C_0+C_1+\cdots+C_{N-1}$. Indeed the bijection between DC graphs and Dyck paths maps the connected component of vertex $1$ to the first excursion above $0$ (the portion of the path until the first return to $0$). Removing the initial up step and the final down step from this excursion yields for every $N'\in\llbracket 1,N\rrbracket$ a bijection between connected graphs in $\DC_{N'}$ and Dyck paths of length $2N'-2$.

Theorem~\ref{theorem:mainspeed} will follow from the more general Theorem~\ref{theorem:formulaSpeed}. Even though the speed is only a piecewise rational function of the parameters $(\underline a,\underline p)\in P^{2N}$, we prove in Proposition \ref{proposition:continuityInParameters} that it is a continuous function on $P^{2N}$. Furthermore, we show in Theorem \ref{theorem:nonemptiness} that each region $P_G$ has a non-empty interior. Can we say more about the topology of each $P_G$ ?
\begin{question}
Let $G\in\DC_N$ be a DC graph. Is the region $P_G\subset P^{2N}$ connected ? Is it simply connected ?
\end{question}

We can give a precise combinatorial description of the adjacency structure of the regions $P_G$. We denote the boundary of a region $P_G$ by $\partial P_G$. If $(X,\prec)$ is a poset, recall that a subset $X'\subset X$ is a called an \emph{antichain} if no two distinct elements of $X'$ are comparable for the partial order $\prec$. Recall also that, if $A$ and $A'$ are two sets, $A\Delta A'$ denotes the symmetric difference of these two sets.

\begin{theorem}
\label{theorem:mainadjacency}
Let $G_1$ and $G_2$ be two distinct DC graphs in $\DC_N$. Then $\partial P_{G_1} \cap \partial P_{G_2}$ is non-empty if and only if $E(G_1) \Delta E(G_2)$ is an antichain for the poset $(E_N,\preceq_E)$. In this case, the codimension of $\partial P_{G_1} \cap \partial P_{G_2}$ is $ |E(G_1)\Delta E(G_2)| $.
\end{theorem}

\begin{figure}
    \centering
    \includegraphics[scale=0.8]{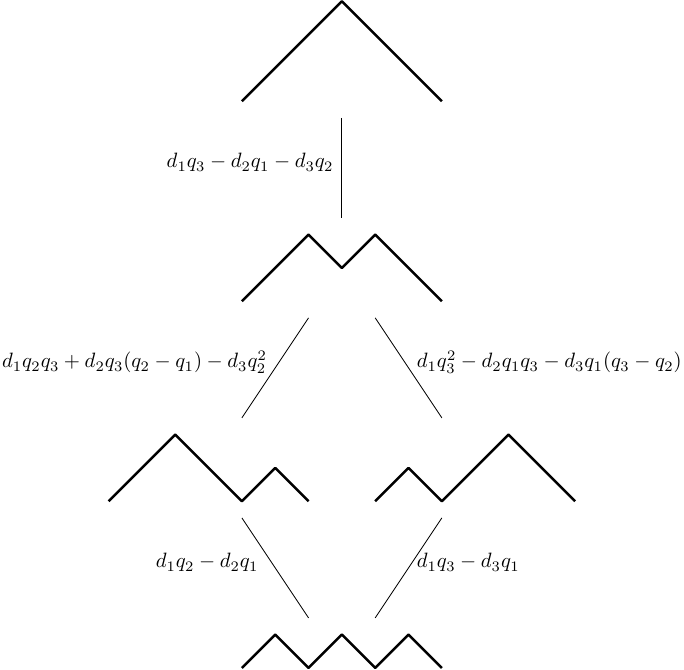}
    \caption{The Hasse diagram of the Stanley lattice for Dyck paths of length $6$. Each Dyck path stands at a node of the diagram. One Dyck path $P_1$ covers another Dyck path $P_2$ whenever $P_1$ lies above $P_2$ in the diagram and they are connected by an edge. Next to each edge is given a polynomial in the variables $(\underline d,\underline q)$. This polynomial vanishes on the boundary between the regions labeled by the two DC graphs that are in bijection with the two Dyck paths at each end of the edge. This polynomial is positive (resp. negative) on the region labeled by the graph corresponding to the Dyck path which is on the top (resp. bottom) end of the edge.}
    \label{fig:Hasse3}
\end{figure}

The notion of dimension that we use here is the Lebesgue covering dimension, also known as the topological dimension. The sets for which we compute the dimension are subsets of $P^{2N}$ and they possess the induced topology of $P^{2N}$.

We will reformulate Theorem~\ref{theorem:mainadjacency} in Proposition~\ref{prop:equivalentcharac} then prove it as Theorem~\ref{theorem:boundaryAnalysis}. The set $\DC_N$ possesses a natural partial order $\prec_G$, whereby $G_1\prec_G  G_2$ if the edge sets $E(G_1)$ and $E(G_2)$ of the DC graphs $G_1$ and $G_2$ satisfy $E(G_1)\subset E(G_2)$. The bijection with Dyck paths transports this partial order to the well-known Stanley lattice for Dyck paths, corresponding to the property that one Dyck path lies below another one. The name ``Stanley lattice'' comes from \cite{BernardiBonichon}, see also \cite{Woodcock}. See Figure~\ref{fig:Hasse3} for a picture of the Stanley lattice when $N=3$. Recall that an element $x$ is said to \emph{cover} an element $y$ in some poset $(X,\leq)$ if $y<x$ and if there exists no element $z\in X$ such that $y<z<x$. By Theorem~\ref{theorem:mainadjacency}, two graphs $G_1$ and $G_2$ have a common boundary of codimension $1$ if and only if one covers the other in the Stanley lattice. The adjacency structure of the regions $P_G$ thus generalizes the Stanley lattice for Dyck paths. Moreover, we provide in Proposition~\ref{proposition:caracPG} the explicit inequalities cutting out each region. See Figure~\ref{fig:Hasse3} for the formulas when $N=3$.

The fact that the speed of the front of the liquid bin model is a piecewise rational function on the parameter space is called a \emph{wall-crossing phenomenon}. The space is divided into regions by hypersurfaces called ``walls" and the speed formulas are different inside different regions. In the present case, each region is indexed by a Dyck path and each wall corresponds to a covering relation in the Stanley lattice. Similar wall-crossing phenomena where the regions are labeled by combinatorial objects have appeared in other fields. In enumerative geometry, double Hurwitz numbers were shown to be piecewise polynomial functions of the parts of the partitions which index them \cite{GJV,SSV,CJM10,CJM11,Johnson}. In mathematical physics, correlation functions for the quantum symmetric simple exclusion process were shown to be piecewise polynomial, with regions of polynomiality indexed by cyclic permutations \cite{BernardJin,Biane}.
 We point out that a phenomenon with a similar flavor appears in \cite{HMSS}, where the authors count equivalence classes of periodic stationary traveling wave solutions to the lattice Nagumo equation and label such equivalence classes by combinatorial objects that are words.
 
 In the case of the liquid bin model, another way to interpret this wall-crossing phenomenon is to see it as a phase transition for an out-of-equilibrium system. Let us illustrate this with the case of $N=2$.
 
 \begin{example}
 \label{ex:N=2}
 Let $N=2$. Up to rescaling space and time, one may assume that $a_2=1$ and that $p_1+p_2=1$. Then there are two free parameters left, $a_1$ and $p_1$, both in $(0,1)$. Applying Theorem \ref{theorem:formulaSpeed} and Proposition~\ref{proposition:caracPG}, we obtain the following explicit speed computation. If $p_1<\tfrac{a_1}{1-a_1}$, then the speed of the liquid bin model is given by $\frac{1}{1-p_1(1-a_1)}$, otherwise it is given by $\tfrac{p_1}{a_1}$. In particular, when $a_1\geq\tfrac{1}{2}$, the speed is given by $\frac{1}{1-p_1(1-a_1)}$ for every value of $p_1\in(0,1)$. However, when $a_1<\tfrac{1}{2}$, the speed is a continuous piecewise rational function of $p_1\in(0,1)$, with a point of non-differentiability at the critical point $p_1^c=\tfrac{a_1}{1-a_1}\in(0,1)$, yielding a phase transition for the system.
 \end{example}

The results for the liquid bin model are proved by mapping it to a more tractable model of cars driving one behind the other on a semi-infinite road. See Section \ref{sec:LBMcar} for the definition of the car model and its coupling to the liquid bin model.

One may be tempted to generalize the liquid bin model to the case when $N=\infty$. Caution should be exercised in the definition to avoid finite-time explosion. For instance one may assume that the $p_i$ parameters have a finite sum. In that setting, the existence of a stationary trajectory should be obtained in a similar fashion as for the finite $N$ case. However the uniqueness of such a stationary trajectory does not seem to be clear. An explicit computation of the speed seems even harder to obtain.

\subsection{Relation to circular extensions}

Another remarkable feature of the liquid bin model is its connection to the growing field of enumerating extensions of partial cyclic orders to total cyclic orders. A \emph{cyclic order} on a set $X$ is a subset $Z$ of triples of distinct elements of $X$ satisfying the following three axioms, respectively called cyclicity, asymmetry and transitivity:
\begin{enumerate}
 \item $\forall x,y,z\in X, (x,y,z)\in Z \Rightarrow (y,z,x)\in Z$;
 \item $\forall x,y,z\in X, (x,y,z)\in Z \Rightarrow (z,y,x)\notin Z$;
 \item $\forall x,y,z,u\in X, (x,y,z)\in Z$ and $(x,z,u)\in Z \Rightarrow (x,y,u)\in Z$.
\end{enumerate}
A cyclic order $Z$ is called \emph{total} if for every triple of distinct elements $(x,y,z)\in X^3$, either $(x,y,z)\in Z$ or $(z,y,x)\in Z$. Otherwise, it is called \emph{partial}. A total cyclic order $Z$ on $X$ is a way of placing
all the elements of $X$ on a circle such that a triple $(x,y,z)$ lies in $Z$ whenever $y$ lies on the cyclic interval from $x$ to $z$ when turning around the circle in the clockwise direction. This provides a bijection between total cyclic orders on $X$ and cyclic permutations on $X$. If $Z$ (resp. $Z'$) is a total (resp. partial) cyclic order on $X$, $Z$ is called a \emph{circular extension} of $Z'$ if $Z'\subset Z$. The enumeration of particular classes of circular extensions started in \cite{Ramassamy}. In \cite{AJVR,GHMY} and \cite{PSBTW}, more classes were enumerated and related to the volumes of certain integral polytopes.  

Denote by $Q^{2N}$ the subset of all $(\underline a,\underline p)\in P^{2N}$ such that $\Gr(\underline a,\underline p)$ is connected and such that no two cursors jump at the same time in the stationary evolution associated with $(\underline a,\underline p)$. For any $(\underline a,\underline p)\in Q^{2N}$, one may consider the cyclic order in which the cursors jump in the stationary evolution $\tilde{x}_\infty$ associated with $(\underline a,\underline p)$. This is a total cyclic order $Z$ on the set $\llbracket 1,N\rrbracket$. Let us explain how to recover the DC graph $\Gr(\underline a,\underline p)$ from $Z$. For every $m\geq3$, the $m$-tuple $(i_1,\ldots,i_m)\in\llbracket 1,N\rrbracket^m$ is called \emph{a $Z$-chain} if for every $2\leq k\leq m-1$, $(i_1,i_k,i_{k+1})\in Z$. It intuitively means that, starting from $i_1$ and turning clockwise, we first see $i_2$, then $i_3$, etc, before returning to $i_1$. By convention, all $m$-tuples are $Z$-chains for $m\in\{1,2\}$.

\begin{proposition}
\label{prop:mapfactorialCatalan}
Let $Z$ denote the total cyclic order in which the cursors jump in the stationary evolution $\tilde{x}_\infty$ associated with $(\underline a,\underline p)\in Q^{2N}$. For every $i\in\llbracket 1,N\rrbracket$, let $\beta_Z(i)$ be the largest $m\in\llbracket i,N\rrbracket$ such that $(i,i+1,\ldots,m)$ forms a $Z$-chain. Let $G$ be the graph on the vertex set $\llbracket 1,N\rrbracket$ such that $(i,j)\in E_N$ is an edge if and only if $i\leq j\leq \beta_Z(i)$. Then $G=\Gr(\underline a,\underline p)$.
\end{proposition}

\begin{proof}
Denote by $T$ the duration of one stationary cycle, namely the time it takes for the first $a_N$ units of liquid to shift by one bin to the right in the stationary evolution $\tilde x_\infty$. Let $i\in\llbracket1,N\rrbracket$. Up to a translation in time, one may assume that the $i$-th cursor jumps at time $0$. Thus its next jump is at time $T$. The largest cursor present in the same bin as cursor $i$ just before time $T$ is cursor $b_{\Gr(\underline a,\underline p)}(i)$. Between times $0$ and $T$, cursors $i+1, i+2, \ldots,b_{\Gr(\underline a,\underline p)}(i)$ will have to jump in this order, hence $(i,i+1,\ldots,b_{\Gr(\underline a,\underline p)}(i))$ forms a $Z$-chain. Assume that $b_{\Gr(\underline a,\underline p)}(i)<N$. Since $\Gr(\underline a,\underline p)$ is connected, $(b_{\Gr(\underline a,\underline p)}(i),b_{\Gr(\underline a,\underline p)}(i)+1)$ has to be an edge of this graph. Thus there exists a time between $0$ and $T$ when both cursors $b_{\Gr(\underline a,\underline p)}(i)$ and $b_{\Gr(\underline a,\underline p)}(i)+1$ are in the same bin. Since they are not in the same bin at time $0$, it means that cursor $b_{\Gr(\underline a,\underline p)}(i)+1$ jumps between time $0$ and the jump time of cursor $b_{\Gr(\underline a,\underline p)}(i)$. Thus $(i,i+1,\ldots,b_{\Gr(\underline a,\underline p)}(i),b_{\Gr(\underline a,\underline p)}(i)+1)$ is not a $Z$-chain, hence $b_{\Gr(\underline a,\underline p)}(i)=\beta_Z(i)$, which concludes the proof.
\end{proof}

Proposition~\ref{prop:mapfactorialCatalan} provides a map $F_N$ from the set of cardinality $(N-1)!$ of total cyclic orders on $\llbracket 1,N\rrbracket$ to the set of cardinality $C_{N-1}$ of connected DC graphs $G\in\DC_N$. While this map looks natural, we do not know whether it has previously appeared in the literature. Every connected DC graph $G\in\DC_N$ has at least one pre-image by $F_N$, since $P_G$ is non-empty by Theorem \ref{theorem:nonemptiness}.

Conversely, to each connected DC graph $G\in\DC_N$, we associate the partial cyclic order $Z'_G$, by requiring that for every maximal edge $(i,j)$ of $G$, the tuples $(i,i+1,\ldots,j)$, $(i,i-1,j)$ and $(j,i,j+1)$ are $Z'_G$-chains. An edge $(i,j)\in E(G)$ is called \emph{maximal} if it is a maximal element in the poset $(E(G),\prec_E)$. It is not hard to check that $Z'_G$ is a well-defined partial cyclic order and that the set of all circular extensions of $Z'_G$ is precisely the fiber of $F_N$ above $G$, \textit{i.e.} for every circular extension $Z$ of $Z'_G$, $F_N(Z)=G$.

\begin{conjecture}
\label{conj:cyclicorder}
For every connected DC graph $G\in\DC_N$, every circular extension of $Z'_G$ arises as the total cyclic order of cursor jumps for some value of parameters $(\underline a,\underline p)\in P_G$.
\end{conjecture}

We have numerically verified this conjecture for every $N\leq4$. When $G$ is the complete graph, there is a single circular extension, namely the only total cyclic order $Z$ for which $(1,2,\ldots,N)$ forms a $Z$-chain. When $G$ is the graph that has exactly $N-1$ edges, connecting $i$ to $i+1$ for every $1\leq i\leq N-1$, then the number of circular extensions is the Euler zigzag number \cite{Ramassamy}.

An \emph{integral polytope} in dimension $d$ is a polytope with vertices in $\ZZ^d$. Its \emph{normalized volume} is the integer obtained by multiplying its volume by $d!$. A branch of research in enumerative combinatorics is concerned with mapping certain natural families of integral polytopes to certain families of objects that are enumerated by the normalized volumes of these polytopes, see e.g. the references in the introduction of \cite{AJVR}. In \cite{Stanleyposet}, Stanley associates to each partially ordered set two integral polytopes whose normalized volumes are equal and enumerate the linear extensions of the original partially ordered set. In a similar spirit, but for circular extensions instead of linear extensions, the recent papers \cite{AJVR,GHMY,PSBTW} associate to some partial cyclic order two integral polytopes whose normalized volumes are equal and enumerate the circular extensions of the partial cyclic order. However it is not known whether such a construction is possible for every partial cyclic order. The partial cyclic orders $Z'_G$ for connected DC graphs $G$ provide new examples of partial cyclic orders on which to try this construction. A natural candidate would be the consecutive coordinate polytopes appearing in \cite{AJVR,GHMY}.

\subsection*{Organization of the paper}
In Section~\ref{sec:LBMcar}, we provide a rigorous construction of the liquid bin model and we show that it is coupled to a model of cars. In Section~\ref{sec:stationary} we prove the existence and uniqueness of a stationary trajectory for the car model and we show that the shifted trajectories of cars starting from an arbitrary initial configuration converge exponentially fast to the stationary trajectory. Section~\ref{sec:speed} is dedicated to the proof of Theorem~\ref{theorem:mainspeed}, obtained by deriving a more refined result for the car model. In that section we partition the parameter space into regions $P_G$, we associate to each region a linear system and we solve it. We also show that each region is non-empty. Finally in Section~\ref{sec:adjacency} we prove Theorem~\ref{theorem:mainadjacency} about the adjacency structure of the regions $P_G$.

\section{The liquid bin model}
\label{sec:LBMcar}

In the introduction we provided a heuristic definition of the liquid bin model. The aim of Subsection~\ref{subsec:defLBM} is to give a rigourous construction of the liquid bin model. In Subsection~\ref{subsec:car} we describe its coupling with a model of cars, the study of which will yield in further sections the proofs of all the results announced in the introduction.

\subsection{Definition of the liquid bin model}
\label{subsec:defLBM}

Say that a \emph{configuration} is an element $x = ( x_k )_{k \in \ZZ} $ of $(\RR_{\geq0})^{\ZZ}$. One may interpret $x$ as a configuration of liquid in bins indexed by $\ZZ$, where $x_k\geq0$ represents the quantity of liquid in the $k$-th bin, for all $ k \in \ZZ$.  A configuration $x$ is called \emph{admissible} if the following two conditions are satisfied:
\begin{itemize}
    \item There is an infinite quantity of liquid in the configuration: $\sum_{k \in \ZZ} x_k = +\infty$.
    \item There exists $f(x) \in \ZZ$ such that for all $k \in \ZZ$, $x_k$ is positive if and only if $k \leq f(x)$. $f(x)$ is called the \emph{front} of $x$.
\end{itemize}
Denote the set of admissible configurations of bins by $\binConf$.

Let $N\in\ZZ_{>0}$. Recall from Subsection~\ref{subsec:LBMproperties} the definitions of the parameter space $P^{2N}$ and the alternative parameters $\underline d$ and $\underline q$. We adopt the conventions that $a_0:=0$, $a_{N+1}:=\infty$ and $q_0:=0$.

Let us introduce a deterministic dynamics associated to the parameters $(\underline a,\underline p)\in P^{2N}$ on the set of configurations $\binConf$. 
For $x \in \binConf$ and $ 1 \leq i \leq\nbParam$, the \emph{$i$-th cursor} $\cursorLetterBin_i(x)$ is defined as the highest index among bins, such that the total amount of liquid in that bin and to its right is at least $a_i$. 
More precisely, 
\begin{equation*}
    \cursorLetterBin_i(x) := \max\lrSet{m \in \ZZ}{\sum_{k \geq m } x_k \geq a_i}.
\end{equation*}
for all $i \in \llbracket 1 , {\nbParam} \rrbracket$. Notice that $ -\infty < \cursorLetterBin_i(x) < +\infty$ for all $x \in \binConf$ since there is an infinite amount of liquid in $x$ and the front of $x$ is non-trivial.  By definition of the $i$-th cursor, 
\[\sum_{k \geq \cursorLetterBin_i(x)+1 } x_k <  a_i \leq \sum_{k \geq \cursorLetterBin_i(x) } x_k. \]

In Subsection~\ref{subsec:IBMLBM} we gave a heuristic definition of the liquid bin model. Let us now construct the dynamics $\binDynLetter$ of the liquid bin model in more rigorous terms. 
More precisely, denote by $e(j) \in (\RR_{\geq0})^{\ZZ}$ the sequence such that $e_k(j) = \mathds{1}_{k=j}.$ Define $\binDynLetter^{(0)}$ as the map from $\RR_{\geq0} \times \binConf$ to $\binConf$, such that for all $x \in \binConf$ and $t \in \RR_{\geq0}$, 
\begin{align}\label{eq:defDynBin0}
    \binDynLetter^{(0)}(t,x) := x +  t \sum_{i = 1} ^{\nbParam} p_i \cdot e(\cursorLetterBin_i(x) + 1).
\end{align}
Let $\nextJumpTimeBin_1 (x)$ be the first positive time at which a cursor changes position starting from the configuration $x$ for the dynamics $\binDynLetter^{(0)}$: 
\begin{align}\label{eq:defTimeBin1}
\nextJumpTimeBin_1 (x) := \inf\lrSet{t \in \RR_{>0}}{ \exists i \in \llbracket 1 , \nbParam \rrbracket ,\,  \sum_{k \geq \cursorLetterBin_i(x)+1}(\binDynLetter^{(0)} (t,x))_{k} = a_i } .
\end{align}

\begin{remark}
By definition of $\nextJumpTimeBin_1$, notice that \[\nextJumpTimeBin_1 (x) := \min_{1\leq i \leq \nbParam}\left(\frac{a_i - \sum_{k \geq \cursorLetterBin_i(x) + 1} x_k}{\sum_{j, \, \cursorLetterBin_j(x) \geq \cursorLetterBin_i(x)} p_j}\right).\]
\end{remark}

For all times $t \leq \nextJumpTimeBin_1 (x)$, $\binDynLetter^{(0)}(t, x)$ corresponds to the dynamics of the bin model. After time $\nextJumpTimeBin_1 (x)$, at least one cursor position changed. Therefore, $\binDynLetter^{(0)}(t,x)$ no longer aligns with the dynamics of the liquid bin model as previously heuristically defined.
Then, set the map $\binDynLetter $ such that for all $x \in \binConf$ and $0 \leq  t \leq \nextJumpTimeBin_1(x)$,  
\begin{equation*}
    \binDyn{t}{x}:=\binDynLetter^{(0)} (t,x).
\end{equation*}
By induction on $l$, set $\nextJumpTimeBin_l(x) = \nextJumpTimeBin_{l-1}(x) + \nextJumpTimeBin_{1}( \binDyn{\nextJumpTimeBin_{l-1}(x)}{x})$ for all $l \geq 2$, with the convention that $\nextJumpTimeBin_0(x)=0$. Let us construct the dynamics $\binDynLetter_t$ by induction, resetting the dynamics according to the cursors at each time $\nextJumpTimeBin_l(x)$.
We define $\binDyn{t}{x}$ for all values of $t \in [ \nextJumpTimeBin_l(x) , \nextJumpTimeBin_{l+1}(x) ]$ by induction on $l$ as  
\begin{equation*}
    \binDyn{t}{x} := \binDynLetter^{(0)}(t-\nextJumpTimeBin_l(x), \binDyn{\nextJumpTimeBin_l(x)}{x}).
\end{equation*}

\begin{remark}\label{remark:timesToInftyBin}
    With this definition, $\binDyn{t}{x}$ is well-defined for all $t \in \RR_{\geq0}$ for all $x \in \binConf$. To prove this, let us show that $\nextJumpTimeBin_l(x)$ converges to $+\infty$ as $l$ tends to $+\infty$ for every configuration $x \in \binConf$.
    Consider $(t_l)_{l \geq 1}$ the subsequence of $(\nextJumpTimeBin_l(x))_{l \geq 1}$ consisting  of all times at which cursor $1$ changes its position. Notice that at all times, the rate of liquid added to the right of cursor $1$ is bounded below by $q_1 = p_1$ and bounded above by $q_{\nbParam} = p_1 + \dots + p_{\nbParam}$. Therefore, for all $l \geq 1$, $\frac{a_1}{q_{\nbParam}} \leq t_l -t_{l-1} \leq \frac{a_1}{q_1}$. Since $(\nextJumpTimeBin_l(x))_{l \geq 1}$ is increasing, it tends to infinity as $l$ goes to infinity. 
\end{remark}

\begin{example}
Let $x^{(0)}$ be the configuration such $x^{(0)}_k = 1.5 $ for $k\leq1$, $x^{(0)}_2 = 1 $ and  $x^{(0)}_k = 0 $ for $k\geq3$, as illustrated on the top-left picture of Figure \ref{fig:LMBdyn}. Then the three bins represented on the bottom-left picture of Figure \ref{fig:LMBdyn} correspond to the bins $1$, $2$ and $3$ of the configuration $x^{(1)} := \binDynLetter_{0.25}(x^{(0)})$. The three bins on the bottom-right picture of Figure \ref{fig:LMBdyn} correspond to the configuration $\binDynLetter_{0.5}(x^{(1)}) = \binDynLetter_{0.75}(x^{(0)}) $.  Note that $\nextJumpTimeBin_1(x^{(0)}) = 0.25$ and $\nextJumpTimeBin_1(x^{(1)}) = 0.5$, as these are the first two times at which a cursor changes bins. 
\end{example}

\subsection{Coupling with a car model}
\label{subsec:car}

In this subsection, we give a coupling of the liquid bin model with a model of cars evolving on $\RR_{\geq0}\cup\{ +\infty\}$. Heuristically, the main idea is to draw the bins with a fixed unit height, rather than with a fixed unit width. We associate a car to each wall separating two consecutive bins, so that the amount of liquid in a bin corresponds to the distance between two consecutive cars. See Figure \ref{fig:couplingLBMcar} for an illustration.

\begin{figure}
    \centering
    \includegraphics[scale=0.71]{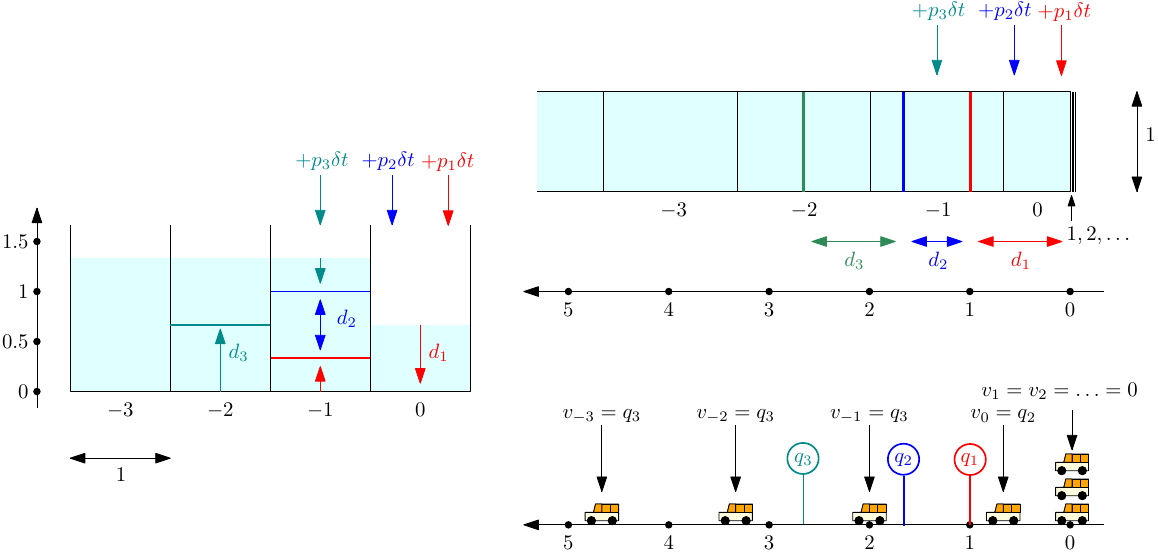}
    \caption{Illustration of the coupling between the liquid bin model and the car model with parameters $N=3$, $d_1 = 1$, $d_2 =\frac{2}{3}$, $d_3 = 1$. On the left a picture of the bin configuration with unit-width bins. On the top right a picture of the same bin configuration with unit-height bins. On the bottom right the corresponding car configuration. For this bin configuration, we assume that the bins with indices greater than zero are empty so that the front  position is $0$. The difference of positions between cars in the car model corresponds to the quantity of liquid in the bin model. The speed of each car is indicated above it. The locations of the cars at a positive position that are depicted here are the following: $y_0=\tfrac23$, $y_{-1}=2$, $y_{-2}=\tfrac{10}{3}$, $y_{-3}=\tfrac{14}{3}$.}
    \label{fig:couplingLBMcar}
\end{figure}

\begin{figure}
    \centering
    \includegraphics[scale=0.9]{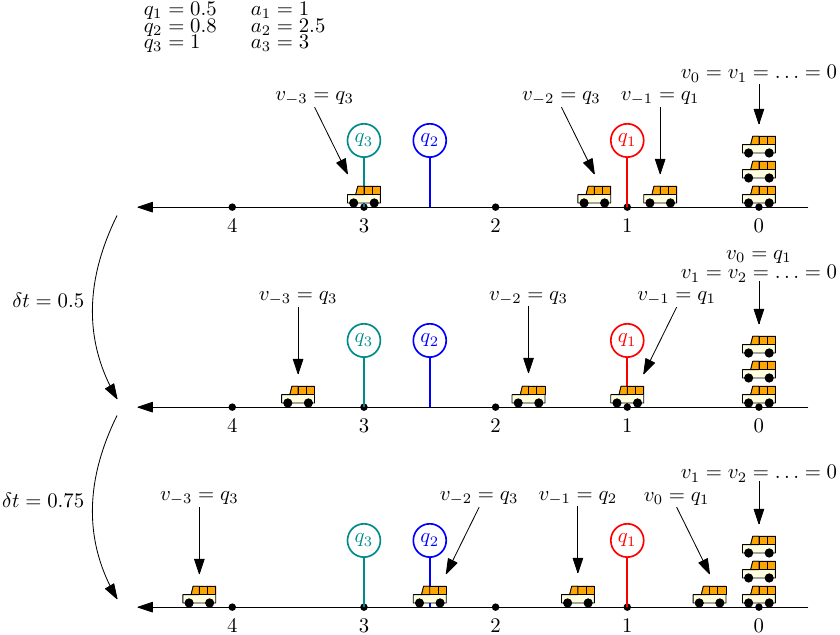}
    \caption{Illustration of the dynamics of the car model. The initial configuration is depicted at the top. It would correspond to a bin configuration with the front at position $-1$. On the figure, $v_i$ denotes the speed of the car with index $i$.  After $0.5$ units of time, the car $-1$ arrives at position $a_1$. Therefore, the car $0$ starts moving at speed $q_1$. Then, after $0.75$ units of time, the car $-2$ arrives at position $a_2$. Therefore, the car $-1$ now moves at speed $q_2$.}
    \label{fig:carModel}
\end{figure}

The car model is defined heuristically as follows: at position $a_i$, there is a road sign telling cars to move at speed $q_i=p_1+\cdots+p_i$. A road sign is visible to a car if and only if no other car lies between that car and the road
sign (one may imagine that the speed limits are painted on the floor, with cars hiding them from those behind).
At any time, the speed of each car is equal to the maximal speed restriction that they have seen or can see:   
A car moves at speed $q_i$ if and only if $i$ is the greatest index of road signs visible to this car, or already passed by this car. 
See Figure \ref{fig:carModel} for an illustration of the car model. As will be apparent when we describe the coupling between the bin model and the car model, it is more convenient to direct to positive real half-line towards the left for the car model. 

More precisely, denote by $\carConf$ the set consisting of all the elements $(y_k)_{k \in \ZZ} \in (\RR_{\geq0}\cup\{+\infty\})^{\ZZ}$ such that for some $f \in \ZZ$, 
\begin{itemize}
    \item $y_k = 0$ for all $k > f$, 
    \item $0 < y_f < y_{f-1} < \dots$,
    \item $y_k \xrightarrow[k \rightarrow -\infty]{} +\infty$.
\end{itemize}
An element $y$ of $\carConf$ can be seen as a configuration of cars, where $y_k$ is the position of the $k$-th car. 
When $ y_k = +\infty$, assume there is no car with index $k$ or less in the configuration $y$. For $y \in \carConf $, denote by $\front{y} = \max\Set{k \in \ZZ}{y_k > 0}$ the index of the car closest to $0$ among cars with positive positions. 
Note that by definition of $\carConf$, for any configuration $y \in \carConf$, there is an infinite number of cars at position $0$. Moreover, by the third point, there is a finite number of cars with position in any compact subset of $\RR_{>0}$.

As we did for the liquid bin model, let us construct the dynamics $\Psi$ of the car model.  For $i \in \llbracket 1 , \nbParam \rrbracket $, set \begin{equation*}
\cursorLetterCar_i(y) := \max\lrSet{k \in \ZZ}{y_k \geq a_i } .
\end{equation*}
By definition of $\carConf$, $-\infty < \cursorLetterCar_i(y)< +\infty$ for all $y \in \carConf$. According to the heuristic of the dynamics, since $ \cursorLetterCar_i(y)  $ is the last car which passed by the road sign at position $a_i$, we want $\cursorLetterCar_i(y) + 1$ to be the largest index among indices of cars moving at speed at least $q_i$. Then, consider the map $\carDynLetter^{(0)}$ such that for all $t \in \RR_{\geq0}$ and $y \in \carConf$, 
\begin{align}\label{eq:defDynCar0}
\carDynLetter^{(0)}(t,y) = y + t \sum_{i = 1}^{\nbParam} p_i \cdot \widehat{e}(\cursorLetterCar_i(y) +1),
\end{align}
where for all $j \in \ZZ$, $\widehat{e}(j) \in \RR^{\ZZ}$ is the sequence where $\widehat{e}_k(j) = \mathds{1}_{k \leq j}$.  
One may give an alternative definition of $ \carDynLetter^{(0)} $ that is equivalent to \eqref{eq:defDynCar0}.  
We adopt the convention that $a_0=0$. For $ k \in \ZZ$, set $\indexOfSpeedCar_{k}(y) := \max\Set{ i \in \llbracket 0 , \nbParam \rrbracket }{ a_i \leq y_{k-1}}$, representing the greatest index of the road signs passed by the $(k-1)$-th car. 
For $t \in \RR_{\geq0}$, notice that 
\begin{align}\label{eq:defDynCar0bis}
    \carDynLetter^{(0)}_k(t,y) = y_k +  t \cdot q_{\indexOfSpeedCar_{k}(y) } ,
\end{align} where $\carDynLetter^{(0)}_k(t,y)$ corresponds to the position of the $k$-th car in $\carDynLetter^{(0)}(t,y)$.
With this definition, $\indexOfSpeedCar_{k}(y)$ is the index of the leftmost road sign that is or was visible to the $k$-th car in the configuration $y \in \carConf$.  

Set $\nextJumpTimeCar_1 (y)$ to be the first time at which a car passes a road sign:
\begin{align}\label{eq:defTimeCar1}
    \nextJumpTimeCar_1 (y) := \inf\lrSet{t \in \RR_{>0}}{ \exists (k,i) \in \ZZ\times \llbracket 1 , \nbParam \rrbracket ,\,  \carDynLetter^{(0)}_k (t,y) = a_i }.
\end{align}
After this time, the speed of at least one of the cars changes. Then, define the map $\carDynLetter $ such that for all $y \in \carConf$ and $0 \leq  t \leq \nextJumpTimeCar_1(y)$,  
\begin{equation*}
    \carDyn{t}{y}:=\carDynLetter^{(0)} (t,y).
\end{equation*}
To construct $\carDynLetter(t,y)$ for all $t > 0$, proceed recursively as done with the liquid bin model: set by induction $\nextJumpTimeCar_l(y) = \nextJumpTimeCar_{l-1}(y) + \nextJumpTimeCar_{1}( \carDyn{\nextJumpTimeCar_{l-1}(y)}{y})$ for all $l \geq 2$. 
We define $\carDyn{t}{y}$ for all values of $t \in [ \nextJumpTimeCar_l(y) , \nextJumpTimeCar_{l+1}(y) ]$ by induction on $l$ as  
\begin{equation*}
    \carDyn{t}{y} := \carDynLetter^{(0)}_{}(t-\nextJumpTimeCar_l(y), \carDyn{\nextJumpTimeCar_l(y)}{y}).
\end{equation*} 

Denote by $\Sigma : \binConf \longrightarrow \carConf$ the map such that for any bin configuration $x\in\binConf$,
\begin{equation}
\label{eq:defsigma}
\Sigma x := \left( \sum_{j = k}^{+\infty} x_j\right)_{k \in \ZZ}.
\end{equation}
It is easy to check that $\Sigma$ is a one-to-one correspondence from $\binConf$ to $\carConf$. In this subsection, we prove the equivalence of the liquid bin model and the car model via the following coupling property, illustrated on Figure \ref{fig:couplingLBMcar}.

\begin{proposition}[coupling property]\label{proposition:couplingLiquidCar}
    For all $t \in \RR_{\geq0}$, $ \carDynLetter(t,\cdot) \circ \Sigma   = \Sigma \circ \binDynLetter(t,\cdot)$.
\end{proposition}

\begin{proof}
    Since the liquid bin model and the car model are constructed recursively, it suffices to show that:
    \begin{enumerate}
        \item for all $t \geq 0$,  $\carDynLetter^{(0)}(t , \cdot) \circ \Sigma =  \Sigma \circ \binDynLetter^{(0)}(t , \cdot) $, 
        \item $\nextJumpTimeBin_1 = \nextJumpTimeCar_1 \circ \Sigma$. 
    \end{enumerate}

    We start with the proof of the first point. Note that $e(j)\notin\binConf$ for any $j\in\ZZ$, because it contains only a finite total amount of liquid. We define $\binConf'$ to be the set of all the configurations of liquid in bins which have a finite front (but may have a finite total amount of liquid). The definition of the map $\Sigma$ can be extended to $\binConf'$, thus $\Sigma e(j)$ becomes well-defined for any $j\in\ZZ$.
    It is straightforward that $\cursorLetterBin_i(x) = \cursorLetterCar_i(\Sigma x)$ for all $i \in \llbracket 1 , \nbParam \rrbracket$, that $\Sigma$ is linear and that $\widehat{e}(j) = \Sigma e(j)$ for all $j \in \ZZ$. Therefore, 
    \begin{align*}
        \carDynLetter^{(0)}(t, \Sigma x ) &= \Sigma x + \sum_{i=1}^{\nbParam} p_i \cdot \widehat{e}( \cursorLetterCar_i(\Sigma x) + 1 ) \\
        &= \Sigma x + \sum_{i=1}^{\nbParam}  p_i \cdot \Sigma e(  \cursorLetterCar_i(\Sigma x) + 1   )\\
        &= \Sigma x + \sum_{i=1}^{\nbParam}  p_i \cdot \Sigma e(  \cursorLetterBin_i( x) + 1   )\\
        &= \Sigma \left(  x + \sum_{i=1}^{\nbParam} p_i \cdot  e( \cursorLetterBin_i( x) + 1 ) \right)\\
        &= \Sigma \left( \binDynLetter^{(0)}( t, x ) \right)
    \end{align*}

    Since the car with index $\cursorLetterCar_i(y)$ is the closest car lying to the left of road sign $i$, by \eqref{eq:defTimeCar1}, notice that 
\begin{align*}
    \nextJumpTimeCar_1 (x) = \inf\lrSet{t \in \RR_{\geq0}}{ \exists i \in \llbracket 1 , \nbParam \rrbracket ,\,  \carDynLetter^{(0)}_{\cursorLetterCar_i(y)+1}(t, y) = a_i }.
\end{align*}
    Therefore, the second point is a straightforward consequence of the first point by \eqref{eq:defTimeBin1}. 
\end{proof}

\begin{remark}\label{remark:timesToInftyCar}
    Since $\nextJumpTimeCar_l(\Sigma x) = \nextJumpTimeBin_l(x)$ for all $l \geq 1$ and $x \in \binConf$, $\carDyn{t}{y}$ is also well-defined for all $t \in \RR_{\geq0}$ and $y \in \carConf$ by Remark \ref{remark:timesToInftyBin}. 
\end{remark}

The car model comes with a natural recurrence relation stated in Proposition~\ref{proposition:recursiveFormulaSpeed}, whereby the trajectory of one car only depends on the trajectory of the car to its left. Considering a fixed-point-type equation for this recurrence will yield the stationary trajectory. While one may rephrase this entirely in terms of the bin model without defining the car model, the quantities that would satisfy such a recurrence would be bulkier, like the right-hand side of formula \eqref{eq:defsigma}, making them harder to work with.

\section{The stationary car model}
\label{sec:stationary}

In this section, the parameters $(\underline a,\underline p)\in P^{2N}$ are kept constant. The goal of this section is to establish the convergence in time of the car model, and consequently, of the liquid bin model, to a stationary configuration. The main result of this section is Theorem \ref{theorem:convergenceOfDynamics}, stating the uniqueness of the stationary trajectory and the convergence result for the car model. 

\subsection{Stationary configurations and trajectories}

From now on, consider an initial configuration of cars $y(0) \in \carConf$ and set \[y(t) := \carDyn{t}{y(0)}\] to be the configuration of cars obtained after time $t\geq0$. Denote by \[y_k(t) :=\carDynLetter_k(t, y(0))\] the position of the $k$-th car at time $t$ for all $k \in \ZZ$. We say that $t\mapsto y_k(t)$ is the \emph{trajectory} of the $k$-th car.

Let us show some basic results on trajectories. 
Consider $k \in \ZZ$. If $y_k(0) \geq a_1$, then the $k$-th car moves at least at speed $q_1$ at any time $t\geq0$ by construction of the car model. Moreover, among all cars with positions in $[0, a_1]$, only the one with maximal position  moves. Since there are finitely many cars with positions in $(0, a_1]$ in any configuration in $\carConf$, each car in the system will eventually move. As a consequence, for any $k\in\ZZ$ and for any position $x > y_k(0)$, there exists a unique time at which the $k$-th car is at position $x$.
Denote by
\[T_k := \inf\Set{t \in \RR_{\geq0}}{y_{k-1}(t)\geq a_1}\]
the time at which the $k$-th car starts moving. Then, the restriction of $y_k$ provides an increasing bijective function from $[T_k, +\infty)$ to $[y_k(0), +\infty)$. Let us denote by $t_k := y^{-1}_k$ the inverse of this bijection.
Notice that for all $x>y_k(0)$, $t_k(x)$ is the only time at which the $k$-th car is at position $x$. For $0 \leq x < y_k(0) $, set $t_k(x) := 0$ by convention. For any $x\in\RR$, we define the positive part of $x$ to be $x_+:=\max(x,0)$.

With these notations, one may give a recursive formula for the sequence of trajectories:

\begin{proposition}[recursive formula]\label{proposition:recursiveFormulaSpeed}
    With the above notations, for all $t\geq0$ and $k \in \ZZ$ we have
\begin{align}\label{eq:recursiveFormulaSpeed}
    y_k(t) = y_k(0) + \sum_{j=1} ^{\nbParam} p_j (t-t_{k-1}(a_j))_+.
    \end{align}
\end{proposition}

\begin{proof}
    It suffices to note that the $k$-th car gets from speed $q_{i-1}$ to speed $q_i$ as soon as car $k-1$ passes the road sign at position $a_i$ for all $k \in \ZZ $ and $i \in \llbracket 1 , \nbParam \rrbracket$, with the convention that $q_0=0$.

    Fix $k \in \ZZ$. For all $i \in \llbracket 1 , \nbParam \rrbracket$, note that $t_{k-1}(a_i) = \inf\Set{t \in \RR_{\geq0}}{y_{k-1}(t) \geq a_i}$.  Moreover, for all $i \in \llbracket 0, \nbParam \rrbracket$ and $t \in [t_{k-1}(a_i), t_{k-1}(a_{i+1}))$, $a_i \leq y_{k-1}(t) < a_{i+1}$. Thus, by \eqref{eq:defDynCar0bis}, for all $t \in [t_{k-1}(a_i), t_{k-1}(a_{i+1}))$,
    \[
    y_{k}(t) =  y_{k}(t_{k-1}(a_i)) + (t-t_{k-1}(a_i)) q_i.
    \]
    As a consequence, for all $t \in [t_{k-1}(a_i), t_{k-1}(a_{i+1}))$, by induction on $i \in \llbracket 0 , \nbParam \rrbracket$, 
    \begin{align*}
        y_{k}(t) &= y_{k}(t_{k-1}(a_1)) + (t-t_{k-1}(a_i))q_i +\sum_{j=1}^{i-1} (t_{k-1}(a_{j+1})-t_{k-1}(a_j))q_j \\
        &= y_{k}(t_{k-1}(a_1)) +  \sum_{j=1}^{i} (t-  t_{k-1}(a_j))p_j,
    \end{align*}
    which coincides with \eqref{eq:recursiveFormulaSpeed} for $t \in [t_{k-1}(a_i), t_{k-1}(a_{i+1}))$, since $y_k(0)=y_k(t_{k-1}(a_1))$.
\end{proof}

For all $k \in \ZZ$ and $t \in \RR_{\geq0}$, let $ \bar{y}_k$ be the trajectory of the $k$-th car in $y$ shifted in time so that a car with this trajectory starts moving at time $0$: 
\begin{equation}\label{eq:defTrajTimeRescaled}
\forall t \in \RR_{\geq0}, \,\bar{y}_k(t) := y_k( t + T_k). 
\end{equation}
The function $\bar{t}_k$ defined as $\bar{t}_k(x) := t_k(x) - T_k $ for all $x \geq y_k(0)$ corresponds to the inverse function of $\bar{y}_k$ from $\RR_{\geq0}$ to $[y_k(0), +\infty)$. We also adopt the convention that $\bar{t}_k(x):=0$ for all $x < y_k(0)$.

Consider $k \in \ZZ$ such that $y_k(0)=0$. By construction of the car model, $ T_k = t_{k-1}(a_1) $ since the time at which the $k$-th car starts moving is the time at which the car of index $k-1$ passes by the road sign at position $a_1$. Then, by Proposition \ref{proposition:recursiveFormulaSpeed}, for all $t \geq 0$ and $k \in \ZZ$ such that $y_k(0)=0$, we have
\begin{align}\label{eq:recursiveFormulaSpeedTilde2}
\bar{y}_k(t) = \sum_{j=1}^\nbParam p_j (t  - \bar{t}_{k-1}(a_j) + \bar{t}_{k-1}(a_1) )_+ .
\end{align}

\begin{definition}
\label{def:stationarytrajectory}
A \emph{stationary trajectory} is a continuous increasing map $t \mapsto \widetilde{y}_\infty(t)$ from $\RR_{\geq0}$ to $\RR_{\geq0}$ such that for all $t\geq 0$,
\begin{align}\label{eq:recursiveFormulaSpeedTilde3}
\widetilde{y}_\infty(t) = \sum_{j=1}^{\nbParam} p_j (t - \widetilde{t}_\infty(a_j) + \widetilde{t}_\infty(a_1) )_+,
\end{align}
 where $\widetilde{t}_\infty := \widetilde{y}_\infty^{-1}$ is the inverse function of $\widetilde{y}_\infty$. In other words, $\widetilde{y}_\infty$ is a stationary trajectory if and only if it is a fixed point of \eqref{eq:recursiveFormulaSpeedTilde2}.
\end{definition}

\begin{definition}
\label{def:canonicalstationaryconfiguration}
Let $\widetilde{y}_\infty$ be a stationary trajectory. We define the \emph{canonical stationary configuration for $\widetilde{y}_\infty$} to be the configuration $(y_k(0))_{k \in \ZZ} \in \carConf$ such that for every $k\leq0$, $y_{k}(0)= \widetilde{y}_\infty(-k\widetilde{t}_\infty(a_1))$ and for every $k>0$, $y_{k}(0)= 0$.
\end{definition}

\begin{lemma}
\label{lem:stationarityofstationaryconfiguration}
Let $\widetilde{y}_\infty$ be a stationary trajectory and let $(y_k(0))_{k \in \ZZ}$ be the canonical stationary configuration for $\widetilde{y}_\infty$. For every $t\geq0$ and $k\in\ZZ$, we have
\begin{equation}
\label{eq:evolcanonicalstatioconfig}
y_k(t)=\begin{cases}
\widetilde{y}_\infty(t-k\widetilde{t}_\infty(a_1)) & \text{ if } t\geq k\widetilde{t}_\infty(a_1)\\
0 & \text{ if } t< k\widetilde{t}_\infty(a_1).
\end{cases}
\end{equation}
\end{lemma}

It follows from Lemma~\ref{lem:stationarityofstationaryconfiguration} that the configuration $(y_k(\widetilde{t}_\infty(a_1)))_{k \in \ZZ}$ is equal to $(y_{k-1}(0))_{k \in \ZZ}$, hence justifying the use of the adjective ``stationary''.

\begin{proof}
Let us prove formula~\eqref{eq:evolcanonicalstatioconfig} for all $k\in\ZZ$ by induction on $k$.

Assume first that $k\leq0$. Then for every $t\geq0$, formula~\eqref{eq:evolcanonicalstatioconfig} specializes to
\begin{equation}
\label{eq:evolcanonicalstatioconfigbis}
y_k(t)=\widetilde{y}_\infty(t-k\widetilde{t}_\infty(a_1)).
\end{equation}

Formula~\eqref{eq:evolcanonicalstatioconfigbis} holds true at time $0$ for all $k\leq0$ by Definition~\ref{def:canonicalstationaryconfiguration}. Let $K_1\leq0$ be such that for all $k\leq K_1$, $y_k(0)\geq a_\nbParam$. Then for all $k\leq K_1$, the car of index $k$ will be moving at constant speed $q_\nbParam$ since time $0$, so that the left-hand side of formula~\eqref{eq:evolcanonicalstatioconfigbis} grows linearly in $t$ at rate $q_\nbParam$. Since $\widetilde{y}_\infty(t)$ is a stationary trajectory, it eventually grows linearly in $t$ at rate $q_\nbParam$, thus there exists $K_2\in\ZZ$ such that for all $k\leq K_2$, the right-hand side of formula~\eqref{eq:evolcanonicalstatioconfigbis} grows linearly in $t$ at rate $q_\nbParam$ since time $0$. We conclude that formula~\eqref{eq:evolcanonicalstatioconfigbis} holds for all $t\geq0$ and $k\leq \min(K_1,K_2)$. This provides the initialization of the induction.


Let $k\leq0$ so that formula~\eqref{eq:evolcanonicalstatioconfigbis} holds for $k-1$. Then the ``inverse'' function $t_{k-1}$ of the trajectory $y_{k-1}$ is given by
\[
y\mapsto
\begin{cases}
\widetilde t_\infty(y)+(k-1)\widetilde t_\infty(a_1) & \text{ if } y> \widetilde{y}_\infty(-(k-1)\widetilde{t}_\infty(a_1)) \\
0 & \text{ otherwise.}
\end{cases}
\]
Let $j_k$ be the largest $1\leq j\leq \nbParam$ such that $a_j\leq \widetilde{y}_\infty(-(k-1)\widetilde{t}_\infty(a_1))$. Then for every $1\leq j\leq j_k$ we have $-\widetilde t_\infty(a_j)-(k-1)\widetilde t_\infty(a_1)\geq0$. Applying formula~\eqref{eq:recursiveFormulaSpeed} to the trajectory $y_k$, we get that for all $t\geq0$,
\begin{equation}
y_k(t) = \widetilde{y}_\infty(-k\widetilde{t}_\infty(a_1)) + \sum_{j=1} ^{j_k} p_j t + \sum_{j=j_k+1} ^{\nbParam} p_j \Big(t-\widetilde t_\infty(a_j)-(k-1)\widetilde t_\infty(a_1)\Big)_+.
\end{equation}
Formula~\eqref{eq:recursiveFormulaSpeedTilde3} yields
\begin{align}
y_k(t) &=\widetilde{y}_\infty(-k\widetilde{t}_\infty(a_1)) + \sum_{j=1} ^{j_k} p_j t - \sum_{j=1} ^{j_k} p_j \Big(t-\widetilde t_\infty(a_j)-(k-1)\widetilde t_\infty(a_1)\Big)_+ + \widetilde{y}_\infty(t-k\widetilde{t}_\infty(a_1)) \\
&=\widetilde{y}_\infty(-k\widetilde{t}_\infty(a_1)) + \sum_{j=1} ^{j_k} p_j t - \sum_{j=1} ^{j_k} p_j (t-\widetilde t_\infty(a_j)-(k-1)\widetilde t_\infty(a_1)) + \widetilde{y}_\infty(t-k\widetilde{t}_\infty(a_1)).
\end{align}
The conclusion follows from the fact that formula~\eqref{eq:recursiveFormulaSpeedTilde3} also implies that
\begin{align}
\widetilde{y}_\infty(-k\widetilde{t}_\infty(a_1)) &= \sum_{j=1} ^{\nbParam} p_j \Big(-\widetilde t_\infty(a_j)-(k-1)\widetilde t_\infty(a_1)\Big)_+ \\
&= \sum_{j=1} ^{j_k} p_j (-\widetilde t_\infty(a_j)-(k-1)\widetilde t_\infty(a_1)).
\end{align}

Assume now that $k\geq1$ and that the trajectory $y_{k-1}(t)$ satisfies formula~\eqref{eq:evolcanonicalstatioconfig}. This implies that $y_{k-1}$ is a time shift of the trajectory $\widetilde{y}_\infty$ so that it starts at time $(k-1)\widetilde{t}_\infty(a_1)$. Its ``inverse'' function $t_{k-1}$ is such that for every $y>0$,
\[
t_{k-1}(y)=\widetilde t_\infty(y)+(k-1)\widetilde t_\infty(a_1)
\]
Thus formula~\eqref{eq:recursiveFormulaSpeed} applied to the trajectory $y_k$ yields that for all $t\geq0$,
\[
y_k(t) = \sum_{j=1} ^{\nbParam} p_j \Big(t-\widetilde t_\infty(a_j)-(k-1)\widetilde t_\infty(a_1)\Big)_+.
\]
This is equal to $0$ for $0\leq t\leq k\widetilde t_\infty(a_1)$ and equal by formula~\eqref{eq:recursiveFormulaSpeedTilde3} to $\widetilde{y}_\infty(t-k\widetilde t_\infty(a_1))$ for $t\geq k\widetilde t_\infty(a_1)$, hence formula~\eqref{eq:evolcanonicalstatioconfig} holds for $k$.
\end{proof}

One may rewrite the statement of Lemma~\ref{lem:stationarityofstationaryconfiguration} as follows. Let $\widetilde{y}_\infty$ be a stationary trajectory and define $k_0=\left\lfloor\tfrac{t}{\widetilde{t}_\infty(a_1)}\right\rfloor$ and $\tau=t-k_0\widetilde{t}_\infty(a_1)$. Then for every $k\geq0$ we have
\begin{equation}
y_{k_0-k}(t)= \widetilde{y}_\infty(k\widetilde{t}_\infty(a_1)+\tau).
\end{equation}
Furthermore $y_k(t)=0$ for every $k>k_0$. This motivates the following definition.

\begin{definition}
\label{def:stationaryconfiguration}
Let $\widetilde{y}_\infty$ be a stationary trajectory. We say that $(y_k(0))_{k \in \ZZ} \in \carConf$ is a \emph{stationary configuration} for $\widetilde{y}_{\infty}$ if there exist $k_0\in\ZZ$ and $\tau\in[0,\widetilde{t}_\infty(a_1))$ such that for every $k\geq0$, $y_{k_0-k}(0)= \widetilde{y}_\infty(k\widetilde{t}_\infty(a_1)+\tau)$ and for every $k>k_0$, $y_{k}(0)= 0$.
\end{definition}

The above reformulation of Lemma~\ref{lem:stationarityofstationaryconfiguration} entails that stationary configurations for $\widetilde{y}_{\infty}$ are obtained, up to a shift of the labels of the cars, by letting the car model flow starting from the canonical stationary configuration for $\widetilde{y}_{\infty}$. The trajectory of each car is then given, up to a shift in time, by the stationary trajectory $\widetilde{y}_{\infty}$. It will follow from the uniqueness of the stationary trajectory $\widetilde{y}_{\infty}$ proved in Theorem~\ref{theorem:convergenceOfDynamics} that the notion of stationary configuration does not depend on a choice of $\widetilde{y}_{\infty}$.

\subsection{Monotonicity property and existence of stationary trajectories}\label{subsec:monotonicityProp}

\begin{definition} We say that a sequence of functions $(\check y_k)_{k\geq0}$ from $\RR_{\geq0}$ to $\RR_{\geq0}$ satisfies \emph{recurrence (R)} if the following statements all hold true:
\begin{itemize}
 \item $\check y_0(0)=0$;
 \item the function $\check y_0$ is 
  continuous, piecewise linear and has at most $N$ points of non-differentiability;
 \item the right derivative of $y_0$ is non-decreasing and takes values in $\{q_1,\ldots,q_N\}$;
 \item for every $k\geq1$, denoting by $\check t_{k-1}$ the inverse function of $\check y_{k-1}$, we have
\begin{align}\label{eq:recursiveequation}
\check{y}_k(t) = \sum_{j=1}^\nbParam p_j (t  - \check{t}_{k-1}(a_j) + \check{t}_{k-1}(a_1) )_+ .
\end{align}
\end{itemize}
\end{definition}

To be a fully rigorous definition, one should check inductively that $\check y_k$ is indeed a bijection from $\RR_{\geq0}$ to $\RR_{\geq0}$. This holds true for $k=0$ by the first three assumptions of the definition, since it is continuous, increasing and takes value $0$ at $0$. Moreover, it is easy to see from \eqref{eq:recursiveequation} that $\check y_k$ takes value $0$ at $0$, is continuous and increasing. For functions $\check y_k$ satisfying recurrence (R), we shall always denote their inverse functions by $\check t_k$ and their right derivative functions by $\check v_k$. Right-differentiating \eqref{eq:recursiveequation}, we get that for all $t \in \RR_{\geq0}$ and $k \in \ZZ$,
\begin{align}\label{eq:recursiveFormulaSpeedTilde1}
\check{v}_k(t) = \sum_{j=1}^\nbParam p_j \mathds{1}_{t \geq \check{t}_{k-1}(a_j) - \check{t}_{k-1}(a_1)}.
\end{align}

It follows from \eqref{eq:recursiveFormulaSpeedTilde2} that the shifted car trajectories $(\bar y_k)_{k\geq k_0}$ satisfy recurrence (R) whenever $y_{k_0}(0)=0$. The following Lemma \ref{lemma:increasingDynamics} shows that recurrence (R) is monotonous in the right derivative $\check v_0$ of $\check y_0$. After that we will consider the behavior of a minimal and a maximal solution to recurrence (R) and deduce from it the convergence in time to a stationary configuration (Theorem \ref{theorem:convergenceOfDynamics}).

\begin{lemma}[Monotonicity property]\label{lemma:increasingDynamics}
    Let $(\check{y}_k^{(1)})_{k\geq0}$ and $(\check{y}_k^{(2)})_{k\geq0}$ be two sequences of functions satisfying recurrence (R), with inverses and right derivatives respectively denoted by $\check{t}_k^{(m)}$ and $\check{v}_k^{(m)}$ for $m=1,2$.
    Assume that 
\begin{align}\label{eq:formulaMonotonicity}
    \check{v}_0^{(1)} \circ  \check{t}_0^{(1)}  \leq \check{v}_0^{(2)} \circ \check{t}_0^{(2)}.
    \end{align} 
    Then for all $k \in \ZZ_{>0}$,     \begin{align}\label{eq:formulaMonotonicity2}
    \check{v}_k^{(1)} \circ  \check{t}_k^{(1)}  \leq \check{v}_k^{(2)} \circ  \check{t}_k^{(2)} . 
    \end{align}    
\end{lemma}

    In other words, if for all $x \in \RR_{>0}$, the speed of the car $0$, at the moment it reaches position $x$, is greater with initial condition $\check{v}_0^{(2)}$ than with initial condition $\check{v}_0^{(1)}$, then it 
    also holds for the $k$-th car for every $k \in \ZZ_{\geq0}$.
    
    In order to prove Lemma \ref{lemma:increasingDynamics}, let us first introduce some notations and make some useful remarks. 
    Consider $j \in \llbracket 1 , \nbParam+1 \rrbracket$, $m\in\{1,2\}$ and $k \in \ZZ_{\geq0}$.
 Set \[z_{k,j}^{(m)} =  \check{t}_k^{(m)}(a_{j}) - \check{t}_k^{(m)}(a_{j-1}).\] By definition, $z_{k,j}^{(m)}$ is the amount of time spent by the $k$-th car between positions $a_{j-1}$ and $a_{j}$, when the car of index $0$ has $\check{y}_0^{(m)}$ as a trajectory, with the conventions that $a_0 = 0$ and $a_{\nbParam +1}= +\infty$. By right differentiating $\check{t}_k^{(m)}$, we get
    \begin{align}\label{eq:formulaMonotonicity3}
        z_{k,j}^{(m)} = \int_{a_{j-1}}^{a_{j}} \frac{1}{\check{v}_k^{(m)}\circ \check{t}_k^{(m)}(s)}ds.
    \end{align}

    Notice that $z_{k,j+1}^{(m)}$ also corresponds to the amount of time during which the $(k+1)$-th car moves at speed $q_{j}$,  assuming that the car of index $0$ has trajectory  $\check{y}^{(m)}_0$. For all $j \in \llbracket 1 , \nbParam \rrbracket $ and $k \in \ZZ_{\geq0}$, set $\delta_{k,j}^{(m)} = q_j z_{k,j+1}^{(m)}$ to be the distance traveled by the $(k+1)$-th car at speed $q_j$ for the car model with initial car trajectory $\check{y}^{(m)}_0$. Note that $\delta_{k,\nbParam} = +\infty$ for all $k \in \ZZ_{\geq0}$. 

    Consider $x \in \RR_{\geq0}$. If $ \delta_{k,1}^{(m)} + \dots +\delta_{k,j-1}^{(m)} \leq  x < \delta_{k,1}^{(m)} + \dots +\delta_{k,j}^{(m)} $, then $\check{v}_{k+1}^{(m)}\circ \check{t}_{k+1}^{(m)} (x)  = q_j. $ Therefore, for all $x \in \RR_{\geq0}$,
    \begin{align}\label{eq:formulaMonotonicity4}
         \check{v}_{k+1}^{(m)}\circ \check{t}_{k+1}^{(m)}(x) = q_{j_{k}^{(m)}(x)},  
    \end{align} 
    where $j_{k}^{(m)}(x) := \min\Set{j\in \llbracket 1, \nbParam \rrbracket }{x < \delta_{k,1}^{(m)} + \dots +\delta_{k,j}^{(m)}}$.

\begin{proof}[Proof of Lemma \ref{lemma:increasingDynamics}]
    Let us prove \eqref{eq:formulaMonotonicity2} by induction on $k$. The case $k = 0$  corresponds to assumption \eqref{eq:formulaMonotonicity}. 
    Now, let us assume that \eqref{eq:formulaMonotonicity2} holds true for some $k \in \ZZ_{\geq0}$. 
    It follows from \eqref{eq:formulaMonotonicity3} and the induction hypothesis that for all $j\in\llbracket 1 , \nbParam \rrbracket $, $z_{k,j+1}^{(2)} \leq z_{k,j+1}^{(1)}$. 
    Therefore $\delta_{k,j}^{(2)} \leq \delta_{k,j}^{(1)}$ for all $j\in\llbracket 1 , \nbParam \rrbracket $, hence $j_{k}^{(1)}(x) \leq j_{k}^{(2)}(x)$ for all $x \in \RR_{\geq0}$. Since the sequence $\underline q$ is increasing, it follows from \eqref{eq:formulaMonotonicity4} that for all $x \in \RR_{\geq0}$ 
    \begin{align*}
    \check{v}_{k+1}^{(1)}\circ \check{t}_{k+1}^{(1)}(x) \leq \check{v}_{k+1}^{(2)}\circ \check{t}_{k+1}^{(2)}(x),  
    \end{align*}
    which concludes the proof of the inductive step.
\end{proof}

The following proposition is a consequence of Lemma \ref{lemma:increasingDynamics}. It provides a way to compare the trajectories of the $k$-th car for two initial  trajectories satisfying some assumption. 

\begin{proposition}\label{proposition:increasingDynamics}
    We keep the notations of Lemma \ref{lemma:increasingDynamics} and assume that \eqref{eq:formulaMonotonicity} holds. Then, for all $k \in \ZZ_{>0}$,
\begin{align}\label{eq:formulaMonotonicity5}
       \check{y}_k^{(1)} \leq \check{y}_k^{(2)}.
\end{align}
\end{proposition}

\begin{proof}
    By Lemma \ref{lemma:increasingDynamics} and \eqref{eq:formulaMonotonicity3}, we have that for all $k \in \ZZ_{>0}$ and $j \in \llbracket 2 , \nbParam \rrbracket$  \[  \check{t}_k^{(2)} (a_{j}) - 
    \check{t}_k^{(2)}(a_{1}) = \int_{a_1}^{a_j} \frac{1}{\check{v}_k^{(2)}\circ \check{t}_k^{(2)}(s)}ds \leq \int_{a_1}^{a_j} \frac{1}{\check{v}_k^{(1)}\circ \check{t}_k^{(1)}(s)}ds =
    \check{t}_k^{(1)}(a_{j}) - 
    \check{t}_k^{(1)}(a_{1}). \]
    Inequality \eqref{eq:formulaMonotonicity5} follows from formula \eqref{eq:recursiveequation}.
\end{proof}

The following two Propositions \ref{proposition:boundingTrajectories0} and \ref{proposition:boundingTrajectories} are corollaries of Proposition \ref{proposition:increasingDynamics}.
Recall that for every $i\in\llbracket 1,N\rrbracket$, $d_i=a_i-a_{i-1}$. Define \begin{align}\label{def:SpeedMajMin}
    {v}^{-}_0(t) := p_1 + \sum_{i=2}^{\nbParam} p_i \mathds{1}_{\left( t - d_1 q_1^{-1} \geq d_2 q_1^{-1} + d_3 q_2^{-1} + \dots + d_{i} q_{i-1}^{-1} \right)} \quad \text{ and }  \quad {v}^{+}_0(t) :=  q_\nbParam.
\end{align} Observe that ${v}^{-}_0$ corresponds to the speed of a car that moves at speed $q_1$ between positions $0$ and $a_1$ and at speed $q_i$ between positions $a_i$ and $a_{i+1}$ for all $i\in \llbracket 1 , \nbParam \rrbracket$. Set also ${y}^{\pm}_0(t) := \int_{0}^{t} {v}^{\pm}_0(s)ds$ for all $t \in \RR_{\geq0}$, where $\pm$ can be replaced either by $-$ or $+$. Finally define ${t}^{\pm}_0$ to be the inverse function of ${y}^{\pm}_0$. Then for all $x \geq 0$,  
\begin{align}\label{def:SpeedMajMin2}
    \left({v}^{-}_0 \circ  {t}^{-}_0\right)(x) =  p_1 + \sum_{i=2}^{\nbParam} p_i \mathds{1}_{ x \geq a_i  } \quad \text{ and }  \quad \left({v}^{+}_0 \circ  {t}^{+}_0\right) (x) =  q_\nbParam.
\end{align}

Let $({y}_k^{\pm})_{k \geq0}$ be the sequences satisfying recurrence (R) with first terms equal ${y}_0^{\pm}$. Set $({t}_k^{\pm})_{k \geq0}$ to be their inverse functions and $({v}_k^{\pm})_{k \geq0}$ to be their right derivatives.

\begin{proposition}[bounding trajectories]\label{proposition:boundingTrajectories0}
    Consider an initial configuration $y(0) \in \carConf$ and set $y(t) = \carDyn{t}{y(0)}$ for every $t\in\RR_{>0}$. Consider  $(\bar{y}_k)_{k \in \ZZ}$ as previously defined in \eqref{eq:defTrajTimeRescaled}. Set $\bar{v}_k$ to be the right derivative of $ \bar{y}_k $ for all $k \in \ZZ$. Assume that for some ${k_0} \in \ZZ$, $y_{k_0}(0)=0$. Then, for all $k \geq {k_0}$,
    \begin{align}\label{eq:encadrementBounds}
     {y}_{k-{k_0}}^{-} \leq \bar{y}_{k} \leq {y}_{k-{k_0}}^{+}.
 \end{align}
\end{proposition}

\begin{proof}

No matter what the initial configuration $y(0)$ is, when a car has position greater than $a_i$, its speed is at least $q_i$. Moreover, the speed of a car in motion is always at least $q_1$ and at most $q_\nbParam$. As a consequence, $ {v}_0^{-} \circ {t}_0^{-} \leq \bar{v}_k \circ \bar{t}_k \leq {v}_0^{+} \circ {t}_0^{+}$ for all $k \in \ZZ$.

Fix ${k_0}$ such that $y_{k_0}(0)=0$.
Since all three sequences $(\bar{y}_{k})_{k \geq {k_0}}$ and $({y}_{k-{k_0}}^{\pm})_{k \geq {k_0}}$ satisfy the recursive formula \eqref{eq:recursiveequation}, it follows from Proposition \ref{proposition:increasingDynamics} that for the bounds \eqref{eq:encadrementBounds} hold for all $k \geq {k_0}$.
\end{proof}

\begin{proposition}[Monotonicity and asymptotics of the bounds]\label{proposition:boundingTrajectories}
    For all $k \in \ZZ_{\geq0}$,
\begin{align}\label{eq:monotonicityOfBounds}
    {y}_k^{-} \leq {y}_{k+1}^{-} \quad \text{ and } \quad {y}_{k+1}^{+} \leq {y}_k^{+}.
\end{align}
Denote by ${y}_{\infty}^{\pm}$ the pointwise limit of $({y}_k^{\pm})_k$ as $k$ tends to infinity. Then both ${y}_{\infty}^{-}$ and ${y}_{\infty}^{+}$ are stationary trajectories.
\end{proposition}

\begin{proof}
By Proposition \ref{proposition:increasingDynamics}, it suffices to show that $ {v}_0^{-} \circ {t}_0^{-} \leq {v}_1^{-} \circ {t}_1^{-}$ and ${v}_1^{+} \circ {t}_1^{+} \leq {v}_0^{+} \circ {t}_0^{+}$. Since ${v}_0^{+} = q_{\nbParam}$ and ${v}_1^{+} \circ {t}_1^{+}$ takes values in $\{q_1, \dots , q_{\nbParam}\}$, the second point is straightforward. 

Let us turn to the first point. Since ${y}_1^{-}$ is constructed according to the dynamics of the car model, when the corresponding car is at a position $x$ in $(a_i, a_{i+1}]$ for $i\in \llbracket 1 , \nbParam \rrbracket$, it moves at speed at least $q_i={v}_0^{-} \circ {t}_0^{-}(x)$. Since this car starts from $0$, when it is at a position $x$ in $(0,a_1]$, its speed is at least $q_1={v}_0^{-} \circ {t}_0^{-}(x)$. This proves the first point.

    Therefore, by Proposition \ref{proposition:increasingDynamics} and by induction on $k \in \ZZ_{\geq0}$,
\begin{align*}
    {y}_k^{-} \leq {y}_{k+1}^{-} \quad \text{ and } \quad {y}_{k+1}^{+} \leq {y}_k^{+}.
\end{align*}

By \eqref{eq:monotonicityOfBounds}, $({y}_k^{\pm})_{k\geq0}$ converge to trajectories ${y}_{\infty}^{\pm}$ by monotonicity. Since $({y}_k^{\pm})_{k \geq0}$ satisfy formula \eqref{eq:recursiveequation}, it implies that ${y}_{\infty}^{\pm}$ are stationary trajectories.
\end{proof}

    At this point, there is no reason for the two stationary trajectories  ${y}_{\infty}^{\pm}$ to be equal. In Subsection \ref{subsec:contractivity}, we prove that they are equal, by uniqueness of the stationary trajectory, and that the convergence in time starting from any initial trajectory to that stationary trajectory occurs exponentially fast.  

\subsection{Uniqueness of the stationary trajectory and exponential convergence}\label{subsec:contractivity}

The main result of this subsection is the following theorem, the proof of which is completed at the end of this section:
\begin{theorem}\label{theorem:convergenceOfDynamics}
    For any $(\underline a,\underline p)\in P^{2N}$, there is a unique stationary trajectory $\widetilde{y}_{\infty}$ for the car model. Moreover, there exists a constant $\kappa > 0$ such that for any initial configuration $y(0) \in \carConf$, any $k_0 \in\ZZ$ such that $y_{k_0}(0) = 0$ and any $k \in \ZZ_{\geq 0}$, 
\begin{equation}
\label{eq:exponentialconvergence}
 \| \bar{y}_{k+k_0} - \widetilde{y}_{\infty}\|_{\infty} \leq \kappa \cdot \left(1 - \frac{q_1}{q_{\nbParam}}\right)^k,  
 \end{equation}
 with $1-\tfrac{q_1}{q_N}\in[0,1)$.
\end{theorem}

Consider $y(0)\in\carConf$ and set $y(t)=\Psi(t,y(0))$ for all $t\geq0$. To alleviate notation, we will assume without loss of generality that $k_0=0$, namely that $y_0(0)=0$. Recall that the shifted trajectories $\bar y_k$ for $k\geq0$ satisfy the recursive formula \eqref{eq:recursiveFormulaSpeedTilde2}. Denoting the inverse function of each $\bar y_k$ by $\bar t_k$ and applying formula \eqref{eq:recursiveFormulaSpeedTilde2} at time $t=\bar t_k(a_i)$, we get

\begin{align}
\label{eq:systemtranscar}
\forall i\in\llbracket 1 , \nbParam \rrbracket , \,a_i = \sum_{j=1}^\nbParam p_j ( \bar{t}_k(a_i)  - \bar{t}_{k-1}(a_j) + \bar{t}_{k-1}(a_1) )_+ .
\end{align}
Knowing $(\bar{t}_{k-1}(a_i))_{i \in \llbracket 1 , \nbParam \rrbracket}$, it is easy to deduce $(\bar{t}_k(a_i))_{i \in \llbracket 1 , \nbParam \rrbracket}$ from \eqref{eq:systemtranscar} if we know for which values of $(i,j)$ we have 
\begin{align}
\label{eq:positivepart}
\bar{t}_k(a_i) > \bar{t}_{k-1}(a_j) - \bar{t}_{k-1}(a_1).
\end{align}
The car of index $k$ starts moving when the car of index $k-1$ reaches position $a_1$. From that moment on, the car of index $k-1$ takes an amount of time $\bar{t}_{k-1}(a_j) - \bar{t}_{k-1}(a_1)$ to reach position $a_j$. Inequality \eqref{eq:positivepart} is equivalent to requiring that the car of index $k$ has not yet reached position $a_i$ when the car of index $k-1$ is at position $a_j$. For a fixed $i\in\llbracket 1 , \nbParam \rrbracket$, the collection of values $j$ such that \eqref{eq:positivepart} holds is thus some interval of the form $\llbracket 1 , j_{\max}\rrbracket$ for some $j_{\max}\in \llbracket i, N \rrbracket$. In particular, if \eqref{eq:positivepart} holds for some $1 \leq i < j \leq \nbParam$, then 
\[\bar{t}_k(a_{i'})  - \bar{t}_{k-1}(a_{j'}) + \bar{t}_{k-1}(a_1) \geq 0\]
for all $i \leq {i'} <  {j'} \leq  j $ by monotonicity of $\bar{t}_k$ and $\bar{t}_{k-1}$. We shall encode by a graph all the pairs $(i,j)$ with $i<j$ such that \eqref{eq:positivepart} holds.

Recall from Subsection~\ref{subsec:LBMproperties} the definition of downward closed graphs (DC graphs) and the notations $\DC_N$ and $b_G(i)$. Let $G_k$ be the graph on the vertex set $\llbracket 1 , \nbParam \rrbracket$ with edge set
\[E(G_{k}) = \Set{(i,j)}{1 \leq i < j \leq \nbParam ,\,  \bar{t}_k(a_i) - \bar{t}_{k-1}(a_j) + \bar{t}_{k-1}(a_1)  > 0}\]
Then $G_k\in \DC_N$. Formula \eqref{eq:positivepart} holds if either $i\geq j$ or if $(i,j)$ is an edge of $G_k$.

For any $G\in \DC_N$, set $T(G)$ to be the map from $\RR^{\nbParam}$ to $\RR^{\nbParam}$ such that \[ 
\forall s = (s(1), \dots, s(\nbParam)) \in \RR^{\nbParam}, \, \forall i \in \llbracket 1 , \nbParam \rrbracket , \, T(G)(s)(i) := \frac{1}{q_{b_G(i)}} \left( a_i + \sum_{j = 1}^{b_G(i)} p_j( s(j) - s(1)) \right). 
\]
It follows from \eqref{eq:systemtranscar} that the sequence $((\bar{t}_{k}(a_i))_{i \in \llbracket 1 , \nbParam \rrbracket })_{k\geq0}$ is recursively defined by 
\begin{align}\label{eq:convergenceOfDyn1}
    \forall k \in \ZZ_{>0}, (\bar{t}_k(a_i))_{i \in \llbracket 1 , \nbParam \rrbracket}=\,T(G_k)((\bar{t}_{k-1}(a_i))_{i \in \llbracket 1 , \nbParam \rrbracket }).
\end{align}

The following lemma states that for a fixed  DC graph $G$, the map $T(G)$ mainly  behaves like a contraction. 
It is a first step in the proof of Theorem \ref{theorem:convergenceOfDynamics}.

\begin{lemma}\label{lemma:contractiveMap}
    Consider $s^{(m)}= (s^{(m)}(i))_{i \in \llbracket 1 , \nbParam \rrbracket} \in (\RR_{\geq0})^{\nbParam}$ with $m \in \{0,1\}$. 
    Assume that $0 < s^{(0)}(i)- s^{(0)}(i-1) \leq s^{(1)}(i)- s^{(1)}(i-1)$ for all $i \in \llbracket 1 , \nbParam  \rrbracket$, with the convention that $s^{(m)}(0)=0$. 
    Then for all $G\in \DC_N$ and $i \in \llbracket 1 , \nbParam  \rrbracket$, \[T(G)(s^{(0)})(i) \leq T(G)(s^{(1)})(i).\]   
    Moreover, \[ \|T(G)(s^{(1)}) - T(G)(s^{(0)}) \|_{\infty} \leq \left(1 - \frac{q_1}{q_{\nbParam}}\right) \|s^{(1)} - s^{(0)} \|_{\infty}. \] 
\end{lemma}
\begin{proof}
    Consider $i  \in \llbracket 1 , \nbParam \rrbracket$. By definition of $T(G)$, 
    \begin{align*}
    T(G)(s^{(1)})(i) - T(G)(s^{(0)})(i) = \frac{1}{q_{b_{G}(i)}}  \sum_{j=1}^{b_{G}(i)} p_j\left[ s^{(1)}(j) -  s^{(0)}(j) - (s^{(1)}(1) -  s^{(0)}(1))  \right]. 
    \end{align*}
    By assumption on $s^{(0)}$ and $s^{(1)}$, $s^{(0)}(j)- s^{(0)}(1)  \leq s^{(1)}(j)- s^{(1)}(1) $ for all $j \in \llbracket 1 , \nbParam \rrbracket$. Therefore, we obtain the first part of the lemma:
    \begin{align*}
    0 \leq T(G)(s^{(1)})(i) - T(G)(s^{(0)})(i). 
    \end{align*}
    Moreover, notice that the term for $j = 1$ vanishes. We assumed that $s^{(0)}(1) \leq s^{(1)}(1)$, thus 
    \begin{align*}
    0 \leq T(G)(s^{(1)})(i) - T(G)(s^{(0)})(i) \leq \frac{1}{q_{b_{G}(i)}}  \sum_{j=2}^{b_{G}(i)} p_j\left[ s^{(1)}(j) -  s^{(0)}(j)  \right]. 
    \end{align*}

Since $ \frac{q_{b_G(i)} - q_1}{q_{b_G(i)}} \leq1-\frac{q_1}{q_{N}}$, we conclude that 
\begin{align*}
     \| T(G)(s^{(1)}) - T(G)(s^{(0)})\|_{\infty} \leq \left(1-\frac{q_1}{q_{N}}\right) \| s^{(1)} -  s^{(0)}  \|_{\infty}. 
\end{align*} 
\end{proof}

Lemma \ref{lemma:contractiveMap} holds for a fixed graph $G$. Since the graph $G_k$ used to construct $(\bar{t}_k(a_i))_k$ via the recursive formula \eqref{eq:convergenceOfDyn1} depends on $k$, we need to prove a stronger version of Lemma \ref{lemma:contractiveMap}. For $m \in \{0,1\}$ and $s^{(m)} \in (\RR_{>0})^{\nbParam}$ such that $0 < s^{(m)}(1) < \dots < s^{(m)}(\nbParam)$, let $G^{(m)}$ be the DC graph with $\nbParam$ vertices such that 
\[ E(G^{(m)}) = \Set{(i,j)}{1 \leq i < j \leq \nbParam ,\,  t^{(m)}(a_i) - s^{(m)}(j) + s^{(m)}(1)  > 0}, \]
where $t^{(m)}$ is the inverse function of 
\begin{equation}
\label{eq:ym}
y^{(m)} : t \in \RR_{\geq0} \mapsto \sum_{j = 1 }^{\nbParam} p_j (t - s^{(m)}(j) + s^{(m)}(1))_+.
\end{equation}

For all $i \in \llbracket 1 , \nbParam \rrbracket$, by evaluating the function $ y^{(m)} $ at $ t^{(m)}(a_i) $ and using the definition of $G^{(m)}$, we get
\begin{equation}
\label{eq:TG=t}
t^{(m)}(a_i) = T(G^{(m)})(s^{(m)})(i).
\end{equation}

\begin{lemma}\label{lemma:contractiveMapImproved}
    Consider $(s^{(m)}(i))_{i \in \llbracket 1 , \nbParam \rrbracket} \in (\RR_{>0})^{\nbParam}$ with $m \in \{0,1\}$. 
    Assume that $0 < s^{(0)}(i)- s^{(0)}(i-1) \leq s^{(1)}(i)- s^{(1)}(i-1)$ for all $i \in \llbracket 1 , \nbParam  \rrbracket$, with the convention that $s^{(m)}(0)=0$. Then, 
    \[ \|T(G^{(1)})(s^{(1)}) - T(G^{(0)})(s^{(0)}) \|_{\infty} \leq \left(1 - \frac{q_1}{q_{\nbParam}}\right) \|s^{(1)} - s^{(0)} \|_{\infty}. \] 
\end{lemma}

\begin{proof}
    For all $\tau \in (0,1)$, set $s^{(\tau)} := \tau s^{(1)} + (1-\tau) s^{(0)}$ and  
    \begin{equation}
    \label{eq:defytau}
y^{(\tau)} : t \in \RR_{\geq0} \mapsto \sum_{j = 1 }^{\nbParam} p_j (t - s^{(\tau)}(j) + s^{(\tau)}(1))_+.
\end{equation}
Denote the inverse function of  $y^{(\tau)}$ by $t^{(\tau)}$. Let $G^{(\tau)}$ be the DC graph with $\nbParam$ vertices such that 
\[ E(G^{(\tau)}) = \Set{(i,j)}{1 \leq i < j \leq \nbParam ,\,  t^{(\tau)}(a_i) - s^{(\tau)}(j) + s^{(\tau)}(1)  > 0} \]
For all $i \in \llbracket 1 , \nbParam \rrbracket$, by evaluating the function $ y^{(\tau)} $ at $ t^{(\tau)}(a_i) $ and using the definition of $G^{(\tau)}$, we get
\begin{equation}
\label{eq:TGtau=t}
t^{(\tau)}(a_i) = T(G^{(\tau)})(s^{(\tau)})(i).
\end{equation}

Let us first prove that $ \tau \in [0,1] \mapsto T(G^{(\tau)})(s^{(\tau)})$ is continuous. By \eqref{eq:TGtau=t} and since $t^{(\tau)}$ is the inverse function of $y^{(\tau)}$, it suffices to prove that $\varphi : (\tau,t)\mapsto(\tau, y^{(\tau)}(t))$ is a homeomorphism from $[0,1]\times \RR_{\geq 0} $ to itself. The continuity of $\varphi$ follows from the definition \eqref{eq:defytau} of $y^{(\tau)}(t)$. Since $ t \mapsto y^{(\tau)}(t) $ is continuous, strictly increasing, bounded below by $q_1t$ and equal to $0$ at $t=0$ for all $\tau \in [0,1]$, $\varphi$ is invertible. It remains to prove that $\varphi^{-1}$ is continuous. Consider $M>0$. Since $ y^{(\tau)}(t) $ is bounded below by $q_1 t$ for all $t \in\RR_{\geq 0}$ and $\tau \in [0,1]$, the image of $[0,1]\times [0,M]$ by $\varphi$ contains $[0,1]\times [0,q_1 M]$. Since $[0,1]\times [0,M]$ is a compact, the restriction of $\varphi$ to this set is a closed function. In particular, this implies that the restriction of $\varphi^{-1}$ to $[0,1]\times [0,q_1 M]$ is continuous. Since this result holds for all $M>0$, $\varphi^{-1}$ is continuous. 

Next, let us show that there exist $r\geq0$ and $0=\tau_0 < \tau_1 < \dots < \tau_r < \tau_{r+1}=1$ such that $G^{(\tau)}$ is constant on each interval of the form $(\tau_h,\tau_{h+1})$.
If $\hat\tau\in(0,1)$ is a value at which $G^{(\tau)}$ is not constant on an open neighborhood of $\hat\tau$, then there exist two distinct DC graphs $G_0$ and $G'_0$ which are accumulation points of $\tau \mapsto G^{(\tau)}$ at $\hat\tau$. 

Consider $(u_n)_{n\geq 1}$ and $(u_n')_{n\geq 1}$ two sequences of elements of $[0,1]$ both converging to $\hat{\tau}$ such that $G^{(u_n)} = G_0$ and  $G^{(u_n')} = G'_0$ for all $n\geq 1$.

For every $1\leq i < j \leq N$, $\tau \in [0,1]$ and $G\in \DC_N$, set \[f_{G,i,j}(\tau) := T(G)(s^{(\tau)})(i) - s^{(\tau)}(j) + s^{(\tau)}(1) .\] 

Note that $f_{G^{(u_n)},i,j}(u_n) = f_{G_0,i,j}(u_n)$ and $f_{G^{(u_n')},i,j}(u_n') = f_{G'_0,i,j}(u_n')$ for all $n \geq 1$. Moreover, for any fixed $(i,j) \in E(G_0)\Delta E(G'_0)$, of the two quantities $f_{G_0,i,j}(u_n)$ and $f_{G'_0,i,j}(u_n')$, one is positive while the other is non-positive. Without loss of generality, let us assume that $f_{G_0,i,j}(u_n)>0$ and $f_{G'_0,i,j}(u_n')\leq0$.  By continuity of the maps $\tau \mapsto T(G^{(\tau)})(s^{\tau})$ and $\tau \mapsto s^{\tau}$, the map $ \tau \mapsto f_{G^{(\tau)},i,j}(\tau)$ is also continuous. 
As a consequence, 
\begin{equation}\label{eq:equalZeroBoundaryDyn}
    \forall (i,j) \in E(G_0)\Delta E(G'_0) ,\, f_{G_0,i,j}(\hat{\tau})  = f_{G'_0,i,j}(\hat{\tau}) = 0 .
\end{equation}

The map $\tau\mapsto f_{G_0,i,j}(\tau)$ is linear in $\tau$, vanishes at $\hat\tau$ and is positive along the sequence $u_n$, hence this linear map is non-constant and one can express its root $\hat\tau$ as
\[ \hat{\tau} =  \frac{f_{G_0,i,j}(0)}{f_{G_0,i,j}(0) - f_{G_0,i,j}(1)}. \]
Since there is a finite number of DC graphs $G$ with $N$ vertices, and a finite number of $(i,j)$ with $1\leq i < j \leq N$, there is a finite number of possible $\hat{\tau}\in[0,1]$ such that $\tau \mapsto G^{(\tau)}$ has more than one accumulation point around  $\hat{\tau}$. 

By the triangle inequality, we have
\begin{equation}\label{eq:proofContractionImproved1}
        \|T(G^{(1)})(s^{(1)}) - T(G^{(0)})(s^{(0)}) \|_{\infty} \leq \sum_{h=0}^r  \|T(G^{(\tau_{h+1})})(s^{(\tau_{h+1})}) - T(G^{(\tau_{h})})(s^{(\tau_{h})}) \|_{\infty}.
    \end{equation}
     Now, for $h \in \llbracket 0, r \rrbracket$, set $G_h$ to be the DC graph equal to $G^{(\tau)}$ for all $\tau \in ( \tau_h, \tau_{h+1})$. By continuity of $ \tau \mapsto T(G^{(\tau)})(s^{(\tau)})$,  for all $h \in \llbracket 0 , r \rrbracket$, 
     \[ T(G^{(\tau_{h})})(s^{(\tau_{h})}) =  T(G_{h-1})(s^{(\tau_{h})}) = T(G_h)(s^{(\tau_{h})}). \]
    
     As a consequence, by \eqref{eq:proofContractionImproved1}, 
     \begin{align*}
         \|T(G^{(1)})(s^{(1)}) - T(G^{(0)})(s^{(0)}) \|_{\infty} \leq \sum_{h=0}^r  \|T(G_h)(s^{(\tau_{h+1})}) - T(G_h)(s^{(\tau_{h})}) \|_{\infty}.
     \end{align*}
     By Lemma \ref{lemma:contractiveMap}, 
     \begin{align*}
         \|T(G^{(1)})(s^{(1)}) - T(G^{(0)})(s^{(0)}) \|_{\infty} \leq \left(1-\frac{q_1}{q_{N}}\right) \sum_{h=0}^r  \|s^{(\tau_{h+1})} - s^{(\tau_{h})} \|_{\infty}.
     \end{align*}
     By collinearity of the $s^{(\tau_{h+1})} - s^{(\tau_{h})}$ for $h \in \llbracket 0 , r \rrbracket$, \[ \| s^{(\tau_{h+1})} - s^{(\tau_{h})} \|_{\infty} = (\tau_{h+1}-\tau_h)\| s^{(1)} - s^{(0)} \|_{\infty} .\] Thus \begin{align*}
         \|T(G^{(1)})(s^{(1)}) - T(G^{(0)})(s^{(0)}) \|_{\infty} \leq \left(1-\frac{q_1}{q_{N}}\right) \|s^{(1)} - s^{(0)} \|_{\infty}.
     \end{align*}
   
\end{proof}

Now, let us complete the proof of Theorem \ref{theorem:convergenceOfDynamics}. Recall that $y_0(0)=0$. By Proposition \ref{proposition:boundingTrajectories0}, 
\begin{align*}
     y_k^{-} \leq \bar{y}_{k} \leq y_k^{+}.
 \end{align*}
Consider $ \widetilde{y}_{\infty} $ any stationary trajectory, which exists by Proposition \ref{proposition:boundingTrajectories} and satisfies 
\begin{align*}
    y_k^{-} \leq \widetilde{y}_{\infty} \leq y_k^{+}.
\end{align*} 
As a consequence, 
\begin{align*}
    \|\bar{y}_{k} - \widetilde{y}_{\infty} \|_{\infty} \leq \|y^{+}_{k} - y^{-}_{k} \|_{\infty}.
\end{align*} 

By \eqref{eq:recursiveequation}, for all $t \in \RR_{\geq0}$ and $k \geq1$,
\begin{align*}
    0 \leq {y}^{+}_{k}(t) - {y}^{-}_{k}(t) &= \sum_{j=1}^\nbParam p_j \left[(t  - {t}^{+}_{k-1}(a_j) + {t}^{+}_{k-1}(a_1) )_+  - (t  - {t}^{-}_{k-1}(a_j) + {t}^{-}_{k-1}(a_1) )_+\right].
\end{align*}
Therefore, 
\begin{align*}
    \| y^{+}_{k} - y^{-}_{k} \|_\infty &\leq q_{\nbParam} \max_{j \in \llbracket 1 , \nbParam \rrbracket} | t^{-}_{k-1}(a_j) - t^{-}_{k-1}(a_1) - t^{+}_{k-1}(a_j) + t^{+}_{k-1}(a_1) | \\
    &\leq 2 q_{\nbParam} \max_{j \in \llbracket 1 , \nbParam \rrbracket} | t^{-}_{k-1}(a_j) - t^{+}_{k-1}(a_j) |.
\end{align*}

Now, let us apply Lemma \ref{lemma:contractiveMapImproved} with $s^{(0)} := ( t^{+}_{k-1}(a_i) )_{i \in \llbracket 1 , \nbParam \rrbracket} $ and $s^{(1)} := ( t^{-}_{k-1}(a_i) )_{i \in \llbracket 1 , \nbParam \rrbracket} $ (note that the assumptions of Lemma \ref{lemma:contractiveMapImproved} are satisfied by Proposition~\ref{proposition:boundingTrajectories0}). Writing $G^+:=G^{(0)}$ and $G^-:=G^{(1)}$, we have that for all $i \in \llbracket 1, \nbParam \rrbracket $, 
\[T(G^{\pm})((t^{\pm}_{k-1}(a_j))_{j \in \llbracket 1 , \nbParam \rrbracket}))(i) = t_{k}^{\pm}(a_i).\]
Hence \[\max_{j \in \llbracket 1 , \nbParam \rrbracket} | {t}^{-}_{k-1}(a_j) - {t}_{k-1}^{+}(a_j) | \leq \left(1-\frac{q_1}{q_N}\right)^{k-1} \max_{j \in \llbracket 1 , \nbParam \rrbracket}| {t}^{+}_0(a_j) - {t}^{-}_0(a_j) | .\]
When $N=1$, equation~\eqref{eq:recursiveFormulaSpeedTilde2} implies that there is a unique stationary trajectory given by $t\mapsto q_1t$. Otherwise, if $N\geq2$, setting $\kappa = \frac{2 q_{\nbParam}^2}{q_N-q_1}\max_{j \in \llbracket 1 , \nbParam \rrbracket}| {t}^{+}_0(a_j) - {t}^{-}_0(a_j) | $ we obtain \eqref{eq:exponentialconvergence}. The uniqueness of the stationary trajectory $\widetilde y_\infty$ follows from the uniqueness of the limit. This concludes the proof of Theorem \ref{theorem:convergenceOfDynamics}.

Let us now prove Theorem~\ref{theorem:mainstatio} using Theorem~\ref{theorem:convergenceOfDynamics}.

Let $\tilde{\tilde y}_\infty$ denote the stationary evolution of the car model starting at time $0$ from the canonical stationary configuration. The position of the $k$-th car at time $t$ in the stationary evolution is denoted by $\tilde{\tilde y}_{\infty,k}(t)$.

The existence and uniqueness of the stationary evolution $\tilde x_\infty$ for the bin model follow from the fact that the evolutions of the bin model and of the car model are conjugated by the mapping $\Sigma$, and $\tilde x_\infty(t)$ is simply given by $\Sigma^{-1}(\tilde{\tilde y}_\infty(t))$ for every $t\geq0$. We denote by $\tilde x_{\infty,k}(t)$ the content at time $t$ of the bin of index $k$ in the stationary evolution $\tilde x_\infty$.

Let $x(0)\in\binConf$ be an arbitrary initial bin configuration and let $y(0)=\Sigma(x(0))\in\carConf$. Up to shifting the origin of time and the labels of the cars, one may assume that at time $0$ all of the following conditions hold true for $y(0)$:
\begin{itemize}
 \item for every $k\geq0$ we have $y_k(0)=0$ ;
 \item $y_{-1}(0)=a_1$ ;
 \item for every $k\leq -1$ such that $y_{k+1}(0)\leq a_N$, we have $y_k(0)-y_{k+1}(0)\geq a_1$.
\end{itemize}

In terms of the bin model, it corresponds to having at time $0$ the front bin at position $-1$ and containing $a_1$ units of liquid, and to having the first $a_N$ units of liquid distributed in bins containing at least $a_1$ units of liquid each. These three conditions are also satisfied by the canonical stationary configuration $\tilde{\tilde y}_\infty(0)$.

Set $M:=\lceil \tfrac {a_N}{a_1}\rceil$. Then the first $a_N$ units of liquid in $x(t)$ are contained in the front $M+1$ non-empty bins of $x(t)$ for all $t\geq0$. Define $\proj(x(t))$ as the unique vector $(x'_{-M}(t),\ldots,x'_{0}(t))$ in $\RR_{>0}^{M+1}$  corresponding to the front $M+1$ non-empty bins of $x(t)$.

Recall that for every $k\geq0$, $T_k$ denotes the time when the car of index $k$ in $y$ starts moving. In terms of the bin model, $T_k$ represents the time when the front of $x$ arrives at position $k$, i.e. when the bin of index $k$ starts receiving liquid. We shall consider the evolution of the bin model on each interval $I_k:=(T_{k},T_{k+1}]$, which corresponds to the time when the front is at position $k$. For the stationary evolution $\tilde x_\infty$, the time spent by any bin of nonnegative index as the front bin is equal to $\tilde t_\infty(a_1)$.

We will prove the following two statements, which constitute the precise formulation of the convergence statement in Theorem~\ref{theorem:mainstatio}:
\begin{enumerate}[label=\Alph*.]
 \item The time $T_{k+1}-T_k$ spent by the front at position $k$ converges exponentially fast to $\tilde t_\infty(a_1)$. More precisely,
 \begin{equation}
\label{eq:exponentialcvperiods}
\left|T_{k+1}-T_k-\tilde t_\infty(a_1)\right|\leq \frac{\kappa}{q_1}  \left(1 - \frac{q_1}{q_{\nbParam}}\right)^k.
\end{equation}
 \item The projection $\proj(x(t+T_k))$ of $x$ for $t\in(0,T_{k+1}-T_k]$ converges exponentially fast to the projection $\proj(\tilde x_\infty(t))$ of $\tilde x_\infty(t)$ for every $t\in(0,\min(T_{k+1}-T_k,\tilde t_\infty(a_1))]$. More precisely,
 \begin{equation}
\sup_{\substack{0\leq \ell\leq M \\ t\in(0,\min(T_{k+1}-T_k,\tilde t_\infty(a_1))]}}\left|x'_{-\ell}(t+T_k)-\tilde x_{\infty,-\ell}(t)\right|\leq 2\kappa\left(\frac{M q_N}{q_1}+1\right)  \left(1 - \frac{q_1}{q_{\nbParam}}\right)^k.
 \end{equation}
\end{enumerate}

Recall that $T_{k+1}-T_k=\bar t_k(a_1)$. Item A follows from inequality~\eqref{eq:exponentialconvergence} in Theorem~\ref{theorem:convergenceOfDynamics} and from the fact that the slopes of both $\bar y_k$ and the stationary trajectory $\tilde y_\infty$ are bounded below by $q_1$.

Let $0\leq \ell\leq M$ and let $t\in(0,\min(T_{k+1}-T_k,\tilde t_\infty(a_1))]$. Then $x'_{-\ell}(t+T_k)=x_{k-\ell}(t+T_k)$ since the front of $x$ is at position $k$ at time $t+T_k$. Thus
\begin{align*}
x'_{-\ell}(t+T_k)&=y_{k-\ell}(t+T_k)-y_{k-\ell+1}(t+T_k) \\
&=\bar y_{k-\ell}(t+T_k-T_{k-\ell})-\bar y_{k-\ell+1}(t+T_k-T_{k-\ell+1}).
\end{align*}

On the other hand $\tilde x_{\infty,-\ell}(t)=\tilde{\tilde y}_{\infty,-\ell}(t)-\tilde{\tilde y}_{\infty,-\ell+1}(t)$. Applying Lemma~\ref{lem:stationarityofstationaryconfiguration}, we get $\tilde x_{\infty,-\ell}(t)=\tilde y_{\infty}(t+\ell\tilde t_\infty(a_1))-\tilde y_{\infty}(t+(\ell-1)\tilde t_\infty(a_1))$.

Thus
\begin{multline*}
\left|x'_{-\ell}(t+T_k)-\tilde x_{\infty,-\ell}(t)\right|\leq \left|\bar y_{k-\ell}(t+T_k-T_{k-\ell})-\tilde y_{\infty}(t+\ell\tilde t_\infty(a_1))\right| \\ + \left|\bar y_{k-(\ell-1)}(t+T_k-T_{k-(\ell-1)})-\tilde y_{\infty}(t+(\ell-1)\tilde t_\infty(a_1))\right|.
\end{multline*}
The first term of the right-hand side may be dominated by
\[
\left|\bar y_{k-\ell}(t+T_k-T_{k-\ell})-\bar y_{k-\ell}(t+\ell\tilde t_\infty(a_1))\right|+\left|\bar y_{k-\ell}(t+\ell\tilde t_\infty(a_1))-\tilde y_{\infty}(t+\ell\tilde t_\infty(a_1))\right|,
\]
which is itself dominated by
\[
q_N\left|T_k-T_{k-\ell}-\ell\tilde t_\infty(a_1)\right|+\kappa \left(1 - \frac{q_1}{q_{\nbParam}}\right)^k.
\]
Applying inequality~\eqref{eq:exponentialcvperiods} we obtain for the first term of the right-hand side a bound of
\[
\left(\frac{\ell q_N}{q_1}+1\right)\kappa  \left(1 - \frac{q_1}{q_{\nbParam}}\right)^k.
\]
A similar computation provides for the second term an upper bound of
\[
\left(\frac{(\ell-1) q_N}{q_1}+1\right)\kappa \left(1 - \frac{q_1}{q_{\nbParam}}\right)^k.
\]

We conclude that item B holds true.

This concludes the proof of Theorem~\ref{theorem:mainstatio}.

\begin{remark}
\label{rem:finitetime}
    In some cases, the stationary regime is reached in finite time. For instance, this is the case when the DC graph of the stationary configuration is complete.
    In this case, there is a time at which there is no car between road signs $1$ and $\nbParam$. Then, one may check that any car starting to move after this time has the stationary trajectory. We conjecture that the complete graph is the only connected graph for which the stationary regime is reached in finite time. 
\end{remark}

\section{Formula for the front speed}
\label{sec:speed}

In this section we prove the formula for the front speed of the liquid bin model, stated in Theorem~\ref{theorem:mainspeed}. For this, we first partition $P^{2N}$ into regions to which we associate a linear system (Subsection~\ref{subsec:formulaMatrices}), then we solve this linear system in Subsection~\ref{subsec:solving}. The inverse of the front speed is the first component of the vector that is the solution of the linear system. Finally in Subsection~\ref{subsec:phasesAndBoundaries} we show that the solutions of the linear systems are continuous across regions and that each region has non-empty interior.

\subsection{Partitioning the parameter space}\label{subsec:formulaMatrices}

Let $(\underline a,\underline p) \in P^{2\nbParam }$ and set $\widetilde{y}_{\infty}$ to be the stationary trajectory associated to these parameters. Set $\widetilde{t}_{\infty} = (\widetilde{y}_{\infty})^{-1}$ to be the inverse function of $\widetilde{y}_{\infty}$.
Given $\widetilde{t}_\infty(a_i)$ for all $i \in \llbracket 1 , \nbParam \rrbracket$, one can easily recover $\widetilde{y}_\infty$, since it is a piecewise linear continuous function which is differentiable away from the points $\widetilde{t}_\infty(a_i)-\widetilde{t}_\infty(a_1)$, with derivative equal to $q_i$ on each interval of the form $(\widetilde{t}_\infty(a_{i})-\widetilde{t}_\infty(a_1),\widetilde{t}_\infty(a_{i+1})-\widetilde{t}_\infty(a_1))$. It follows from \eqref{eq:recursiveFormulaSpeedTilde3} that $(\widetilde{t}_\infty(a_i))_{1 \leq i \leq \nbParam}$ satisfies the following non-linear relations:
\begin{align}\label{eq:systemStatCar}
    \forall i \in \llbracket 1 , \nbParam \rrbracket,\, a_i = \sum_{j=1}^{\nbParam} p_j (\widetilde{t}_\infty(a_i) - \widetilde{t}_\infty(a_j) + \widetilde{t}_\infty(a_1))_+.
\end{align}

Recall from Definition~\ref{def:canonicalstationaryconfiguration} that the canonical stationary configuration  is given by $y_{-k}(0)= \widetilde{y}_\infty(k\widetilde{t}_\infty(a_1))$ for all $k\geq1$ and $y_{k}(0)= 0$ for all $k\geq0$. As in Subsection~\ref{subsec:contractivity}, one may associate a DC graph to each car of the canonical stationary configuration. This time the DC graph is independent of the index of the car because of stationarity.
\begin{definition}
Let $(\underline a,\underline p)\in P^{2N}$. We define the downward closed graph $\Gr(\underline a,\underline p)\in\DC_N$ associated to $(\underline a,\underline p)$ to be the directed graph with vertex set $\llbracket 1, N\rrbracket$ and edge set given by all the pairs $(i,j)$ with $i<j$ satisfying
\begin{equation}
\label{eq:positivepartpositive}
 \widetilde{t}_{\infty}(a_i) > \widetilde{t}_{\infty}(a_j)-\widetilde{t}_{\infty}(a_1).
\end{equation}
\end{definition}
Recall that inequality \eqref{eq:positivepartpositive} is also satisfied whenever $i\geq j$, but we do not add such directed edges $(i,j)$ to the DC graph. It is not hard to see that this definition of $\Gr(\underline a,\underline p)$ from the stationary car model coincide with the definition of $\Gr(\underline a,\underline p)$ from the stationary liquid bin model given in Subsection~\ref{subsec:LBMproperties}.

We have the following useful interpretation of $\Gr(\underline a,\underline p)$ in terms of the stationary car model.
For every $(\underline a,\underline p)\in P^{2N}$ and $1\leq i\leq N$, define
   \begin{equation}
   \label{eq:defxi}
    \chi_i:=\widetilde y_\infty(\widetilde t_\infty (a_1)+\widetilde t_\infty (a_i)).
   \end{equation}
The quantity $\chi_i$ corresponds to the position of the car of index $-1$ in a stationary configuration where the car of index $0$ is at position $a_i$. The following lemma is a straightforward reformulation of the definition in \eqref{eq:positivepartpositive}:

 \begin{lemma}
\label{lem:Grcarinterpretation}
Let $(\underline a,\underline p)\in P^{2N}$ and $1\leq i<j\leq N$. Then $\chi_i>a_j$ if and only if $(i,j)\in E(\Gr(\underline a,\underline p))$.
\end{lemma}

For any DC graph $G\in\DC_N$, set
\[P_G := \Set{(\underline a,\underline p)\in P^{2\nbParam }}{\Gr(\underline a,\underline p)= G}\] to be the set of all parameters for which the DC graph of the stationary trajectory is $G$.
Since there is exactly one stationary trajectory for given $(\underline a,\underline p)\in P^{2N}$ by Theorem \ref{theorem:convergenceOfDynamics}, we obtain the following partition of the parameter space $P^{2N}$ by the sets $P_G$:
\[
P^{2N}=\bigsqcup_{G\in \DC_N} P_G.
\]

If we know that $(\underline a,\underline p)\in P_G$ for some graph $G$, then the relations \eqref{eq:systemStatCar} become linear:
\begin{align}\label{eq:systemStatCarBis}
    \forall i \in \llbracket 1 , \nbParam \rrbracket,\, a_i = \sum_{j=1}^{b_G(i)} p_j (\widetilde{t}_{\infty}(a_i) - \widetilde{t}_{\infty}(a_j) + \widetilde{t}_{\infty}(a_1))
\end{align}
where we recall that $b_G(i)$ is either the largest $j$ such that there exists some edge $(i,j)$ in $G$, or $i$ if no such edge exists.

\begin{remark}
\label{rem:speedz1}
    The reciprocal of the front speed in the stationary liquid bin model corresponds to the time elapsed between two consecutive jumps of the cursor $c_1$. In terms of the stationary car model, it corresponds to the time elapsed between two consecutive departures of cars from $0$, namely $\widetilde{t}_\infty(a_1)$. This is readily computed by matrix inversion as soon as one knows in which region $P_G$ the parameters lie.
\end{remark}

\subsection{Solving the linear system}
\label{subsec:solving}

In this subsection we shall solve the linear system \eqref{eq:systemStatCarBis} with unknowns $\widetilde{t}_{\infty}(a_{i})$, which holds whenever the parameters $(\underline a,\underline p)$ are restricted to the region $P_G$ for a fixed DC graph $G\in\DC_N$.

Let us first perform a change of variables. For all $i \in \llbracket 1 , \nbParam\rrbracket$, set 
\begin{equation}\label{eq:definitionOfZi}
     z_i(\underline a,\underline p) := \widetilde{t}_{\infty}(a_{i})-\widetilde{t}_{\infty}(a_{i-1})
\end{equation}
with the convention that $a_{0}= 0$ (implying that $\widetilde{t}_{\infty}(a_{0})=0$). With these new variables, $(i,j)$ is an edge in $\Gr(\underline a,\underline p)$ if and only if 
\begin{equation}
\label{eq:edgeGrapz}
    z_1(\underline a,\underline p) > z_{i+1}(\underline a,\underline p) + \dots + z_j(\underline a,\underline p).
\end{equation}

\begin{definition}
\label{def:linearsystem}
Let $G\in\DC_N$ and let $(\underline a,\underline p)\in P^{2N}$. The linear system $\mathcal S^{(G)}(\underline a,\underline p)$ associated with $G$ and $(\underline a,\underline p)$ is the following system of equations with unknowns $\zeta_1,\ldots,\zeta_N$:
\begin{equation}
\label{eq:newlinearsystem}
\forall i\in\llbracket1,N\rrbracket,\ a_i=\sum_{j=1}^{b_G(i)}p_j((\zeta_1+\cdots+\zeta_i)-(\zeta_1+\cdots+\zeta_j)+\zeta_1)
\end{equation}
\end{definition}

It follows from \eqref{eq:systemStatCarBis} that the $(z_i(\underline a,\underline p))_{i\in\llbracket 1,N\rrbracket}$ are solutions of $\mathcal S^{(G)}(\underline a,\underline p)$ if $(\underline a,\underline p)\in P_G$.

We will define rational functions $z_i^{(G)}$ on $P^{2N}$ for all $1\leq i\leq N$ and $G\in\DC_N$. Then we will show in Theorem~\ref{theorem:formulaSpeed} that $z_i$ coincides with $z_i^{(G)}$ on each $P_G$ for all $1\leq i\leq N$.

Recall from Subsection~\ref{subsec:LBMproperties} the definitions of $\gamma^{(G)}_{i,j}$ and $\Gamma^{(G)}_{i,j}$. Note that by construction, $\Gamma^{(G)}_{i,j}$ is non-negative for all $1 \leq i < j \leq \nbParam$.

We adopt the convention that $b_G(0)=1$. For every DC graph $G\in\DC_N$ and every $2 \leq i \leq \nbParam$, define 
\begin{equation}
\label{eq:eq:z1formula}
    z_1^{(G)}(\underline a,\underline p) := \frac{\sum_{j=1}^N \Gamma^{(G)}_{1,j} \frac{ d_j }{q_{b_G(j-1)}}}{ 1 + \sum_{j=1}^N \Gamma^{(G)}_{1,j} \frac{  q_{b_G(j)}-q_{b_G(j-1)}}{q_{b_G(j-1)}}}
\end{equation}
and
\begin{align}
    z_i^{(G)}(\underline a,\underline p) &:=  \left(\sum_{j = i}^N  \Gamma^{(G)}_{i,j} \frac{d_j}{q_{b_G(j-1)}}\right) -  z_1^{(G)}(\underline a,\underline p) \left( \sum_{j = i}^N  \Gamma^{(G)}_{i,j} \frac{q_{b_G(j)}-q_{b_G(j-1)}}{q_{b_G(j-1)}} \right) \label{ziformula} \\
    &=  \left(\sum_{j = i}^N  \Gamma^{(G)}_{i,j} \frac{d_j}{q_{b_G(j-1)}}\right) -  \frac{\left( \sum_{j = i}^N  \Gamma^{(G)}_{i,j} \frac{q_{b_G(j)}-q_{b_G(j-1)}}{q_{b_G(j-1)}} \right) \left(\sum_{j=1}^N \Gamma^{(G)}_{1,j} \frac{ d_j }{q_{b_G(j-1)}}  \right) }{   1 + \sum_{j=1}^N \Gamma^{(G)}_{1,j} \frac{  q_{b_G(j)}-q_{b_G(j-1)}}{q_{b_G(j-1)}} }.
\end{align}

Then the linear system $\mathcal S^{(G)}(\underline a,\underline p)$ can be solved as follows. 
\begin{theorem}\label{theorem:formulaSpeed}
For every DC graph $G\in\DC_N$ and every $(\underline a,\underline p) \in P^{2N}$, the linear system $\mathcal S^{(G)}(\underline a,\underline p)$ has a unique solution, given by $(z_i^{(G)}(\underline a,\underline p))_{i\in\llbracket1,N\rrbracket}$. Thus, when $(\underline a,\underline p) \in P_G$, for every $1 \leq i \leq \nbParam$ we have
    \begin{align}\label{eq:formulaInvSpeed}
    z_i(\underline a,\underline p) = z_i^{(G)}(\underline a,\underline p) .
\end{align}
This implies in particular that, for every DC graph $G\in\DC_N$ and every $(\underline a,\underline p) \in P_G$, the speed of the front of the stationary liquid bin model with parameters $(\underline a,\underline p)$ is given by $1/z_1^{(G)}(\underline a,\underline p)$.
   \end{theorem}

\begin{proof}
Let us prove the first statement. The second statement will follow from \eqref{eq:systemStatCarBis} while the third one is a simple consequence of Remark~\ref{rem:speedz1}.

By considering the difference of rows $i$ and $i-1$ in the linear system $\mathcal S^{(G)}(\underline a,\underline p)$ and with the conventions that $b_G(0) = 1$ and $q_0 = 0$, we have for every $i \in \llbracket 1 , \nbParam \rrbracket$,  
\begin{equation*}\label{eq:systD'}
    d_i = (q_{b_G(i)}-q_{b_G(i-1)})\zeta_1 + q_{b_G(i-1)}\zeta_{i} - \sum_{j=b_G(i-1)+1}^{b_G(i)}p_j \sum_{l=i+1}^j \zeta_l.
\end{equation*}
The sum over $j$ on the right-hand side can be rewritten as follows :
\begin{align*}
& \sum_{j=b_G(i-1)+1}^{b_G(i)}p_j \sum_{l=i+1}^j \zeta_l = \sum_{j=b_G(i-1)+1}^{b_G(i)} (q_j - q_{j-1}) \sum_{l=i+1}^j \zeta_l   \\
&= q_{b_G(i)}\sum_{l=i+1}^{b_G(i)} \zeta_l-q_{b_G(i-1)}\sum_{l=i+1}^{b_G(i-1)+1} \zeta_l-\sum_{j=b_G(i-1)+2}^{b_G(i)} q_{j-1}\zeta_j  \\
&=\sum_{j=i+1}^{b_G(i)} q_{b_G(i)}\zeta_{j}-\sum_{j=i+1}^{b_G(i-1)+1} q_{b_G(i-1)}\zeta_{j}-\sum_{j=b_G(i-1)+2}^{b_G(i)} q_{j-1}\zeta_j \\
&=\sum_{j=i+1}^{b_G(i)} (q_{b_G(i)}-q_{\max(j-1,b_G(i-1))})  \zeta_{j}.
\end{align*}
Putting everything together, we obtain the following linear system for the $\zeta_i$:
\begin{equation}\label{eq:systD}
    \forall i \in \llbracket 1 , \nbParam \rrbracket,\, d_i = (q_{b_G(i)}-q_{b_G(i-1)})\zeta_1 + q_{b_G(i-1)}\zeta_{i} - \sum_{j=i+1}^{b_G(i)} (q_{b_G(i)}-q_{\max(j-1,b_G(i-1))})  \zeta_{j}.
\end{equation}

Notice that if we assume that we know $\zeta_1$, this linear system in $(\zeta_2, \dots, \zeta_{\nbParam})$ becomes upper-triangular.
For all $i \in \llbracket 1 , \nbParam \rrbracket$, set 
\[\alpha_i := \frac{d_i - (q_{b_G(i)}-q_{b_G(i-1)})\zeta_1}{q_{b_G(i-1)}}. \]   
Then the linear system \eqref{eq:systD} becomes 
\begin{align*}
    \forall i \in \llbracket 1 , \nbParam \rrbracket, \  \alpha_i =  \zeta_{i} - \sum_{j=i+1}^{b_G(i)} \gamma^{(G)}_{i,j}  \zeta_{j}.
\end{align*}

By descending induction on $i$ from $N$ to $1$, one obtains that for all $i \in \llbracket 1 , \nbParam \rrbracket$, 
\begin{align}\label{eq:systD2}
     \zeta_{i} =  \sum_{j = i}^N  \Gamma^{(G)}_{i,j} \alpha_j .
\end{align}
In passing we make use of the formula
\[
\Gamma^{(G)}_{i,j}=\sum_{h = i+1}^{\min(j,b_G(i))} \gamma^{(G)}_{i,h}\Gamma^{(G)}_{h,j}
\]
which holds true by decomposing any increasing path from $i$ to $j$ according to its first step $(i,h)$.

Formula~\eqref{eq:systD2} entails that $\zeta_i$ equals the right-hand side of \eqref{ziformula} for $2\leq i\leq N$. Taking $i = 1$ in \eqref{eq:systD2} gives
\begin{align*}
    \zeta_1 = \sum_{j=1}^N \Gamma^{(G)}_{1,j} \frac{ d_j - (q_{b_G(j)}-q_{b_G(j-1)})\zeta_1}{q_{b_G(j-1)}}.  
\end{align*}
Solving this linear equation in $\zeta_1$ shows that $\zeta_1$ equals the right-hand side of \eqref{eq:eq:z1formula}.
\end{proof}

We stress that the linear system $\mathcal S^{(G)}(\underline a,\underline p)$ is defined for every $(\underline a,\underline p)\in P^{2N}$, not only for $(\underline a,\underline p)\in P_G$. We will sometimes need to consider its solutions $(z^{(G)}_i(\underline a,\underline p))_{i\in\llbracket 1,N\rrbracket}$ for parameters $(\underline a,\underline p)\notin P_G$.

\begin{definition}
Let $G$ be a DC graph. An edge $(i,j)\in E(G)$ is called a \emph{maximal edge} of $G$ if it is a maximal element of the poset $(E(G),\prec_E)$.
\end{definition}

\begin{remark}\label{remark:reduxSystem}
    Define $B = \Set{i \in \llbracket 1 , \nbParam \rrbracket}{ b_G(i-1) \neq b_G(i)}$. The elements $i$ of $B$ are either isolated vertices in $G$ (if $b_G(i)=i$) or the starting point of a maximal edge (if $b_G(i)>i$). For all $i \in \llbracket 1 , \nbParam \rrbracket \backslash B$, it follows from \eqref{eq:systD} that
    \[z_i^{(G)}(\underline a, \underline p) = \frac{d_i}{q_{b_G(i-1)}}.\]
    Therefore, one can reduce the size of the system from $\nbParam$ to $|B|$ equations. 
\end{remark}

\begin{example}
\label{ex:completespeed}
In the case where $G= K_{\nbParam}$ is the complete graph with $N$ vertices, $\Gamma^{(G)}_{i,j} = \mathds{1}_{i=1} \frac{q_{\nbParam} - q_{j-1}}{q_1}$ for all $1 \leq i < j \leq \nbParam$. Thus for every $(\underline a,\underline p)\in P_{K_N}$, the speed of the front is given by
\begin{equation}
\label{eq:completespeed}
\frac{1}{z_1^{(K_{\nbParam})}(\underline a,\underline p)} = \frac{q_{\nbParam}^2}{\sum_{j=1}^{\nbParam} \left( q_{\nbParam}-q_{j-1}\right)d_j } .
\end{equation}
\end{example}

\begin{example}
\label{ex:linespeed}
If $G=L_{\nbParam}$ is the graph with $N$ vertices such that $E(L_{\nbParam}) = \Set{(i,i+1)}{i \in \llbracket 1 , \nbParam -1 \rrbracket}$ (we call such a graph a line graph), then $\gamma^{(L_N)}_{i,j} = \mathds{1}_{j=i+1} \frac{p_{i+1}}{q_i} $ for all $1 \leq i < j \leq \nbParam$. Therefore, 
\[\Gamma^{(L_{\nbParam})}_{i,j} = \frac{p_{i+1}\dots p_j}{q_i \dots q_{j-1}}\] for all $1 \leq i < j \leq \nbParam$. 
As a consequence, with the convention that the empty product equals $1$, the front speed is
\begin{equation}
\label{eq:linespeed}
\frac{1}{z_1^{(L_{\nbParam})}(\underline a,\underline p)} = \frac{\sum_{j=0}^{\nbParam-1} p_{j+1} \left( \prod_{h=1}^j \frac{p_h}{q_h} \right)  }{\sum_{j=1}^{\nbParam} d_j \left( \prod_{h=1}^j \frac{p_h}{q_h} \right)}.
\end{equation}
\end{example}

\subsection{Every region has non-empty interior}\label{subsec:phasesAndBoundaries}

 The goal of this subsection is to prove the following result:

\begin{theorem}\label{theorem:nonemptiness}
    For every  $G\in\DC_N$, the interior of $P_G$ is non-empty. 
\end{theorem}

Let us first prove a result on the regularity of the $z_i$ in the parameters $(\underline a,\underline p)$.

\begin{proposition}\label{proposition:continuityInParameters}
    For every $ i \in \llbracket 1 , \nbParam \rrbracket$, the map $(\underline a,\underline p) \in P^{2N} \mapsto z_i (\underline a,\underline p)$ is continuous. 
\end{proposition}
\begin{proof}
    Fix $i \in \llbracket 1 , \nbParam \rrbracket$. By Subsection \ref{subsec:formulaMatrices}, for every $G\in\DC_N$, the restriction of $z_i$ to the region $P_{G}$ is a rational function of the parameters $(\underline a,\underline p)$. 
    It remains to prove that $z_i$ is continuous across boundaries of regions.
    Let $G_1$ and $G_2$ be two distinct DC graphs with $\nbParam$ vertices such that $\partial P_{G_1} \cap \partial P_{G_2}$ is non-empty. Pick $(\underline a,\underline p)$ in $\partial P_{G_1} \cap \partial P_{G_2}$. 
    The functions $z_i$ and $z_i^{(G_l)}$ coincide on $P_{G_l}$ for $l\in\{1,2\}$. Therefore, it suffices to show that $z^{(G_1)}_i(\underline a,\underline p) =  z^{(G_2)}_i (\underline a,\underline p) $. 

    Since $(\underline a,\underline p)\in\partial P_{G_1} \cap \partial P_{G_2}$, there exists a sequence $(\underline a^{s,1},\underline p^{s,1})_{s\in\ZZ_+}$  (resp. $(\underline a^{s,2},\underline p^{s,2})_{s\in\ZZ_+}$) of elements of $P_{G_1}$ (resp. $P_{G_2}$) converging to $ (\underline a,\underline p) $ as $s$ goes to infinity. 
    If for every $s\in\ZZ_+$ and $l\in\{1,2\}$ we set
    \begin{align*}
        \widetilde{y}_{\infty}^{s,l}(t) &:= \sum_{j=1}^{\nbParam} p_j^{s,l} \left( t - \left( z_2^{(G_l)}(\underline a^{s,l},\underline p^{s,l}) + \dots  + z_j^{(G_l)}(\underline a^{s,l},\underline p^{s,l})\right)\right)_+ , 
    \end{align*}
    then $\widetilde{y}_{\infty}^{s,l}$ is the only stationary trajectory for parameters $(\underline a^{s,l},\underline p^{s,l})$. In other words, $\widetilde{y}_{\infty}^{s,l}$ is the only function satisfying \eqref{eq:recursiveFormulaSpeedTilde3} for parameters $(\underline a^{s,l},\underline p^{s,l})$. The functions $t\mapsto\widetilde{y}_{\infty}^{s,l}(t)$ and its inverse $x\mapsto\widetilde{t}_{\infty}^{s,l}(x)$ are both continuous piecewise affine functions with $N$ points of non-differentiability. The coordinates of these points of non-differentiability are rational functions of the parameters $(\underline a,\underline p)$. 
    When $s$ goes to infinity, we denote by $\widetilde{y}_{\infty}^{l}(t)$ and $\widetilde{t}_{\infty}^{l}(x)$ the pointwise limits of $\widetilde{y}_{\infty}^{s,l}(t)$ and $\widetilde{t}_{\infty}^{s,l}(x)$. The functions $t\mapsto\widetilde{y}_{\infty}^{l}(t)$ and $x\mapsto\widetilde{t}_{\infty}^{l}(x)$ are still inverses of each other and they satisfy \eqref{eq:recursiveFormulaSpeedTilde3} for parameters $(\underline a,\underline p)$.
    
    Since the stationary trajectory for parameters $(\underline a,\underline p)$ is unique by Theorem \ref{theorem:convergenceOfDynamics}, $\widetilde{y}_{\infty}^{1} = \widetilde{y}_{\infty}^{2}$.  
    As a consequence, the coordinates of the points of non-differentiability of $\widetilde{y}_{\infty}^{1}$ and $\widetilde{y}_{\infty}^{2}$ are the same, which means that $ z^{(G_1)}_i(\underline a,\underline p) =  z^{(G_2)}_i (\underline a,\underline p) $ for all $ i \in \llbracket 1 , \nbParam \rrbracket$. 
\end{proof}

We will also need the continuity lemma below. Whenever we have some finite sequence $\underline s$ of real numbers of length at least $2$, we shall denote by $\underline s^\dagger$ the sequence obtained from $\underline s$ by removing its last element. For example, if $\underline p=(p_1,\ldots,p_N)$, then $\underline p^\dagger=(p_1,\ldots,p_{N-1})$.

\begin{lemma}\label{lemma:continuityInParameterPN}
    Fix $\nbParam \geq 2$. Let $(z_i(\underline a,\underline p))_{i \in \llbracket 1 , \nbParam \rrbracket}$ (resp. $z'_i(\underline a^\dagger,\underline p^\dagger)_{i \in \llbracket 1 , \nbParam - 1 \rrbracket}$) be the functions defined by \eqref{eq:definitionOfZi} for $(\underline a,\underline p)\in P^{2N}$ (resp. $(\underline a^\dagger,\underline p^\dagger)\in P^{2N-2}$). Then  we have
    \[ \forall i \in \llbracket 1 , \nbParam - 1 \rrbracket, \, z_i(\underline a,\underline p) \xrightarrow[p_{\nbParam} \rightarrow 0]{} z_i'(\underline a^\dagger,\underline p^\dagger).\]
Moreover
\[
z_N(\underline a,\underline p) \xrightarrow[p_{\nbParam} \rightarrow 0]{} \frac{a_N-a_{N-1}}{q_{N-1}}.
\]
\end{lemma}

 \begin{proof}
Since $(\underline a,\underline p^\dagger,0)$ lies in the closure $\overline {P^{2N}}$ of $P^{2N}=\bigsqcup_{G\in\DC_N} P_G$, there exists some $G\in\DC_N$ such that $(\underline a,\underline p^\dagger,0)\in\overline{P_G}$. Pick a sequence $(\underline a^s, \underline p^s )_{s \in \ZZ_{>0}}$ of elements of $P_G$ converging to $(\underline a,\underline p^\dagger,0)$. On $P_G$, $z_i$ coincides with $z_i^{(G)}$ for every $i\in\llbracket 1 , \nbParam-1 \rrbracket$.

Set $ \widetilde{y}_{\infty}^{s,G} $ to be the unique solution of the system \eqref{eq:recursiveFormulaSpeedTilde3}: 
\begin{align*}
        \widetilde{y}_{\infty}^{s,G}(t) &:= \sum_{j=1}^{\nbParam} p_j^{s} \left( t - \left( z_2^{(G)}(\underline a^{s},\underline p^{s}) + \dots  + z_j^{(G)}(\underline a^{s},\underline p^{s})\right)\right)_+. 
\end{align*}
By definition of a stationary trajectory \eqref{eq:recursiveFormulaSpeedTilde3}, 
\begin{align}\label{eq:recursiveFormulaSpeedTilde3BIS}
\widetilde{y}_\infty^{s,G}(t) = \sum_{j=1}^{\nbParam} p_j^s \left(t - \widetilde{t}^{s,G}_\infty(a_j^s) + \widetilde{t}^{s,G}_\infty(a_1^s) \right)_+,
\end{align}
where $\widetilde{t}_{\infty}^{s,G}$ is the inverse bijection of $\widetilde{y}_{\infty}^{s,G}$, and $(\widetilde{t}^{s,G}_\infty(a_j^s))_{j \in \llbracket 1 , \nbParam \rrbracket}$ is non-decreasing.  

For every $i\in\llbracket 1 , \nbParam \rrbracket$, the function $z_i^{(G)}$ is rational hence continuous in the parameters $(\underline a,\underline p)$. Thus, when $s$ goes to infinity, $\widetilde{y}_{\infty}^{s,G} $ pointwise converges on $\RR_{\geq0}$ to $\widetilde{y}_{\infty}^{G} $ defined by 
\begin{align*}
        \widetilde{y}_{\infty}^{G}(t) &:= \sum_{j=1}^{\nbParam-1} p_j \left( t - \left( z_2^{(G)}(\underline a,\underline p^\dagger,0) + \dots  + z_j^{(G)}(\underline a,\underline p^\dagger,0)\right)\right)_+ . 
\end{align*}
The function $\widetilde{y}_{\infty}^{G}$ is a bijection since $p_1>0$. Since $\widetilde{y}_{\infty}^{s,G}$ is a piecewise affine function for which the slopes and the coordinates of the points of non-differentiability  are continuous functions of the parameters, $\widetilde{t}_{\infty}^{s,G}$ also pointwise converges to $\widetilde{t}_{\infty}^{G}$, where $\widetilde{t}_{\infty}^{G}$ is the inverse bijection of $\widetilde{y}_{\infty}^{G}$. 
By taking the limit $s$ goes to infinity in \eqref{eq:recursiveFormulaSpeedTilde3BIS}, one obtains that  
\begin{align}
\widetilde{y}_\infty^{G}(t) = \sum_{j=1}^{\nbParam-1} p_j \left(t - \widetilde{t}^{G}_\infty(a_j) + \widetilde{t}^{G}_\infty(a_1) \right)_+,
\end{align}
where $\widetilde{t}^{G}_{\infty}$ is the inverse bijection of $\widetilde{y}^{G}_{\infty}$. 

Therefore, $\widetilde{y}_{\infty}^{G}$ is a stationary trajectory for parameters $(\underline a^\dagger,\underline p^\dagger)$. 

Since the stationary trajectory is unique by Theorem \ref{theorem:convergenceOfDynamics}, $\widetilde{y}_\infty^{G} = \widetilde{y}'_\infty$, where $\widetilde{y}'_\infty$ is the stationary trajectory for parameters $(\underline a^\dagger,\underline p^\dagger)$, with inverse denoted by $\widetilde{t}'_\infty$. 
By comparing the points of non-differentiability of $\widetilde{y}_\infty^{G}$ and $\widetilde{y}'_\infty$, one obtains that $z_i^{(G)}(\underline a,\underline p^\dagger,0)  = z_i'(\underline a^{\dagger},\underline p^{\dagger}) $ for all $i \in \llbracket 2 , \nbParam -1 \rrbracket$. To get this equality for $i=1$, observe that $ z_1^{(G)}(\underline a^s,\underline p^s)=\widetilde{t}_{\infty}^{s,G}(a_1)$ for all $s>0$, which implies that \[z_1^{(G)}(\underline a,\underline p^\dagger,0)=\widetilde{t}_{\infty}^{G}(a_1)=\widetilde{t}'_{\infty}(a_1)=z_1'(\underline a^{\dagger},\underline p^{\dagger}).\]
The result for $z_N$ follows from taking the limit $p_N$ goes to $0$ in the formula
\[
z_N= \int_{a_{N-1}}^{a_N} \frac{1}{\widetilde v_\infty\circ \widetilde t_\infty(s)}ds,
\]
since every car beyond position $a_{N-1}$ has a speed equal to either $q_{N-1}$ or $q_{N-1}+p_N$.
\end{proof}

\begin{proof}[Proof of Theorem \ref{theorem:nonemptiness}]
    Let us proceed by induction on $\nbParam=|V(G)|$. For $\nbParam= 1$, the only possible graph is the one with a single vertex $G=(\{ 1 \} , \varnothing )$. In this case, $\mathring{P_G}$ is the whole parameter space $P^2$ which is non-empty.

    Now, fix $N\geq 2$ and assume that $\mathring{P_{G'}}$ is non-empty for every DC graph $G'$ with at most $N-1$ vertices. Let $(z_i(\underline a,\underline p))_{i \in \llbracket 1 , \nbParam \rrbracket}$ (resp. $z'_i(\underline a',\underline p')_{i \in \llbracket 1 , \nbParam - 1 \rrbracket}$) be the functions defined by \eqref{eq:definitionOfZi} for $(\underline a,\underline p)\in P^{2N}$ (resp. $(\underline a',\underline p')\in P^{2N-2}$).
    
    Let $G\in\DC_N$ and define $G'$ to be the restriction of $G$ to the vertex set $\llbracket 1 , \nbParam - 1 \rrbracket$.

    It suffices to find $(\underline a,\underline p)\in P^{2N}$ satisfying the following two conditions:
    \begin{enumerate}
    \item For every $1\leq i<j\leq N$ such that $(i,j)\in E(G)$,
    \begin{equation}
    \label{eq:proofExistParam1} 
         z_1(\underline a, \underline p) > z_{i+1}(\underline a, \underline p) + \dots + z_j(\underline a, \underline p).
    \end{equation}
    \item For every $1\leq i<j\leq N$ such that $(i,j)\notin E(G)$,
    \begin{equation}
    \label{eq:proofExistParam2} 
         z_1(\underline a, \underline p) < z_{i+1}(\underline a, \underline p) + \dots + z_j(\underline a, \underline p).
    \end{equation}
    \end{enumerate}
    Indeed, by the definition \eqref{eq:edgeGrapz} of $\Gr(\underline a,\underline p)$ in terms of the $z_i$, this will imply that $(\underline a,\underline p)\in P_G$. Since the inequalities are strict and the $z_i$ are continuous functions of $(\underline a,\underline p)\in P^{2N}$ by Proposition \ref{proposition:continuityInParameters}, every collection of parameters in a neighborhood of $(\underline a,\underline p)$ will satisfy the same inequalities, thus we will conclude that $\mathring{P_G}$ is non-empty.

By induction hypothesis, $\mathring{P_{G'}}$ is non-empty. By \eqref{eq:edgeGrapz}, this implies that for every $(\underline a',\underline p')\in\mathring{P_{G'}}$ and every $1\leq i<j\leq N-1$ such that $(i,j)\in E(G')$, we have
\begin{equation}
\label{eq:proofExistParam1G'}
z'_1(\underline a', \underline p') > z'_{i+1}(\underline a', \underline p') + \dots + z'_j(\underline a', \underline p').
\end{equation}
Moreover, for every $(\underline a',\underline p')\in\mathring{P_{G'}}$ and every $1\leq i<j\leq N-1$ such that $(i,j)\notin E(G')$, we have by \eqref{eq:edgeGrapz}
\begin{equation}
\label{eq:proofExistParam2G'large}
z'_1(\underline a', \underline p') \leq z'_{i+1}(\underline a', \underline p') + \dots + z'_j(\underline a', \underline p').
\end{equation}
The functions $z'_{i+1}+\cdots+z'_j-z'_1$ are non-constant rational functions on the non-empty open set $\mathring{P_{G'}}$, thus one can find $(\underline a', \underline p') := (a_1, \dots , a_{\nbParam - 1}, p_1,\ldots,p_{\nbParam - 1 }) \in \mathring{P_{G'}}$ such that for every $1\leq i<j\leq N-1$ with $(i,j)\notin E(G')$, we have
\begin{equation}
\label{eq:proofExistParam2G'strict}
z'_1(\underline a', \underline p') < z'_{i+1}(\underline a', \underline p') + \dots + z'_j(\underline a', \underline p').
\end{equation} 

From now on, the values of $(\underline a',\underline p')$ are fixed as in the previous paragraph. We complete $(\underline a', \underline p')$ to $(\underline a,\underline p)$, where $a_N>a_{N-1}$ and $p_N>0$ are not fixed for the moment and will be appropriately chosen later. With this completion, $(\underline a',\underline p')=(\underline a^\dagger,\underline p^\dagger)$.

 By Lemma \ref{lemma:continuityInParameterPN}, for every $1\leq i\leq N$, the function $z_i$ can be extended by continuity to the case when $p_N=0$ by setting
    \begin{align}
    z_i(\underline a^\dagger,a_N,\underline p^\dagger,0)&:=z'_i(\underline a^\dagger,\underline p^\dagger) \quad \text{ if } i\in\llbracket 1 , \nbParam -1 \rrbracket; \label{eq:zipnzero} \\
    z_N(\underline a^\dagger,a_N,\underline p^\dagger,0)&:=\frac{a_N-a_{N-1}}{q_{N-1}}. \label{eq:znpnzero}
    \end{align}
With this extension, we have all the inequalities \eqref{eq:proofExistParam1} and \eqref{eq:proofExistParam2} whenever $j\neq N$ and $p_N=0$, regardless of the choice of $a_N>a_{N-1}$.

Consider the parameters $(\underline a,\underline p)$, with $(\underline a^\dagger,\underline p^\dagger)$ fixed as above and $a_N>a_{N-1}$ and $p_N\geq0$ free parameters for now. Recall that for every $1\leq i\leq N$, the position of the car of index $-1$ in a stationary configuration where the car of index $0$ is at position $a_i$ is given by $\chi_i=\widetilde y_\infty(\widetilde t_\infty (a_1)+\widetilde t_\infty (a_i))$. We extend this definition to the case $i=0$ by setting $\chi_0:=a_1$. Consider the quantities $\chi_{i}$ as functions of $(a_N,p_N)$. Observe that for every $i\in\llbracket 1,N-1\rrbracket$, $\chi_{i}(a_N,0)$ is independent of the value of $a_N>a_{N-1}$. We denote it by $\chi_{i}(\cdot,0)$.

   Define
   \[i_0:= \min \Set{i \in \llbracket 1 , \nbParam \rrbracket}{b_G(i)=N}.\]   If $i_0\leq N-2$, since $G$ is a DC graph containing the edge $(i_0,N)$, we also have $(i_0,N-1)\in E(G')$. Since $(\underline a^\dagger,\underline p^\dagger)\in P_{G'}$, Lemma \ref{lem:Grcarinterpretation} implies that $\chi_{i_0}(\cdot,0)>a_{N-1}$ whenever $i_0\leq N-2$. This inequality also clearly holds when $i_0=N-1$. We also have that $\chi_{i_0}(\cdot,0)>\chi_{i_0-1}(\cdot,0)$. Pick $a_N$ such that
   \[\max(a_{N-1},\chi_{i_0-1}(\cdot,0))<a_N<\chi_{i_0}(\cdot,0).\]
   The value of $a_N$ is now fixed and we now consider $\chi_{i_0}(a_N,p_N)$ and $\chi_{i_0-1}(a_N,p_N)$ as functions of a single variable $p_N\geq0$. It follows from \eqref{eq:defxi} that $\chi_{i_0}(a_N,p_N)$ and $\chi_{i_0-1}(a_N,p_N)$ may be entirely expressed as continuous functions of the $z_i$. By Proposition \ref{proposition:continuityInParameters} and Lemma \ref{lemma:continuityInParameterPN}, we have that $\chi_{i_0}(a_N,p_N)$ and $\chi_{i_0-1}(a_N,p_N)$ are continuous functions of $p_N\geq0$. Thus, the inequalities
    \[\max(a_{N-1},\chi_{i_0-1}(a_N,p_N))<a_N<\chi_{i_0}(a_N,p_N),\]
    which hold for $p_N=0$ also hold for every $p_N\in(0,2\epsilon)$ for some $\epsilon>0$. Fix $p_N=\epsilon$. Up to reducing the value of $\epsilon>0$, by continuity, we may assume that inequalities \eqref{eq:proofExistParam1} and \eqref{eq:proofExistParam2} also hold for all $j<N$ for this choice of $(\underline a,\underline p)$. The $2N$-tuple $(\underline a,\underline p)$ is now completely fixed.
    
    Since $a_N<\chi_{i_0}(a_N,p_N)$, Lemma \ref{lem:Grcarinterpretation} implies that $(i_0,N)$ is an edge of $\Gr(\underline a,\underline p)$. Hence inequality \eqref{eq:proofExistParam1} holds for $(i,j)=(i_0,N)$ and thus also for all $(i,N)$ with $i\in\llbracket i_0,N-1\rrbracket$, which are precisely the pairs of the form $(i,N)$ in $E(G)$.

    Since $a_N>\chi_{i_0-1}(a_N,p_N)$, Lemma \ref{lem:Grcarinterpretation} implies that $(i_0-1,N)$ is not an edge of $\Gr(\underline a,\underline p)$. Hence inequality \eqref{eq:proofExistParam2} holds for $(i,j)=(i_0-1,N)$ and thus also for all $(i,N)$ with $i\in\llbracket 1,i_0-1\rrbracket$, which are precisely the pairs of the form $(i,N)$ which are not in $E(G)$.

    Collecting everything, we have found $(\underline a,\underline p)\in P^{2N}$ such that, in the case when there exists an edge of the form $(i,N)\in E(G)$, all the inequalities \eqref{eq:proofExistParam1} and \eqref{eq:proofExistParam2} hold.

If $i_0=N$, we pick $a_N>\chi_{N-1}(\cdot,0)$ and we conclude using the same line of proof as above.
\end{proof}
   
\section{Adjacency structure of the regions}
\label{sec:adjacency}

In this section we characterize when two regions $P_{G_1}$ and $P_{G_2}$ are adjacent and we compute the dimension of their common boundary.

Let $m(G)$ denote the set of maximal edges of $G$. Denote by $M(G)$ the set of pairs $(i,j)\in E_N\setminus E(G)$ such that $G' = (\llbracket 1 , \nbParam \rrbracket, E(G)\cup \{(i,j)\} )$ is a DC graph.
Equivalently, $M(G)$ is the set of minimal elements of $E_N\setminus E(G)$ for the partial order $\preceq_E$. 

We can now state the main result of this section:

\begin{theorem}[Adjacency between regions]\label{theorem:boundaryAnalysis}
    Let $G_1$ and $G_2$ be two distinct DC graphs in $\DC_N$. Then $\partial P_{G_1} \cap \partial P_{G_2}$ is non-empty if and only if
    \begin{equation} \label{eq:boundaryAnalysisCondition} E(G_1)\backslash E(G_2) \subseteq m(G_1) \quad \text{and} \quad E(G_2) \backslash E(G_1) \subseteq M(G_1). \end{equation}
    In this case, the codimension of $\partial P_{G_1} \cap \partial P_{G_2}$ is $ |E(G_1)\Delta E(G_2)| $. 
\end{theorem}

The following proposition guarantees that Theorem~\ref{theorem:boundaryAnalysis} is equivalent to Theorem~\ref{theorem:mainadjacency}. 
    
\begin{proposition}
\label{prop:equivalentcharac}
    Condition \eqref{eq:boundaryAnalysisCondition} holds if and only if $E(G_1) \Delta E(G_2)$ is an antichain for the poset $(E_N,\preceq_E)$.
\end{proposition}
    
\begin{proof}
    Let us prove that equivalence. First, assume that \eqref{eq:boundaryAnalysisCondition} does not hold. Therefore, either there is an edge $e'$ in $E(G_1)\backslash E(G_2)$ which is not maximal in $G_1$, or there is an edge $e$ in $E(G_1)\backslash E(G_2)$ which is not in $M(G_1)$. In the first case, since $e' = (i',j')$ is in $E(G_1)$ but not in $m(G_1)$, there exists an $(i,j) \in E(G_1)$ such that $(i,j)\neq (i',j')$ and $i\leq i' < j' \leq j$. In addition, since $(i',j')$ is not in $E(G_2)$ and since $G_2$ is a DC graph, $(i,j)$ is not in $G_2$. Therefore, $\{(i,j),(i',j')\}\subseteq E(G_1)\backslash E(G_2)\subseteq E(G_1)\Delta E(G_2)$. In the second case, since $e = (i,j)$ is not in $E(G_1)$ and not in $M(G_1)$, there exists an edge $(i',j') \notin E(G_1)$ such that $(i,j)\neq (i',j')$ and $i \leq i' < j' \leq j$. Since $G_2$ is a DC graph and since $(i,j)$ is in $E(G_2)$, $(i',j')$ is also in $E(G_2)$. Therefore, $\{(i,j),(i',j')\}\subseteq E(G_2)\backslash E(G_1)\subseteq E(G_1)\Delta E(G_2)$. 

    Reciprocally, assume that the rightmost term in the equivalence  does not hold. Then there exists distinct edges $(i,j) \neq (i',j')$ which are both contained in $E(G_1) \Delta E(G_2)$. There are four possibilities: 
    \begin{itemize}
        \item $\{(i,j),(i',j')\} \subseteq E(G_1)\backslash E(G_2)$,
        \item $\{(i,j),(i',j')\} \subseteq E(G_2)\backslash E(G_1)$,
        \item $(i,j) \in E(G_2)\backslash E(G_1)$ and $(i',j') \in E(G_1)\backslash E(G_2)$,
        \item $(i,j) \in E(G_1)\backslash E(G_2)$ and $(i',j') \in E(G_2)\backslash E(G_1)$. 
    \end{itemize}
    The third case is impossible since $G_2$ is a DC graph: if $(i,j)$ is in $E(G_2)$, then $(i',j')$ must be in $E(G_2)$. Similarly, the fourth point is also impossible since $G_1$ is a DC graph. 
    In the first (resp. second) case, it is easy to check that $(i',j')$ is not in $m(G_1)$ (resp. that $(i,j)$ is not in $M(G_1)$) which implies that \eqref{eq:boundaryAnalysisCondition} does not hold.
\end{proof}

In Subsection \ref{subsec:prelimResults}, we will characterize the boundary of a region $P_{G}$ (Corollary \ref{corollary:caracBorder}) and the intersection of the boundaries of two regions $P_{G_1}$ and $P_{G_2}$ (Proposition \ref{proposition:step2and3}). Using this, we will prove Theorem \ref{theorem:boundaryAnalysis} in Subsection \ref{subsec:proofTheoremBoundaryAnalysis}.

\subsection{Characterization of the boundaries}\label{subsec:prelimResults}

Let $G\in\DC_N$ be a DC graph and let $(\underline a,\underline p)\in P^{2N}$. For every $0 \leq i < j \leq \nbParam$, define
\[
    Z_{i,j}^{(G)}(\underline a, \underline p):= z_{i+1}^{(G)}(\underline a, \underline p) + \dots +z_{j}^{(G)}(\underline a, \underline p)
\]
and
\[
    Z_{i,j}(\underline a, \underline p):= z_{i+1}(\underline a, \underline p) + \dots +z_{j}(\underline a, \underline p).
\]
The following notation will be intensively used in the rest of this section.

\begin{definition}
For every subset $S\subset E_N$ and every binary relation $\mathcal R\in\{>,<,\leq,\geq,=\}$, we say $(G,\underline a ,\underline p)$ satisfies the condition $(C_S^{\mathsmaller{\mathcal R}})$ if
\[
\forall (i,j)\in S, \quad z_1^{(G)}(\underline a,\underline p)\ \mathcal R\  Z_{i,j}^{(G)}(\underline a,\underline p).
\]
\end{definition}
For example, $(G,\underline a ,\underline p)$ satisfies the condition $(C_{m(G)}^{\mathsmaller{>}})$ if
\[
\forall (i,j)\in m(G), \quad z_1^{(G)}(\underline a,\underline p)>  Z_{i,j}^{(G)}(\underline a,\underline p).
\]
Note that we do not require in this definition that $(\underline a, \underline p)\in P_G$.

The following lemma will be useful in the remainder of this subsection. 
\begin{lemma}\label{lemma:positivityC2C3}
    Consider a graph $G\in \DC_N$ and parameters $(\underline a , \underline p )\in P^{2\nbParam} $ such that $(G,\underline a ,\underline p)$ satisfies $(C_{m(G)}^{\mathsmaller{\geq}})$ and $(C_{M(G)}^{\mathsmaller{\leq}})$.
    Then for every $ i \in \llbracket 1 , \nbParam \rrbracket$, $z_i^{(G)}(\underline a , \underline p ) = z_i(\underline a , \underline p )$. As a consequence, for such a choice of parameters $(\underline a , \underline p )$, we have $z_i^{(G)}(\underline a , \underline p ) > 0$. 
\end{lemma}

\begin{proof}
    Let us first prove that $z_i^{(G)}(\underline a,\underline p)$ is non-negative for every $i \in \llbracket 1 , \nbParam \rrbracket$. By  \eqref{eq:eq:z1formula}, $z_1^{(G)}(\underline a,\underline p)$ is positive. Recall from Remark \ref{remark:reduxSystem} that $B$ denotes the set of vertices that are either isolated or the starting point of a maximal edge. If $i\in 
    \llbracket 2 , \nbParam \rrbracket \backslash B$, the positivity of $z_i^{(G)}(\underline a,\underline p)$ follows from Remark \ref{remark:reduxSystem}. If $i\in 
    \llbracket 2 , \nbParam \rrbracket$ is an isolated vertex, then $(i-1,i)\in M(G)$. Thus condition $(C_{M(G)}^{\mathsmaller{\leq}})$ implies that
    \[
z_i^{(G)}(\underline a , \underline p)\geq z_1^{(G)}(\underline a , \underline p)>0.
    \]    

    Assume that $i \in \llbracket 2 , \nbParam -1\rrbracket$ is the starting point of a maximal edge. Then $(i, b_G(i)) \in m(G)$ and $(i-1 , b_G(i-1) +1) \in M(G)$, thus $(C_{M(G)}^{\mathsmaller{\leq}})$ and $(C_{m(G)}^{\mathsmaller{\geq}})$ imply that 
    \begin{align*}
        z_1^{(G)}(\underline a,\underline p)&\leq Z_{i-1,b_G(i-1)+1}^{(G)}(\underline a,\underline p) \\
        z_1^{(G)}(\underline a,\underline p)&\geq Z^{(G)}_{i,b_G(i)}(\underline a,\underline p).
    \end{align*}
 
    Therefore, \[ Z^{(G)}_{i,b_G(i)}(\underline a,\underline p) \leq Z^{(G)}_{i-1,b_G(i-1)+1}(\underline a,\underline p).\]
    Since $b_G(i-1) < b_G(i)$, we have
    \begin{equation}
    \label{eq:positivezi}
    Z^{(G)}_{b_G(i-1)+1,b_G(i)}(\underline a,\underline p) \leq z^{(G)}_{i}(\underline a,\underline p),
    \end{equation}
    with the convention that the sum on the left-hand side vanishes if $b_G(i-1)+2>b_G(i)$.
    Since the indices $j$ of the $z^{(G)}_{j}(\underline a,\underline p)$ that appear on the left-hand side of \eqref{eq:positivezi} are greater than $i$, descending induction on $i$ yields that $z^{(G)}_{i}(\underline a,\underline p) \geq 0$ for every $i \in \llbracket 1 , \nbParam \rrbracket$. 

    Now, let us prove that the function 
    \[\hat{y}_{\infty} : t \in \RR_{\geq0} \mapsto \sum_{j= 1}^{\nbParam} p_j \left(t - Z_{0,j}^{(G)}(\underline a,\underline p ) + z_1^{(G)}(\underline a,\underline p ) \right)_+ \] is the  stationary trajectory for parameters $(\underline a,\underline p )$. It is clear that $\hat{y}_{\infty}$ is a continuous increasing map since $p_j$ is positive for all $j \in \llbracket 1 , \nbParam \rrbracket$. Therefore, the inverse function $\hat{t}_{\infty}$ of $\hat{y}_{\infty}$ exists. Following Definition~\ref{def:stationarytrajectory}, it remains to prove that $Z_{1,i}^{(G)}(\underline a , \underline p ) $ is equal to $ \hat{t}_{\infty}(a_i) - \hat{t}_{\infty}(a_1) $ for all $i \in \llbracket 2 , \nbParam \rrbracket$. Let us prove the stronger statement that for all $i \in \llbracket 1 , \nbParam \rrbracket$ \[ Z_{0,i}^{(G)}(\underline a , \underline p ) = \hat{t}_{\infty}(a_i). \] 
    
    By definition of $m(G)$, for all $(i,j) \in E(G) $, there exists an edge $(i',j') \in m(G)$ such that $(i,j)\preceq_E (i',j')$.
    Similarly, for all $(i,j) \in E_N\setminus E(G)$, there exists  $(i',j') \in M(G)$ such that $(i',j')\preceq_E (i,j)$.
    The above non-negativity property implies that 
    \begin{equation}
    \label{eq:equivziedgeBis1}
    \forall (i,j) \in E(G) , \, z_1^{(G)}(\underline a,\underline p) \geq Z^{(G)}_{i,j}(\underline a,\underline p).
    \end{equation}
    \begin{equation}
    \label{eq:equivziedgeBis2}
    \forall (i,j) \in E_N\setminus E(G) , \, z_1^{(G)}(\underline a,\underline p) \leq Z^{(G)}_{i,j}(\underline a,\underline p).
    \end{equation}
    It follows from \eqref{eq:equivziedgeBis1}, \eqref{eq:equivziedgeBis2} and the non-negativity of the $z_i^{(G)}(\underline a , \underline p)$ that for all $i \in \llbracket 1 , \nbParam \rrbracket$, 
    \begin{align*}
        &\hat{y}_{\infty}(   Z_{0,i}^{(G)}(\underline a,\underline p)) \\= &\sum_{j = 1}^{b_G(i)} p_j (Z_{0,i}^{(G)}(\underline a,\underline p) - Z_{0,j}^{(G)}(\underline a,\underline p) + z_1^{(G)}(\underline a,\underline p ) ).
    \end{align*}

    By Theorem \ref{theorem:formulaSpeed}, the $(z_i^{(G)}(\underline a,\underline p))_{i\in\llbracket1,N\rrbracket}$ are solutions of the linear system $\mathcal S^{(G)}(\underline a,\underline p)$. As a consequence, 
    $
    \hat{y}_{\infty}(Z_{0,i}^{(G)}(\underline a,\underline p))  = a_i
    $
    for all $i \in \llbracket 1 , \nbParam \rrbracket$. 
    Equivalently, $\hat{t}_{\infty}(a_i) = Z_{0,i}^{(G)}(\underline a,\underline p)$ for all $i \in \llbracket 1 , \nbParam \rrbracket$. This concludes the proof of the fact that $ \hat{y}_{\infty} $ is the stationary trajectory.

    Since $Z_{0,i}(\underline a,\underline p)$ is defined as the image of $a_i$ by the inverse function of the stationary trajectory, one obtains that 
    \[ Z_{0,i}(\underline a,\underline p)=Z_{0,i}^{(G)}(\underline a,\underline p) \] for all $i \in \llbracket 1 , \nbParam \rrbracket$, which implies that $z_i^{(G)}(\underline a , \underline p ) = z_i(\underline a , \underline p )$ for all $i \in \llbracket 1 , \nbParam \rrbracket$.

    The quantity $z_i(\underline a , \underline p )$ corresponds to the amount of time spent by a car between positions $a_{i-1}$ and $a_i$ in the stationary regime. Since the speeds of the cars remain finite, $ z_i^{(G)}(\underline a , \underline p ) =z_i(\underline a , \underline p) >0 $ for all $i \in \llbracket 1 , \nbParam \rrbracket$.
 \end{proof}

The following proposition gives criteria for parameters to be in the region $P_G$. 

\begin{proposition}[Inequalities characterizing $P_G$]\label{proposition:caracPG}
    Let $G\in\DC_N$ and $(\underline a,\underline p)\in P^{2N}$. Then $(\underline a,\underline p)$ belongs to the region $P_G$ if and only if $(G,\underline a ,\underline p)$ satisfies both conditions $(C_{m(G)}^{\mathsmaller{>}})$ and $(C_{M(G)}^{\mathsmaller{\leq}})$.
\end{proposition}

\begin{proof}
By definition of $m(G)$ and $M(G)$, $m(G) \subseteq E(G)$ and $M(G) \subseteq E_N\setminus E(G)$. Therefore, it follows from the definition of $P_G$, from Theorem \ref{theorem:formulaSpeed} and from  \eqref{eq:edgeGrapz} that, if $(\underline a,\underline p)\in P_G$, then $(G,\underline a ,\underline p)$ satisfies conditions $(C_{m(G)}^{\mathsmaller{>}})$ and $(C_{M(G)}^{\mathsmaller{\leq}})$. 

Conversely, assume that $(G,\underline a ,\underline p)$ satisfies conditions $(C_{m(G)}^{\mathsmaller{>}})$ and $(C_{M(G)}^{\mathsmaller{\leq}})$. 
    Then, $(G,\underline a ,\underline p)$ also satisfies conditions $(C_{m(G)}^{\mathsmaller{\geq}})$, thus by Lemma \ref{lemma:positivityC2C3}, $z_i^{(G)}(\underline a , \underline p ) = z_i(\underline a , \underline p ) >0$ for every $i \in \llbracket 1 , \nbParam \rrbracket$. Hence conditions $(C_{m(G)}^{\mathsmaller{>}})$ and $(C_{M(G)}^{\mathsmaller{\leq}})$ become: 
    \begin{align*} \forall (i,j)\in m(G), \, z_1(\underline a,\underline p)>Z_{i,j}(\underline a,\underline p), \\ \forall (i,j)\in M(G), \, z_1(\underline a,\underline p) \leq Z_{i,j}(\underline a,\underline p). 
    \end{align*}
    Since $z_i(\underline a,\underline p)$ is positive for every $i \in \llbracket 1 , \nbParam \rrbracket$, one obtains that 
    \begin{align*} \forall (i,j)\in E(G), \, z_1(\underline a,\underline p)>Z_{i,j}(\underline a,\underline p), \\ \forall (i,j)\in E_N\setminus E(G), \, z_1(\underline a,\underline p) \leq Z_{i,j}(\underline a,\underline p). 
    \end{align*}
    By definition of $P_G$, $(\underline a,\underline p)$ is in $P_G$. 
\end{proof}

Let us now characterize the interior and the closure of $P_G$.

\begin{proposition}\label{proposition:interiorAndClosure}
    For every $G \in \DC_N$, denote by $\overset{\circ}{P_G}$ (resp. $\overline{P_G}$)  the interior (resp. the closure) of $P_G$ in $P^{2N}$. Then, 
    \[
        \overset{\circ}{P_G} = \Set{(\underline a , \underline p) \in P^{2N}}{(G,\underline a ,\underline p) \text{ satisfies } (C_{m(G)}^{\mathsmaller{>}}) \text{ and } (C_{M(G)}^{\mathsmaller{<}})} \]
    and
    \[    \overline{P_G} = \Set{(\underline a , \underline p) \in P^{2N}}{(G,\underline a ,\underline p) \text{ satisfies } (C_{m(G)}^{\mathsmaller{\geq}}) \text{ and } (C_{M(G)}^{\mathsmaller{\leq}})}. \]
\end{proposition}

In order to prove this proposition, let us introduce some additional notation.

For every DC graph $G\in\DC_N$ and $1 \leq i < j \leq \nbParam$, set
\begin{equation}
\label{eq:defZij}
    \tilde Z_{i,j}^{(G)}(\underline a, \underline p):= z_{1}^{(G)}(\underline a, \underline p) + \frac{Z_{i,j}^{(G)}(\underline a, \underline p)-z_{1}^{(G)}(\underline a, \underline p)}{1 + \sum_{k=i+1}^j \sum_{l = k}^N  \Gamma^{(G)}_{k,l} \frac{q_{b_G(l)}-q_{b_G(l-1)}}{q_{b_G(l-1)}}}.
\end{equation}

\begin{remark}
\label{rem:samesign}
This definition immediately implies that the two quantities $z_{1}^{(G)}(\underline a, \underline p) - Z_{i,j}^{(G)}(\underline a, \underline p) $ and $z_{1}^{(G)}(\underline a, \underline p) - \tilde Z_{i,j}^{(G)}(\underline a, \underline p)$ have the same strict sign (either they are simultaneously positive, or they are simultaneously negative, or they simultaneously vanish).
\end{remark}

\begin{remark}
\label{rem:positivecoeffs}
Let $G \in \DC_{\nbParam}$ be a DC graph and let $1 \leq i < j \leq \nbParam$. 
Using \eqref{ziformula}, we have that
\begin{equation}
\label{eq:Zijformula}
\tilde Z_{i,j}^{(G)}(\underline a, \underline p)= \frac{\sum_{k=i+1}^j \sum_{l=k}^N \Gamma^{(G)}_{k,l} \frac{d_l}{q_{b_G(l-1)}}}{1 + \sum_{k=i+1}^j \sum_{l = k}^N  \Gamma^{(G)}_{k,l} \frac{q_{b_G(l)}-q_{b_G(l-1)}}{q_{b_G(l-1)}}}
\end{equation}

It follows from \eqref{eq:eq:z1formula} (resp. \eqref{eq:Zijformula}) that $z_1^{(G)}(\underline a, \underline p)$ (resp. $\tilde Z_{i,j}^{(G)}(\underline a , \underline p)$) is a linear combination of $d_1,\ldots,d_N$ (resp. $d_{i+1},\ldots,d_N$) with positive coefficients:
\begin{align}
    z_1^{(G)}(\underline a, \underline p) &= \sum_{l=1}^{\nbParam} f_{1,l}^{(G)}(\underline p) d_l \label{eq:Zf1m} \\
    \tilde Z_{i,j}^{(G)}(\underline a , \underline p) &= \sum_{l=i+1}^{\nbParam} f_{i,j,l}^{(G)}(\underline p) d_l, \label{eq:Zfijm}
\end{align}
where the coefficients $f_{1,l}^{(G)}(\underline p)$ and $f_{i,j,l}^{(G)}(\underline p)$ are positive for all $l$ and $\underline p$.
\end{remark}

\begin{proof}[Proof of Proposition \ref{proposition:interiorAndClosure}]
    By Proposition \ref{proposition:caracPG} and by continuity of $z_i^{(G)}$ and $Z_{i,j}^{(G)}$ for all $i<j$, one obtains that 
    \begin{align*}
        \overset{\circ}{P_G} &\supseteq \Set{(\underline a , \underline p) \in P^{2N}}{(G,\underline a ,\underline p) \text{ satisfies } (C_{m(G)}^{\mathsmaller{>}}) \text{ and } (C_{M(G)}^{\mathsmaller{<}})}, \\
        \overline{P_G} &\subseteq \Set{(\underline a , \underline p) \in P^{2N}}{(G,\underline a ,\underline p) \text{ satisfies } (C_{m(G)}^{\mathsmaller{\geq}}) \text{ and } (C_{M(G)}^{\mathsmaller{\leq}})}.
    \end{align*}
    Reciprocally, let us prove that 
    \begin{align*}
        \overset{\circ}{P_G} &\subseteq \Set{(\underline a , \underline p) \in P^{2N}}{(G,\underline a ,\underline p) \text{ satisfies } (C_{m(G)}^{\mathsmaller{>}}) \text{ and } (C_{M(G)}^{\mathsmaller{<}})}.
    \end{align*}
    Consider $(\underline a , \underline p) \in \overset{\circ}{P_G}$. Since $(\underline a , \underline p) \in P_G$, $(G,\underline a ,\underline p)$  satisfies $(C_{m(G)}^{\mathsmaller{>}})$ and $(C_{M(G)}^{\mathsmaller{\leq}})$. Let us prove that $(G,\underline a ,\underline p)$  satisfies $(C_{M(G)}^{\mathsmaller{<}})$. 
    Assume by contradiction that $z_1^{(G)}(\underline a,\underline p) = Z_{i,j}^{(G)}(\underline a,\underline p)$ for some $(i,j) \in M(G)$. By Remark~\ref{rem:samesign} this is equivalent to having $z_1^{(G)}(\underline a,\underline p) = \tilde Z_{i,j}^{(G)}(\underline a,\underline p)$. By Remark~\ref{rem:positivecoeffs}, $\tilde Z_{i,j}^{(G)}(\underline a,\underline p)$ does not depend on $d_1$ while $z_1^{(G)}(\underline a,\underline p)$ is linear in $d_1$ with a positive coefficient in front of $d_1$. Therefore, there exists some $\varepsilon >0$ such that for all $ d_1' \in (d_1 , d_1 + \varepsilon) $, 
    \[z_1^{(G)}(\underline a',\underline p') > \tilde Z_{i,j}^{(G)}(\underline a',\underline p')\]
    and \[ (\underline a',\underline p') \in \overset{\circ}{P_G} ,\] 
    where $d_i' = d_i$ for all $i \in \llbracket 2 , \nbParam \rrbracket$ and $p_i' = p_i$ for all $i \in \llbracket 1, \nbParam \rrbracket$. Since $(\underline a',\underline p')$ is in $P_G$, $(G,\underline a',\underline p')$ should satisfy $(C_{M(G)}^{\mathsmaller{\leq}})$, which is a contradiction with the previous inequality.  

    It remains to prove that 
    \begin{align*}
        \overline{P_G} &\supseteq \Set{(\underline a , \underline p) \in P^{2N}}{(G,\underline a ,\underline p) \text{ satisfies } (C_{m(G)}^{\mathsmaller{\geq}}) \text{ and } (C_{M(G)}^{\mathsmaller{\leq}})}.
    \end{align*}
    Let $(\underline a , \underline p )\in P^{2N}$ be such that  $(C_{m(G)}^{\mathsmaller{\geq}})$ and $(C_{M(G)}^{\mathsmaller{\leq}})$ are satisfied by $(G,\underline a',\underline p')$. Observe that the dynamics is unchanged by re-scaling all the distances by a common multiplicative factor. Without loss of generality, let us choose this scaling such that $z_1^{(G)}(\underline a , \underline p ) =1$.
    Define the following four sets:
    \begin{align*}
A_M^= &:=\Set{(i,j) \in M(G)}{1 = Z_{i,j}^{(G)}(\underline a,\underline p) } \\
A_M^< &:=\Set{(i,j) \in M(G)}{1 < Z_{i,j}^{(G)}(\underline a,\underline p) } \\
A_m^= &:=\Set{(i,j) \in m(G)}{1 = Z_{i,j}^{(G)}(\underline a,\underline p) } \\
A_m^> &:=\Set{(i,j) \in m(G)}{1 > Z_{i,j}^{(G)}(\underline a,\underline p) }.
    \end{align*}

By continuity of the $z_i^{(G)}$ and $Z_{i,j}^{(G)}$ for all $i$ and $j$, there exists $\eta>0$ such that for every $(\underline a' , \underline p' )\in P^{2N}$ satisfying $ \| (\underline a' , \underline p' ) - (\underline a , \underline p ) \|_{\infty} < \eta $, we have 
   \begin{align*}
        \forall (i,j) \in A_M^<, \,& z_1^{(G)}(\underline a',\underline p') < Z_{i,j}^{(G)}(\underline a',\underline p'), \\
        \forall (i,j) \in A_m^>, \,& z_1^{(G)}(\underline a',\underline p') > Z_{i,j}^{(G)}(\underline a',\underline p'). 
    \end{align*}
Let $\varepsilon \in (0,\eta)$. We want to find $(\underline a' , \underline p' )\in P^{2N}$ with $\| (\underline a' , \underline p' ) - (\underline a , \underline p ) \|_{\infty} < \varepsilon$ such that the following three conditions hold:
\begin{itemize}
 \item[$(C_1)$] $z_1^{(G)}(\underline a',\underline p')=1$,
 \item[$(C_2)$] $\forall (i,j) \in A_M^=, \, 1=\tilde Z_{i,j}^{(G)}(\underline a',\underline p')$, 
 \item[$(C_3)$] $\forall (i,j) \in A_m^=, \,  1>\tilde Z_{i,j}^{(G)}(\underline a',\underline p')$.
\end{itemize}
Combined with Remark~\ref{rem:samesign}, it will entail that $(G,\underline a',\underline p')$ satisfies $(C_{M(G)}^{\mathsmaller{\leq}})$ and $(C_{m(G)}^{\mathsmaller{>}})$, and thus by Proposition~\ref{proposition:caracPG} we will conclude that $(\underline a',\underline p')\in P_G$.

Let us first prove that the map $s:(i,j)\in A_M^=\cup A_m^= \mapsto i\in\llbracket 1, \nbParam \rrbracket$ is injective. Assume that $(i,j)$ and $(i,j')$ are two elements of $A_M^=\cup A_m^=$. We distinguish three cases.
\begin{itemize}
\item If $(i,j)$ and $(i,j')$ are both in $A_m^=$, then $j=j'=b_G(i)$.
\item If they are both in $A_M^=$, let us reason by contradiction and assume that $j\neq j'$. Without loss of generality assume $j<j'$. Since $(i,j')$ is a minimal element of $E_N\setminus E(G)$ for the order $\preceq_E$, we have that $(i,j'-1)\in E(G)$. Since $j\leq j'-1$ and $G$ is a DC graph, it implies that $(i,j)\in E(G)$, which is a contradiction.
\item Let us show that $(i,j)$ and $(i,j')$ cannot be for one in $A_m^=$ and for the other in $A_M^=$. Reason by contradiction and assume that $(i,j)\in A_m^=$ and $(i,j')\in A_M^=$. Then $(i,j)\in m(G)$ implies $j=b_G(i)$ and $(i,j')\in M(G)$ implies $j'=b_G(i)+1=j+1$. By definition of $A_m^=$ and $A_M^=$, we have $z_1^{(G)}(\underline a,\underline p) =Z_{i,j}^{(G)}(\underline a , \underline p) = Z_{i,j+1}^{(G)}(\underline a , \underline p) $. This would imply that $ z_{j+1}^{(G)}(\underline a , \underline p) =   Z_{i,j+1}^{(G)}(\underline a , \underline p) - Z_{i,j}^{(G)}(\underline a , \underline p) = 0 $, which contradicts the positivity of $z_{j+1}^{(G)}(\underline a , \underline p )$ when $(C_M^{\mathsmaller{\leq}})$ and $(C_m^{\mathsmaller{\geq}})$ are satisfied by $(G,\underline a,\underline p)$ (Lemma \ref{lemma:positivityC2C3}).
\end{itemize}

 Set $\underline p' := \underline p$ and $ d_{i+1}' := d_{i+1} $ for every $i\in \llbracket 1 , \nbParam -1 \rrbracket \backslash s(A_M^=\cup A_m^=)$.
 Define condition
\[
(C'_3) \quad \forall (i,j) \in A_m^=, \,  1=\tilde Z_{i,j}^{(G)}(\underline a',\underline p').
\]
Then the conditions $(C_1)$, $(C_2)$ and $(C'_3)$ form a linear system of equations in the unknowns $d'_1$ and $(d'_{i+1})_{i\in s(A_M^=\cup A_m^=)}$. Since the map $s$ is injective on $A_M^=\cup A_m^=$, there are as many unknowns as equations. By \eqref{eq:Zf1m} and \eqref{eq:Zfijm}, this system is triangular with non-zero coefficients on the diagonal. It has a unique solution, which we know to be $d'_1=d_1$ and $ d_{i+1}' = d_{i+1} $ for every $i\in  s(A_M^=\cup A_m^=)$. If we now consider conditions $(C_1)$, $(C_2)$ and $(C_3)$, they form a system of linear equations and inequations, defined by hyperplanes that intersect transversely. Thus they possess solutions arbitrarily close to the intersection point.
\end{proof}

Denote by $\partial P_G := \overline{P_G} \backslash \overset{\circ}{P_G}$ the boundary of $P_G$. As an immediate consequence of Proposition \ref{proposition:interiorAndClosure}, one obtains the following characterisation of $\partial P_G$: 
\begin{corollary}\label{corollary:caracBorder}
    We have that $ (\underline a , \underline p)$ is in $\partial P_G$ if and only if $(C_{m(G)}^{\mathsmaller{\geq}})$ and $(C_{M(G)}^{\mathsmaller{\leq}})$ are satisfied by $ (G,\underline a , \underline p)$
and there exists $(i,j)\in m(G) \cup M(G)$ such that 
\[z_1^{(G)}(\underline a,\underline p) = Z_{i,j}^{(G)}(\underline a,\underline p). \]
\end{corollary}

Now, assume that $G_1$ and $G_2$ are two distinct DC graphs and let us characterize the common boundary of the two regions $P_{G_1}$ and $P_{G_2}$. In what follows, we denote by $\setFixedSpeed{1}$ the set of parameters $(\underline a , \underline p )$ such that $z_1(\underline a , \underline p ) = 1$.

\begin{proposition}\label{proposition:step2and3}
    Consider $G_1 \neq G_2$ in $\DC_N$ satisfying \eqref{eq:boundaryAnalysisCondition}. For every $ (\underline a , \underline p ) \in P^{2N}$, $(\underline a , \underline p )$ is in $\partial P_{G_1} \cap \partial P_{G_2} \cap \setFixedSpeed{1}$ if and only if $(G_1,\underline a , \underline p)$ satisfies the four conditions below:
    \begin{align}
    &(C_{E(G_1)\Delta E(G_2)}^{\mathsmaller{=}}) \label{eq:=} \\
    &(C_{m(G_1)\cap E(G_2)}^{\mathsmaller{\geq}}) \label{eq:geq} \\
    &(C_{M(G_1)\cap E(G_2)^c}^{\mathsmaller{\leq}}) \label{eq:leq} \\
    &z_1^{(G_1)}( \underline a , \underline p) = 1. \label{eq:conditionsIntersecLowerBoundary2}
    \end{align}
\end{proposition}

Observe that the three sets $E(G_1)\Delta E(G_2)$, $m(G_1)\cap E(G_2)$ and $M(G_1)\cap E(G_2)^c$ form a partition of $m(G_1)\cup M(G_1)$ by \eqref{eq:boundaryAnalysisCondition}.

\begin{remark}
We could remove the normalization $\{z_1=1\}$ from the statement of Proposition \ref{proposition:step2and3}, in the sense that $( \underline a , \underline p) \in \partial P_{G_1} \cap \partial P_{G_2}$ if and only if conditions \eqref{eq:=}-\eqref{eq:leq} are satisfied. However, as in the proof of Proposition \ref{proposition:interiorAndClosure}, we will later need such a scaling to ensure that a certain linear system is triangular.
\end{remark}

\begin{proof}[Proof of Proposition \ref{proposition:step2and3}]

    First, let us prove that if $(\underline a , \underline p )$ is in $\partial P_{G_1} \cap \partial P_{G_2} \cap \setFixedSpeed{1}$, then \eqref{eq:=}-\eqref{eq:conditionsIntersecLowerBoundary2} are satisfied. 
    Since $(\underline a , \underline p )$ is in $\partial P_{G_1}$, $(C_{m(G_1)}^{\mathsmaller{\geq}})$ and $(C_{M(G_1)}^{\mathsmaller{\leq}})$ are satisfied by $(G_1,\underline a , \underline p )$ by Corollary \ref{corollary:caracBorder}. Therefore, \eqref{eq:geq} and \eqref{eq:leq} are satisfied.
    Note that $(\underline a , \underline p )$ belongs to $\overline{P_{G_1}} \cap \overline{P_{G_2}}$. Then, by Lemma \ref{lemma:positivityC2C3} and Proposition \ref{proposition:interiorAndClosure}, for every $i\in \llbracket 1 , \nbParam \rrbracket$, 
    \begin{equation}\label{eq:zG1G2coincide}
        z_i^{(G_1)}(\underline a , \underline p ) = z_i^{(G_2)}(\underline a , \underline p ) = z_i(\underline a , \underline p ).
    \end{equation}
    Therefore, \eqref{eq:conditionsIntersecLowerBoundary2} is satisfied since $(\underline a , \underline p) \in \{z_1 = 1\}$.
    Now, it remains to prove \eqref{eq:=}. Consider $(i,j) \in E(G_1) \Delta E(G_2)$ and let us prove that
    \begin{equation}
    \label{eq:66equal}
    z_1^{(G_1)}( \underline a , \underline p) -  Z_{i,j}^{(G_1)}( \underline a , \underline p) = 0.
    \end{equation}
    Note that the edges of $E(G_1) \Delta E(G_2)$ are either in $m(G_1) \cap M(G_2)$ or in $M(G_1) \cap m(G_2)$ by hypothesis  \eqref{eq:boundaryAnalysisCondition}. Let us assume that $(i,j)$ belongs to $m(G_1) \cap M(G_2)$ (the other case is treated similarly by symmetry). Then, 
    \begin{equation}
    \label{eq:66geq}
    z_1^{(G_1)}( \underline a , \underline p) -  Z_{i,j}^{(G_1)}( \underline a , \underline p) \geq 0
    \end{equation}
    and 
    \begin{equation}
    \label{eq:66leq}
    z_1^{(G_2)}( \underline a , \underline p) -  Z_{i,j}^{(G_2)}( \underline a , \underline p) \leq 0 .
    \end{equation}
    Equality \eqref{eq:66equal} follows from combining \eqref{eq:zG1G2coincide} with \eqref{eq:66geq} and \eqref{eq:66leq}.

    Now, assume that $(\underline a , \underline p) \in P^{2N}$ satisfies \eqref{eq:=}-\eqref{eq:conditionsIntersecLowerBoundary2}. Let us prove that $(\underline a , \underline p)$ belongs to $\partial P_{G_1 } \cap \partial P_{G_2 } \cap \{z_1= 1\} $.
    Since $G_1\neq G_2$, the set $E(G_1)\Delta E(G_2)$ is non-empty. Combining \eqref{eq:=}-\eqref{eq:leq} with Corollary \ref{corollary:caracBorder}, we get that $( \underline a , \underline p) \in \partial P_{G_1}$. 
    Since $\partial  P_{G_1} \subseteq \overline{P_{G_1}}$ and by Lemma \ref{lemma:positivityC2C3} and Proposition \ref{proposition:interiorAndClosure}, for every $i \in \llbracket 1 , \nbParam \rrbracket$, 
    \[ z_i ( \underline a , \underline p) = z_i^{(G_1)} ( \underline a , \underline p).  \]

    As a consequence, with $i=1$, one obtains that $( \underline a , \underline p) \in \{z_1 = 1\}$. It remains to prove that $( \underline a , \underline p) \in \partial  P_{G_2}$.
    
    By definition of $m(G_1)$, for every $(i,j) \in E(G_1)$, there exists $(i',j') \in m(G_1)$ such that $ i' \leq i < j \leq j'$. Similarly, by definition of $M(G_1)$,  for every $(i,j) \in E(G_1)^c$, there exists $(i',j') \in M(G_1)$ such that $ i \leq i' < j' \leq j$.

    As a consequence, since $z_i^{(G_1)}( \underline a , \underline p)$ is positive for every $i \in \llbracket 1 , \nbParam \rrbracket$ by Lemma~\ref{lemma:positivityC2C3}, and by definition of $m(G)$ and $M(G)$, we have that $(G_1,\underline a,\underline p)$ satisfies conditions $(C_{E(G_1)\cap E(G_2)}^{\mathsmaller{\geq}})$ and $(C_{E(G_1)^c\cap E(G_2)^c}^{\mathsmaller{\leq}})$.
    
    Since $(z_i^{(G_1)}(\underline a , \underline p))_{i \in \llbracket 1 , \nbParam \rrbracket}$ is the unique solution of the system $\mathcal S^{(G_1)}(\underline a,\underline p)$, one obtains that for every $i\in\llbracket1,N\rrbracket$, 
    \begin{align}\label{eq:systemAi}
    a_i=\sum_{j=1}^{b_{G_1}(i)} p_j &\left( z_1^{(G_1)}(\underline a , \underline p) + (z_{1}^{(G_1)}(\underline a , \underline p) + \dots + z_i^{(G_1)}(\underline a , \underline p)) \right. \\&- \left. (z_{1}^{(G_1)}(\underline a , \underline p) + \dots + z_j^{(G_1)}(\underline a , \underline p))\right).\notag
    \end{align}
    
    Let $i\in\llbracket 1,N\rrbracket$. We want to show that
    \begin{align}\label{eq:systemAibis}
    a_i=\sum_{j=1}^{b_{G_2}(i)}p_j &\left( z_1^{(G_1)}(\underline a , \underline p) + (z_{1}^{(G_1)}(\underline a , \underline p) + \dots + z_i^{(G_1)}(\underline a , \underline p)) \right.\\&- \left. (z_{1}^{(G_1)}(\underline a , \underline p) + \dots + z_j^{(G_1)}(\underline a , \underline p))\right). \notag
    \end{align}
    This is clearly true if $b_{G_1}(i)=b_{G_2}(i)$. Otherwise, hypothesis \eqref{eq:boundaryAnalysisCondition} implies that $j:=\max( b_{G_1}(i), b_{G_2}(i))=\min( b_{G_1}(i), b_{G_2}(i)) +1$. In that case, \eqref{eq:systemAi} and \eqref{eq:systemAibis} differ by a single term, which is $z_1^{(G_1)}(\underline a , \underline p) - Z_{i,j}^{(G_1)}(\underline a , \underline p)$. Since $(i,j) \in E(G_1)\Delta E(G_2)$, condition \eqref{eq:=} implies that the extra term vanishes, hence \eqref{eq:systemAibis} also holds in that case. 
    
    As a consequence, $(z_i^{(G_1)}(\underline a , \underline p))_{i \in \llbracket 1 , \nbParam \rrbracket}$ is a  solution of the linear system $\mathcal S^{(G_2)}(\underline a,\underline p)$. Since the unique solution of this system is $(z_i^{(G_2)}(\underline a , \underline p))_{i \in \llbracket 1 , \nbParam \rrbracket}$, we have that $z_i(\underline a , \underline p ) = z_i^{(G_1)}(\underline a , \underline p ) = z_i^{(G_2)}(\underline a , \underline p )$ for every $i \in \llbracket 1 , \nbParam \rrbracket$. Thus $Z_{i,j}^{(G_1)}(\underline a , \underline p) = Z_{i,j}^{(G_2)}(\underline a , \underline p)$ for every $1 \leq i < j \leq \nbParam$. The fact that $(G_1,\underline a,\underline p)$ satisfies conditions $(C_{E(G_1)\Delta E(G_2)}^{\mathsmaller{=}})$, $(C_{E(G_1)\cap E(G_2)}^{\mathsmaller{\geq}})$ and $(C_{E(G_1)^c\cap E(G_2)^c}^{\mathsmaller{\leq}})$ then imply that $(G_2,\underline a,\underline p)$ satisfies respectively conditions $(C_{E(G_1)\Delta E(G_2)}^{\mathsmaller{=}})$, $(C_{m(G_2)\cap E(G_1)}^{\mathsmaller{\geq}})$ and $(C_{M(G_2)\cap E(G_1)^c}^{\mathsmaller{\leq}})$.
    We deduce from Corollary \ref{corollary:caracBorder} that $(\underline a , \underline p)\in\partial P_{G_2}$. To conclude, $(\underline a , \underline p) \in \partial P_{G_1} \cap \partial P_{G_2}\cap \{z_1 = 1\}$.
\end{proof}

\subsection{Proof of Theorem \ref{theorem:boundaryAnalysis}}\label{subsec:proofTheoremBoundaryAnalysis}
We decompose the proof of Theorem \ref{theorem:boundaryAnalysis} into three lemmas: Lemmas \ref{lemma:step1}, \ref{lemma:step2} and \ref{lemma:step3}.
For the rest of the subsection, let us fix $G_1$ and $G_2$ two distinct DC graphs. 
    
\begin{lemma}\label{lemma:step1}
    $\partial P_{G_1} \cap \partial P_{G_2} = \varnothing$ if \eqref{eq:boundaryAnalysisCondition} is not verified.
\end{lemma}

\begin{lemma}\label{lemma:step2}
    $\partial P_{G_1} \cap \partial P_{G_2} \neq \varnothing$ if \eqref{eq:boundaryAnalysisCondition} is verified. 
\end{lemma}

\begin{lemma}\label{lemma:step3}
    If $\partial P_{G_1} \cap \partial P_{G_2} \neq \varnothing$, then  \[ \dim \partial P_{G_1} \cap \partial P_{G_2} = 2 \nbParam - |E(G_1) \Delta E(G_2)|.\] 
\end{lemma}

\subsubsection{Proof of Lemma \ref{lemma:step1}}

    Assume that \eqref{eq:boundaryAnalysisCondition} is not verified. More explicitly, assume that $ (E(G_1) - E(G_2))\cap(m(G_1))^c \neq \varnothing$ or  $ (E(G_2) - E(G_1))\cap(M(G_1))^c \neq \varnothing$ and let us show that $\partial P_{G_1} \cap \partial P_{G_2}$ is empty. 
    
    In the first case, there is an edge $(i,j)$ of $G_1$ which is not in $G_2$ and which is not maximal in $G_1$. Therefore, there exists $(i',j') \in m(G_1)$ a maximal edge of $G_1$ distinct of $(i,j)$ such that $ i' \leq i < j \leq j'$. Since $G_2$ is a DC graph and $(i,j)$ is not in $G_2$, $(i',j')$ is not in $G_2$. 
    
    Now, we know that $(i,j)$ and $(i',j')$ are both in $G_1$ and both not in $G_2$. By contradiction, assume there exists $(\underline a,\underline p) \in \partial P_{G_1} \cap \partial P_{G_2} $. 
    Then, by Proposition \ref{proposition:interiorAndClosure}, for $(i_0,j_0) \in \{(i,j), (i',j')\}$, 
        we have
            $z_1^{(G_1)}(\underline a,\underline p)  \geq Z_{i_0,j_0}^{(G_1)}(\underline a,\underline p)$
        and 
        $z_1^{(G_2)}(\underline a,\underline p)  \leq Z_{i_0,j_0}^{(G_2)}(\underline a,\underline p)$.
    Let $l\in\llbracket 1,N\rrbracket$. For every $(\underline a' , \underline p') \in P_{G_1}$, $z_l(\underline a' , \underline p')=z_l^{(G_1)}(\underline a' , \underline p')$ and for every $(\underline a' , \underline p') \in P_{G_2}$, $z_l(\underline a' , \underline p')=z_l^{(G_2)}(\underline a' , \underline p')$.
    Therefore, by continuity of $z_l$, $z_l^{(G_1)}$ and $z_l^{(G_2)}$, since $(\underline a , \underline p) \in \overline{P_{G_1}}\cap \overline{P_{G_2}}$, $z_l(\underline a , \underline p) = z_l^{(G_1)}(\underline a,\underline p) = z_l^{(G_2)}(\underline a,\underline p)$. 
    As a consequence,
    \[Z_{i',j'}(\underline a,\underline p)= Z_{i,j}(\underline a,\underline p)  = z_1(\underline a,\underline p) .\]

    Since $(i',j') \neq (i,j)$, there are strictly more terms in $Z_{i',j'}(\underline a,\underline p)$ than in $Z_{i',j'}(\underline a,\underline p)$. The fact that each $z_{l}(\underline a,\underline p)$ is positive yields the desired contradiction.
    
    In the second case, there exists an edge $(i,j)$ of $G_2$ which is not in $M(G_1)$. Therefore, there exists an edge $(i',j')$ in $M(G_1)$ distinct of $(i,j)$ such that $ i \leq  i' < j' \leq j $. Then, by Proposition \ref{proposition:interiorAndClosure}, for $(i_0,j_0) \in \{(i,j), (i',j')\}$, 
        we have $z_1^{(G_1)}(\underline a,\underline p)  \leq Z_{i_0,j_0}^{(G_1)}(\underline a,\underline p)$
        and 
        $z_1^{(G_2)}(\underline a,\underline p)  \geq Z_{i_0,j_0}^{(G_2)}(\underline a,\underline p)$.
        We reach a contradiction as in the first case. 
    
\subsubsection{Proof of Lemma \ref{lemma:step2}}

    Assume that \eqref{eq:boundaryAnalysisCondition} is satisfied. It suffices to find $(\underline a , \underline p)\in \setFixedSpeed{1} $ such that $(G_1,\underline a , \underline p)$ satisfies conditions $(C_{E(G_1)\Delta E(G_2)}^{\mathsmaller{=}})$, $(C_{m(G_1)\cap E(G_2)}^{\mathsmaller{>}})$ and $(C_{M(G_1)\cap E(G_2)^c}^{\mathsmaller{<}})$. Indeed, by Proposition \ref{proposition:step2and3}, it will imply that $(\underline a , \underline p) \in \partial P_{G_1} \cap \partial P_{G_2}$.
    
Let us proceed by induction on $N$, along similar lines as the proof of Theorem \ref{theorem:nonemptiness}.

For $N=1$, there is nothing to prove since there is only one DC graph. For $N=2$, assume that $G_1 = (\{1,2\}, \varnothing)$ and $G_2 = (\{1,2\}, \{(1,2)\})$. In this case, $E(G_1)\Delta E(G_2) =\{(1,2)\}$ does not contain any pair of nested edges. Set $\underline a = (1 , 3)$ and $\underline p = (1 , 1)$. Then one computes $z_1^{(G_1)}(\underline a,\underline p)-Z_{1,2}^{(G_1)}(\underline a,\underline p)=z_1^{(G_1)}(\underline a,\underline p)-z_2^{(G_1)}(\underline a,\underline p)=0$ and $z_1^{(G_1)}(\underline a,\underline p)=1$. Thus $(G_1,\underline a , \underline p)$ satisfies $(C_{E(G_1)\Delta E(G_2)}^{\mathsmaller{=}})$, $(C_{m(G_1)\cap E(G_2)}^{\mathsmaller{>}})$, $(C_{M(G_1)\cap E(G_2)^c}^{\mathsmaller{<}})$ and \eqref{eq:conditionsIntersecLowerBoundary2}. 

Now fix $N\geq3$ and assume that the result holds for every pair of distinct DC graphs with $N-1$ vertices.

Consider $G_1, G_2$ two distinct DC graphs with $N$ vertices satisfying \eqref{eq:boundaryAnalysisCondition}. Define $G'_1$ (resp. $G'_2$) to be the restriction of $G_1$ (resp. $G_2$) to the vertex set $\llbracket 1 , \nbParam - 1 \rrbracket$. Then $G'_1$ and $G'_2$ satisfy \eqref{eq:boundaryAnalysisCondition}. By induction hypothesis, there exists $(\underline a',\underline p')\in P^{2N-2}$ such that $(G'_1,\underline a' , \underline p')$ satisfies conditions $(C_{E(G'_1)\Delta E(G'_2)}^{\mathsmaller{=}})$, $(C_{m(G'_1)\cap E(G'_2)}^{\mathsmaller{>}})$,  $(C_{M(G'_1)\cap E(G'_2)^c}^{\mathsmaller{<}})$ and \eqref{eq:conditionsIntersecLowerBoundary2}. Define 
\[(a^{(0)}_1, \dots , a^{(0)}_{\nbParam - 1}, p^{(0)}_1,\ldots,p^{(0)}_{\nbParam - 1 }):=(\underline a',\underline p').\]

Let $l\in\{1,2\}$. For every $1\leq i\leq N$, by definition of $z^{(G_l)}_i$, this function can be extended by continuity to the case when $p_N=0$ by setting
    \begin{align}
    z^{(G_l)}_i(\underline a^\dagger,a_N,\underline p^\dagger,0)&:=z^{(G'_l)}_i(\underline a^\dagger,\underline p^\dagger) \quad \text{ if } i\in\llbracket 1 , \nbParam -1 \rrbracket; \label{eq:zigpnzero} \\
    z^{(G_l)}_N(\underline a^\dagger,a_N,\underline p^\dagger,0)&:=\frac{d_N}{q_{N-1}}. \label{eq:zngpnzero}
    \end{align}
With this extension, we have that $(G_1,\underline a^{(0)} , \underline p^{(0)})$ satisfies the four conditions $(C_{E(G_1)\Delta E(G_2)}^{\mathsmaller{=}})$, $(C_{m(G_1)\cap E(G_2)}^{\mathsmaller{>}})$,  $(C_{M(G_1)\cap E(G_2)^c}^{\mathsmaller{<}})$ and \eqref{eq:conditionsIntersecLowerBoundary2} whenever $j\neq N$ and $p^{(0)}_N=0$, regardless of the choice of $a^{(0)}_N>a^{(0)}_{N-1}$. Let us fix $p^{(0)}_N=0$ and let $a^{(0)}_N>a^{(0)}_{N-1}$ be a free parameter that will be chosen appropriately below. 

For every $1\leq i\leq N$, we denote by $\chi_i=\widetilde y_\infty(\widetilde t_\infty (a^{(0)}_1)+\widetilde t_\infty (a^{(0)}_i))$ the position of the car of index $-1$ in a stationary configuration where the car of index $0$ is at position $a^{(0)}_i$. We extend this definition to the case $i=0$ by setting $\chi_0:=a^{(0)}_1$. Consider the quantities $\chi_{i}$ as functions of $(a^{(0)}_N,p^{(0)}_N)$. Observe that for every $i\in\llbracket 1,N-1\rrbracket$, $\chi_{i}(a^{(0)}_N,0)$ is independent of the value of $a^{(0)}_N>a^{(0)}_{N-1}$. We denote it by $\chi_{i}(\cdot,0)$.

 Let $(z_i(\underline a,\underline p))_{i \in \llbracket 1 , \nbParam \rrbracket}$ (resp. $z'_i(\underline a',\underline p')_{i \in \llbracket 1 , \nbParam - 1 \rrbracket}$) be the functions defined by \eqref{eq:definitionOfZi} for $(\underline a,\underline p)\in P^{2N}$ (resp. $(\underline a',\underline p')\in P^{2N-2}$).  By using the definition of the continuous extensions \eqref{eq:zipnzero}, \eqref{eq:znpnzero}, \eqref{eq:zigpnzero} and \eqref{eq:zngpnzero} of $z_i$ and $z_i^{(G_1)}$ to the case when $p_N=0$ and combining Lemma \ref{lemma:positivityC2C3} with the induction hypothesis, we obtain that for every $i\in\llbracket 1,N-1\rrbracket$,
   \[
   z_i(\underline a^{(0)} , \underline p^{(0)})=z_i(\underline a^{(0)} , \underline p',0)=z_i'(\underline a' , \underline p') = z_i^{(G_1')}(\underline a' , \underline p')= z_i^{(G_1)}(\underline a^{(0)} , \underline p',0)=z_i^{(G_1)}(\underline a^{(0)} , \underline p^{(0)}).
   \]
   We also have that
   \[
   z_N(\underline a^{(0)} , \underline p^{(0)})=\frac{a^{(0)}_N-a^{(0)}_{N-1}}{q^{(0)}_{N-1}}= z_N^{(G_1)}(\underline a^{(0)} , \underline p^{(0)}).
   \]

Set $i_l := \min \Set{i \in \llbracket 1 , \nbParam \rrbracket}{b_{G_l}(i)=\nbParam  } $ for $l=1,2$. By \eqref{eq:boundaryAnalysisCondition}, $ i _1 - i_2 \in \{-1,0,1\}$. By symmetry, assume without loss of generality that $i_1 \leq i_2$. We shall now prove that $(G_1,\underline a^{(0)} , \underline p^{(0)})$ satisfies conditions $(C_{E(G_1)\Delta E(G_2)}^{\mathsmaller{=}})$, $(C_{m(G_1)\cap E(G_2)}^{\mathsmaller{>}})$ and   $(C_{M(G_1)\cap E(G_2)^c}^{\mathsmaller{<}})$ when $j=N$ and $p^{(0)}_N=0$. We consider three cases.

\emph{Case 1: $i_1=i_2\leq N-1$.} The only edges of the form $(i,N)$ in $m(G_1)\cup M(G_1)$

are $(i_1,N)\in m(G_1)\cap E(G_2)$ and $(i_1-1,N)\in M(G_1) \cap E(G_2)^c$ thus it suffices to find $a^{(0)}_N>a^{(0)}_{N-1}$ satisfying the following inequalities:
\begin{equation}
\label{eq:ineqscase1}
   Z^{(G_1)}_{i_1 , N}( \underline a^{(0)} , \underline p',0) < z_1^{(G_1)}( \underline a^{(0)} , \underline p',0) < Z^{(G_1)}_{i_1 -1, N}( \underline a^{(0)} , \underline p',0).
\end{equation}

If $i_1\leq N-2$, since $(i_1,N)$ is an edge of both DC graphs $G_1$ and $G_2$, $(i_1,N-1)$ is also an edge of $G_1$ and $G_2$, hence also of $G'_1$ and $G'_2$.
Since $(i_1,N-1)\in m(G'_1)\cap E(G'_2)$ and $(G'_1,\underline a' , \underline p')$ satisfies $(C_{m(G'_1)\cap E(G'_2)}^{\mathsmaller{>}})$ by induction hypothesis, we have that, if $i_1\leq N-2$, 
\[
z_1^{(G_1)}(\underline a^{(0)},\underline p',0)>Z_{i_1,N-1}^{(G_1)}(\underline a^{(0)},\underline p',0),
\]
which can be rewritten as
\[
z_1(\underline a^{(0)},\underline p',0)>Z_{i_1,N-1}(\underline a^{(0)},\underline p',0).
\]
Lemma \ref{lem:Grcarinterpretation} implies that $a^{(0)}_{N-1} < \chi_{i_1}(\cdot , 0)$ if $i_1\leq N-2$. This inequality also clearly holds for $i_1=N-1$. We also have that $\chi_{i_1}(\cdot,0)>\chi_{i_1-1}(\cdot,0)$. Pick $a^{(0)}_N$ such that
   \[\max(a^{(0)}_{N-1},\chi_{i_1-1}(\cdot,0))<a^{(0)}_N<\chi_{i_1}(\cdot,0).\]
  
   Since $\chi_{i_1-1}(\cdot,0) <  \, a^{(0)}_N < \chi_{i_1}(\cdot,0)$, it follows from Lemma \ref{lem:Grcarinterpretation} that  
   \[
   Z_{i_1 , N}( \underline a^{(0)} , \underline p',0) < z_1( \underline a^{(0)} , \underline p',0) < Z_{i_1 -1, N}( \underline a^{(0)} , \underline p',0),
   \]
   which yields \eqref{eq:ineqscase1}.

\emph{Case 2: $i_1=i_2=N$.} The only edge of the form $(i,N)$ in $m(G_1)\cup M(G_1)$ is $(N-1,N)\in M(G_1) \cap E(G_2)$. We pick $a^{(0)}_N>\chi_{N-1}(\cdot,0)$ and conclude as above that
\[
    z_1^{(G_1)}( \underline a^{(0)} , \underline p',0) < Z_{N -1, N}^{(G_1)}( \underline a^{(0)} , \underline p',0).
   \]

\emph{Case 3: $i_2=i_1+1$.} The proof follows the same lines as the one of case 1, replacing a strict inequality by an equality. The only edges of the form $(i,N)$ in $m(G_1)\cup M(G_1)$ are $(i_1,N)\in E(G_1)\Delta E(G_2)$ and $(i_1-1,N)\in M(G_1) \cap E(G_2)$, thus it suffices to find $a^{(0)}_N>a^{(0)}_{N-1}$ satisfying
\begin{equation}
\label{eq:ineqscase3}
   Z^{(G_1)}_{i_1 , N}( \underline a^{(0)} , \underline p',0) = z_1^{(G_1)}( \underline a^{(0)} , \underline p',0) < Z^{(G_1)}_{i_1 -1, N}( \underline a^{(0)} , \underline p',0).
\end{equation}
Assume first that $i_1\leq N-2$. Since $(i_1,N)$ is in $E(G_1)\Delta E(G_2)$, which contains no pair of nested edges by assumption \eqref{eq:boundaryAnalysisCondition}, we have that $(i_1,N-1)\notin E(G_1)\Delta E(G_2)$. Since $(i_1,N)\in E(G_1)$ with $i_1\leq N-2$, we have that $(i_1,N-1)\in E(G_1)$, hence $(i_1,N-1)\in E(G_2)$ too.
As in case 1, we deduce that $a^{(0)}_{N-1}<\chi_{i_1}(\cdot,0)$. The latter inequality automatically holds when $i_1=N-1$. We then set $a^{(0)}_N:=\chi_{i_1}(\cdot,0)>\max(a^{(0)}_{N-1},\chi_{i_1-1}(\cdot,0))$, which yields
\eqref{eq:ineqscase3}.

Wrapping up the three cases, we have found in each case a value of $a^{(0)}_N$ such that the triple $(G_1,\underline a^{(0)} , \underline p^{(0)})$ satisfies conditions $(C_{E(G_1)\Delta E(G_2)}^{\mathsmaller{=}})$, $(C_{m(G_1)\cap E(G_2)}^{\mathsmaller{>}})$,  $(C_{M(G_1)\cap E(G_2)^c}^{\mathsmaller{<}})$ and \eqref{eq:conditionsIntersecLowerBoundary2} for $p^{(0)}_N=0$. The requirement that $(C_{m(G_1)\cap E(G_2)}^{\mathsmaller{>}})$ and $(C_{M(G_1)\cap E(G_2)^c}^{\mathsmaller{<}})$ be satisfied by $(G_1,\underline a , \underline p)$ forms a collection of strict inequalities involving continuous functions of $(\underline a,\underline p)$ all the way up to $p_N=0$, thus this holds true for $(\underline a,\underline p)$ in a neighborhood $U$ of $(\underline a^0,\underline p^0)$ in the closure $\overline{P^{2N}}$ of $P^{2N}$.

We fix $p_i=p_i^{(0)}$ for every $i\in\llbracket 1,N-1\rrbracket$ and we let $p_N\geq0$ be a free parameter for now. Let us consider a linear system of equations with unknowns $\underline d$, where $d_i=a_i-a_{i-1}$ for every $i\in\llbracket 1,N\rrbracket$.

One equation is given by \eqref{eq:conditionsIntersecLowerBoundary2}. 
The equations
\begin{equation}
\label{eq:equivalentequality}
\forall (i,j) \in E(G_1) \Delta E(G_2), \, z_1^{(G_1)}( \underline a , \underline p) -  \tilde{Z}_{(i,j)}^{(G_1)}(\underline a , \underline p) = 0,
\end{equation}
are equivalent, by Remark~\ref{rem:samesign}, to the fact that $(G_1,\underline a , \underline p)$ satisfies $(C_{E(G_1)\Delta E(G_2)}^{\mathsmaller{=}})$. Combining them with \eqref{eq:conditionsIntersecLowerBoundary2}, we get the following equations:
\begin{equation}
\label{eq:reinjectedequivalentequality}
\forall (i,j) \in E(G_1) \Delta E(G_2), \    \tilde{Z}_{(i,j)}^{(G_1)}(\underline a , \underline p) = 1,
\end{equation}
Finally we also require that
\begin{equation}
\label{eq:samevalues}
\forall i\in\llbracket 1,N-1\rrbracket \setminus s(E(G_1) \Delta E(G_2)), \ d_{i+1}=a^{(0)}_{i+1}-a^{(0)}_{i}.
\end{equation}
The map $s:(i,j)\in E(G_1) \Delta E(G_2)\mapsto i$ is injective by assumption \eqref{eq:boundaryAnalysisCondition}.
Thus the linear system consisting of \eqref{eq:conditionsIntersecLowerBoundary2}, \eqref{eq:reinjectedequivalentequality} and \eqref{eq:reinjectedequivalentequality} has $N$ equations and $N$ unknowns $\underline d$, and it is triangular with non-zero diagonal elements, as in the proof of Proposition \ref{proposition:interiorAndClosure}. Hence it has a unique solution for every choice of $p_N\geq0$ and this solution is a continuous function of $p_N$. When $p_N=0$, this solution is given by $d_i=a^{(0)}_{i}-a^{(0)}_{i-1}$ for every $i\in\llbracket 1, N\rrbracket$. Thus by continuity, we can find $(\underline a ,\underline p)\in U\cap P^{2N}$ such that $(G_1,\underline a , \underline p)$ satisfies $(C_{E(G_1)\Delta E(G_2)}^{\mathsmaller{=}})$ and \eqref{eq:conditionsIntersecLowerBoundary2}. 

Putting everything together, we have found $(\underline a ,\underline p)\in P^{2N}$ such that $(G_1,\underline a , \underline p)$ satisfies conditions $(C_{E(G_1)\Delta E(G_2)}^{\mathsmaller{=}})$, $(C_{m(G_1)\cap E(G_2)}^{\mathsmaller{>}})$, $(C_{M(G_1)\cap E(G_2)^c}^{\mathsmaller{<}})$ and \eqref{eq:conditionsIntersecLowerBoundary2}, which concludes the proof of the inductive step.

\subsubsection{Proof of Lemma \ref{lemma:step3}}

Recall that $\setFixedSpeed{1}$ denotes the set of parameters $(\underline a , \underline p)$ in $P^{2\nbParam}$ such that $z_1(\underline a , \underline p) = 1$. The map 
\begin{alignat*}{3}
h \colon &\RR_{>0} \times \setFixedSpeed{1}&\rightarrow {}& P^{2 \nbParam} &\\
           &(\lambda , (\underline a , \underline p))&\mapsto     {}& (\lambda \underline a , \underline p)&
\end{alignat*}
is a homeomorphism with inverse function
\begin{alignat*}{3}
h^{-1} \colon &P^{2 \nbParam}&\rightarrow {}& \RR_{>0} \times \setFixedSpeed{1} \\
           &(\underline a , \underline p)&\mapsto     {}&  \left(z_1(\underline a , \underline p), \left(\frac{\underline a}{z_1(\underline a , \underline p)},\underline p\right)\right).
\end{alignat*}

The fact that $h^{-1}$ is well-defined and is the inverse function of $h$ is a consequence of the following re-scaling property: for every $\lambda \in \RR_{>0}$, $(\underline a , \underline p) \in P^{2N}$ and $i \in \llbracket 1,\nbParam \rrbracket$, 
\[ z_i(\lambda\underline a, \underline p ) = \lambda z_i(\underline a, \underline p )\]
Moreover, $h$ and $h^{-1}$ are both continuous since $z_1$ is continuous by Proposition \ref{proposition:continuityInParameters} and positive.

    By Proposition \ref{proposition:step2and3}, $( \underline a , \underline p)$ is in $ (\partial P_{G_1} \cap \partial P_{G_2})\cap \setFixedSpeed{1}  $  if and only if \eqref{eq:=}-\eqref{eq:conditionsIntersecLowerBoundary2} are satisfied. Set $\delta:=|E(G_1) \Delta E(G_2)|$. As was the case in the end of the proof of Lemma \ref{lemma:step2}, \eqref{eq:conditionsIntersecLowerBoundary2} and \eqref{eq:=} can be replaced by $\delta+1$ equalities expressing $d_1$ and $(d_{i+1})_{i\in s(E(G_1) \Delta E(G_2))}$ in terms of the $2N- 1 - \delta$ free parameters $(d_i)_{i \in \llbracket 1 , \nbParam -1 \rrbracket \backslash s(E(G_1) \Delta E(G_2))}$ and $\underline p$. Thus
    \[ \dim  ((\partial P_{G_1} \cap \partial P_{G_2})\cap \setFixedSpeed{1}) \leq  2N - 1 - \delta. \]

    The dimension of $ (\partial P_{G_1} \cap \partial P_{G_2})\cap \setFixedSpeed{1}  $ is bounded below by the dimension of the set of parameters $ ( \underline a , \underline p) \in P^{2N} $ satisfying the following two conditions:
    \begin{itemize}
    \item $(G_1,\underline a,\underline p)$ satisfies $(C_{m(G_1)\cap E(G_2)}^{\mathsmaller{>}})$ and $(C_{M(G_1)\cap E(G_2)^c}^{\mathsmaller{<}})$
    \item $(\underline a,\underline p)$ satisfies the $\delta+1$ equalities expressing $d_1$ and $(d_{i+1})_{i\in s(E(G_1) \Delta E(G_2))}$ in terms of $(d_i)_{i \in \llbracket 1 , \nbParam -1 \rrbracket \backslash s(E(G_1) \Delta E(G_2))}$ and $\underline p$.
    \end{itemize}

    This set is non-empty by Lemma \ref{lemma:step2}.
    Moreover, by continuity of the $z_i^{(G_1)}$ and $Z_{i,j}^{(G_1)}$, the strict inequalities remain satisfied on an open neighborhood of the free parameters $(d_i)_{i \in \llbracket 1 , \nbParam -1 \rrbracket \backslash s(E(G_1) \Delta E(G_2))}$ and $\underline p$ around any point $( \underline a , \underline p)$ satisfying the above conditions. Therefore, this set contains an open subset of $ \RR^{N - 1 - \delta }$. As a consequence, the topological dimension of $(\partial P_{G_1} \cap \partial P_{G_2})\cap \setFixedSpeed{1}$ is equal to $N - 1 - \delta$.

    The conclusion follows from the fact that the homeomorphism $h$ maps $\partial P_{G_1} \cap \partial P_{G_2} $ to $\RR_{>0} \times (\partial P_{G_1} \cap \partial P_{G_2})\cap \setFixedSpeed{1}$. 

\subsection{Examples of boundary equations}

In this subsection, we derive from   Proposition \ref{proposition:caracPG} examples of formulas for boundary equations.

\begin{example}
\label{ex:completeboundary}
Continuing Example \ref{ex:completespeed}, when $G= K_{\nbParam}$, $m(K_N)$ consists of the single edge $(1,N)$ and $M(K_N)$ is the empty set. Thus the speed of the front is given by \eqref{eq:completespeed} as soon as $(\underline a,\underline p)$ satisfies $q_{\nbParam} d_1 > \sum_{j = 2}^{\nbParam}q_{j-1}d_j$.
\end{example}

\begin{example}
\label{ex:lineboundary}
Continuing Example \ref{ex:linespeed} for the line graph $L_N$, the speed of the front is given by \eqref{eq:linespeed} as soon as the following $2\nbParam - 3$ equations are satisfied: 
\begin{align*}
\forall i \in \llbracket 1, \nbParam - 1 \rrbracket, \ \frac{\sum_{j=1}^{\nbParam} d_j \left( \prod_{h=1}^j \frac{p_h}{q_h} \right)}{\sum_{j=0}^{\nbParam-1} p_{j+1} \left( \prod_{h=1}^j \frac{p_h}{q_h} \right)  }
&> 
\frac{\sum_{j=i+1}^{\nbParam} d_j \left( \prod_{h=i+1}^j \frac{p_h}{q_h} \right)}{\sum_{j=i}^{\nbParam-1} p_{j+1} \left( \prod_{h=i+1}^j \frac{p_h}{q_h} \right) }\\
\forall i \in \llbracket 1, \nbParam - 2 \rrbracket, \ \frac{\sum_{j=1}^{\nbParam} d_j \left( \prod_{h=1}^j \frac{p_h}{q_h} \right)}{\sum_{j=0}^{\nbParam-1} p_{j+1} \left( \prod_{h=1}^j \frac{p_h}{q_h} \right)  } 
&\leq 
\frac{d_{i+1} + \sum_{j=i+2}^{\nbParam} d_j  \left( \prod_{h=i+3}^j \frac{p_h}{q_h} \right)}{q_{i+2}+\sum_{j=i+2}^{\nbParam-1} p_{j+1}  \left( \prod_{h=i+3}^j \frac{p_h}{q_h} \right)}.
\end{align*}
\end{example}

\begin{example}
\label{ex:DC3}
In Figure \ref{fig:Hasse3}, we depict the Hasse diagram of the Stanley lattice, representing all the DC graphs with $3$ vertices. By Theorem \ref{theorem:boundaryAnalysis}, the boundaries of codimension $1$ between two regions correspond to the edges of this Hasse diagram. Next to each edge of the graph, we represent the numerator of the rational function $z^{(G)}_1 - Z_{i,j}^{(G)}$, where $G$ is the graph located at one end of the edge (it does not matter which one). 
By Proposition \ref{proposition:caracPG}, $P_G$ consists of all the parameters such that for every $G'$ adjacent to $G$ in the Hasse diagram, $z^{(G')}_1 - F_{(i,j)}^{(G')}$ is positive (resp. non-positive) if and only if the edge from $G$ and $G'$ is ascending (resp. descending) in the Hasse diagram. For instance, $d_1 q_3 - d_2 q_1 - d_3 q_2 = 0$ is the equation of the boundary between the two regions indexed by the complete graph $K_3$ and the line graph $L_3$. If $d_1 q_3 - d_2 q_1 - d_3 q_2 > 0$, the DC graph of the stationary solution is the complete graph $K_3$.
\end{example}

\section*{Acknowledgements}
We are grateful to Ksenia Chernysh, Bastien Mallein, Arvind Singh and Damien Thomine for many very fruitful discussions at various stages of this project. We also thank Régine Marchand and Marie Théret for comments on an early version of this manuscript. Finally we thank the referee for making several insightful comments which contributed to improve the exposition. Part of this research was performed while the first author was visiting the Institute for Pure and Applied Mathematics (IPAM), which is supported by the National Science Foundation (Grant No. DMS-1925919). 

\bibliographystyle{alpha}
\bibliography{refs}
\Addresses

\end{document}